\documentclass[preprint]{elsarticle}

\usepackage{amsfonts}
\usepackage{amsmath,amsthm,amsfonts}
\usepackage{amssymb}
\usepackage{multirow}
\usepackage{xspace}

\usepackage[letterpaper, top=1in, bottom=1in, left=1in, right=1in, includefoot]{geometry}

\RequirePackage[T1]{fontenc}
\RequirePackage[utf8]{inputenc}
\usepackage[english]{babel}
\usepackage{alphabeta}
\RequirePackage{stmaryrd}
\usepackage{textcomp}

 \usepackage[pdftex, plainpages = false, pdfpagelabels,
                 bookmarks=true,
                 bookmarksopen = true,
                 bookmarksnumbered = true,
                 pdfauthor=author,
	        pdfstartview =,
                 breaklinks = true,
                 linktocpage,
                 colorlinks = true,
                 linkcolor = blue,
                 hyperfigures
                 ]{hyperref}

\biboptions{sort&compress}

\newtheorem{proposition}{Proposition}[section]
\newtheorem{lemma}{Lemma}[section]
\newtheorem{observation}{Observation}[section]
\newtheorem{corollary}{Corollary}[section]
\newtheorem{claim}{Claim}[section]
\newenvironment{cproof}{\proof[Proof of claim]}{\endproof}

\addto\extrasenglish{}
\addto\extrasenglish{}
\addto\extrasenglish{}
\addto\extrasenglish{}
\addto\extrasenglish{}
\addto\extrasenglish{}
\addto\extrasenglish{}

\DeclareUnicodeCharacter{2286}{\subseteq}
\DeclareUnicodeCharacter{2192}{\ifmmode\to\else\textrightarrow\fi}
\DeclareUnicodeCharacter{2203}{\ensuremath\exists}
\DeclareUnicodeCharacter{183}{\cdot}
\DeclareUnicodeCharacter{2200}{\forall}
\DeclareUnicodeCharacter{2264}{\leq}
\DeclareUnicodeCharacter{2265}{\geq}
\DeclareUnicodeCharacter{8614}{\mathbin{\mapsto}}
\DeclareUnicodeCharacter{8656}{\Leftarrow}
\DeclareUnicodeCharacter{8657}{\Uparrow}
\DeclareUnicodeCharacter{8658}{\Rightarrow}
\DeclareUnicodeCharacter{8659}{\Downarrow}
\DeclareUnicodeCharacter{8669}{\rightsquigarrow}
\newcommand{\eqdef}{\stackrel{{\scriptsize\rm def}}{=}}
\DeclareUnicodeCharacter{8797}{\eqdef}
\DeclareUnicodeCharacter{8870}{\vdash}
\DeclareUnicodeCharacter{8873}{\Vdash}
\DeclareUnicodeCharacter{22A7}{\models}
\DeclareUnicodeCharacter{9121}{\lceil}
\DeclareUnicodeCharacter{9123}{\lfloor}
\DeclareUnicodeCharacter{9124}{\rceil}
\DeclareUnicodeCharacter{2208}{\in}
\DeclareUnicodeCharacter{9126}{\rfloor}
\DeclareUnicodeCharacter{9655}{\triangleright}
\DeclareUnicodeCharacter{9665}{\triangleleft}
\DeclareUnicodeCharacter{9671}{\diamond}
\DeclareUnicodeCharacter{9675}{\circ}
\DeclareUnicodeCharacter{10178}{\bot}
\DeclareUnicodeCharacter{10214}{} 
\DeclareUnicodeCharacter{10215}{} 
\DeclareUnicodeCharacter{10229}{\longleftarrow}
\DeclareUnicodeCharacter{10230}{\longrightarrow}
\DeclareUnicodeCharacter{10231}{\longleftrightarrow}
\DeclareUnicodeCharacter{10232}{\Longleftarrow}
\DeclareUnicodeCharacter{10233}{\Longrightarrow}
\DeclareUnicodeCharacter{10234}{\Longleftrightarrow}
\DeclareUnicodeCharacter{10236}{\longmapsto}
\DeclareUnicodeCharacter{10238}{\Longmapsto} 
\DeclareUnicodeCharacter{10503}{\Mapsto}    
\DeclareUnicodeCharacter{10971}{\mathrel{\not\hspace{-0.2em}\cap}}
\DeclareUnicodeCharacter{65294}{\ldotp}
\DeclareUnicodeCharacter{65372}{\mid}

\newcommand{\bigmid}{{\;\big|\;}}
\newcommand{\lipItem}[1]{#1}

\newcommand\annotated{annotated\xspace}
\newcommand\Annotated{Annotated\xspace}
\newcommand\annotatedextension{annotated extension\xspace}

\newcommand{\adh}{{\sf adh}\xspace}

\newcommand{\bd}{{\sf bd}\xspace}
\newcommand{\btw}{{\sf btw}\xspace}
\newcommand{\cc}{{\sf cc}\xspace}
\newcommand{\ch}{{\sf ch}\xspace}

\newcommand{\cut}{{\sf cut}\xspace}

\newcommand{\FPT}{{\sf FPT}\xspace}
\newcommand{\heir}{{\sf heir}\xspace}

\newcommand{\NP}{{\sf NP}}

\newcommand{\oct}{{\sf oct}\xspace}
\newcommand{\opt}{{\sf opt}\xspace}

\newcommand{\p}{{\sf p}\xspace}
\newcommand{\ph}{\hat{\sf p}_{g,\opt}\xspace}
\newcommand{\phf}{\hat{\sf p}_{f,\opt}\xspace}

\newcommand{\PbPack}{\text{\sc $\Gcal$-Packing}\xspace}
\newcommand{\PbCov}{{\sc Vertex Deletion to $\Gcal$}\xspace}
\newcommand{\Par}{{\sf par}\xspace}
\newcommand{\torso}{{\sf torso}\xspace}
\newcommand{\tw}{{\sf tw}\xspace}

\newcommand{\XP}{{\sf XP}\xspace}
\newcommand{\yes}{{\sf yes}\xspace}

\newcommand{\bN}{{\mathbb{N}}}

\newcommand{\Acal}{\mathcal{A}}
\newcommand{\Bcal}{\mathcal{B}}
\newcommand{\Ccal}{\mathcal{C}}

\newcommand{\Fcal}{\mathcal{F}}
\newcommand{\Gcal}{\mathcal{G}}
\newcommand{\Hcal}{\mathcal{H}}

\newcommand{\Lcal}{\mathcal{L}}
\newcommand{\Mcal}{\mathcal{M}}

\newcommand{\Ocal}{\mathcal{O}}
\newcommand{\Pcal}{\mathcal{P}}
\newcommand{\Qcal}{\mathcal{Q}}

\newcommand{\Tcal}{\mathcal{T}}

\newcommand{\Xcal}{\mathcal{X}}
\newcommand{\Ycal}{\mathcal{Y}}

\title{Dynamic programming on bipartite tree decompositions\tnoteref{t1}}
\tnotetext[t1]{A conference version of this work will appear in the \emph{Proceedings of the 18th International Symposium on Parameterized and Exact Computation (\textbf{IPEC}), Amsterdam, the Netherlands, September \textbf{2023}, volume 285 of LIPIcs, pages 16:1--16:22}.}

\author[1]{Lars Jaffke\fnref{fn1}}
\ead{lars.jaffke@uib.no}
\author[2]{Laure Morelle\corref{cor1}\fnref{fn2,fn3}}
\ead{laure.morelle@lirmm.fr}
\author[2]{Ignasi Sau\fnref{fn2,fn3}}
\ead{ignasi.sau@lirmm.fr}
\author[2]{Dimitrios M. Thilikos\fnref{fn3}}
\ead{sedthilk@thilikos.info}

\cortext[cor1]{Corresponding author}
\fntext[fn1]{Supported by the Research Council of Norway (No 274526).}
\fntext[fn2]{Supported by the ANR project ELIT (ANR-20-CE48-0008-01).}
\fntext[fn3]{Supported by the French-German Collaboration ANR/DFG Project UTMA (ANR-20-CE92-0027), the ANR
project GODASse ANR-24-CE48-4377, and by the Franco-Norwegian project PHC AURORA 2024-2025 (Projet n°
51260WL).}

\affiliation[1]{organization={Department of Informatics, University of Bergen},
addressline={Postboks 7803},
postcode={5020},
city={Bergen},
country={Norway}}

\affiliation[2]{organization={LIRMM, University of Montpellier, CNRS},
addressline={161 rue Ada},
postcode={34095},
city={Montpellier},
country={France}}

\journal{Journal of Computer and System Sciences}

\usepackage[tickmarkheight=0.1cm]{todonotes}

\begin{document}

\begin{frontmatter}

\begin{abstract}
\noindent
We revisit a graph width parameter that we dub \emph{bipartite treewidth} (\textsf{btw}). Bipartite treewidth can be seen as a common generalization of treewidth and the odd cycle transversal number, and is closely related to odd-minors. Intuitively, a \emph{bipartite tree decomposition} is a tree decomposition whose bags induce almost bipartite graphs and whose adhesions contain at most one ``bipartite'' vertex, while the width of such decomposition measures the number of ``non-bipartite'' vertices in a bag. We provide \textsf{para-}\NP-completeness results and develop dynamic programming techniques to solve problems on graphs of small \textsf{btw}. In particular, we show that \textsc{$K_t$-Subgraph-Cover}, \textsc{Weighted Independent Set}, \textsc{Odd Cycle Transversal}, and \textsc{Maximum Weighted Cut} are $\textsf{FPT}$ parameterized by \textsf{btw}. We also provide the following dichotomy when $H$ is a 2-connected graph: if $H$ is bipartite, then \textsc{$H$-\{Subgraph/Induced-Subgraph/Odd-Minor/Scattered\}-Packing} is \textsf{para-NP}-complete parameterized by \textsf{btw} while, if $H$ is non-bipartite, then the problem is solvable in \textsf{XP}-time.
\end{abstract}


\end{frontmatter}

\section{Introduction}
\label{sec-intro}

A graph $H$ is said to be an \emph{odd-minor} of a graph $G$ if it can be obtained from $G$ by iteratively removing vertices, edges, and contracting edge cuts.
Hadwiger's conjecture \cite{Hadwiger43uber}, which is open since 1943, states that if a graph excludes $K_t$ as a minor, then its chromatic number is at most $t-1$.
In 1993, Gerards and Seymour \cite{JensenT11grap} generalized this conjecture to odd-minors, hence drawing attention to odd-minors:
the Odd Hadwiger's conjecture states that if a graph excludes $K_t$ as an odd-minor, then its chromatic number is at most $t-1$.
Since then, a number of papers regarding odd-minors appeared. Most of them focused to the resolution of the Odd Hadwiger's conjecture (see for instance \cite{GeelenGRSV09onth}, and \cite{Steiner22impr} for a nice overview of the results), while some others aimed at extending the results of graph minor theory to odd-minors (see for instance \cite{DemaineHK10deco, Huynh09thel, KawarabayashiRW11}). In particular, Demaine, Hajiaghayi, and Kawarabayashi~\cite{DemaineHK10deco} provided a structure theorem which essentially states that graphs excluding an odd-minor can be obtained by clique-sums of almost-embeddable graphs and almost bipartite graphs.
To prove this, they implicitly  proved the following, which is described more explicitly by Tazari \cite{Tazari12fast} (see \autoref{sec_tw} for the corresponding definitions).

\begin{proposition}[\cite{Tazari12fast}, adapted from \cite{DemaineHK10deco}]\label{taz}
Let $H$ be a fixed graph and let $G$ be a given $H$-odd-minor-free graph. There exists a fixed graph $H'$, $\kappa,\mu\in\bN$ depending only on $H$, and an explicit uniform algorithm that computes a rooted tree decomposition of $G$ such that:
\begin{itemize}
\item the adhesion\footnote{The \emph{adhesion} of two nodes of a tree decomposition $\Tcal$ is the intersection of the corresponding bags of $\Tcal$.} of two nodes has size at most $\kappa$, and
\item the torso\footnote{The \emph{torso} of a bag $B$ in a tree decomposition of a graph $G$ is the graph obtained from $G[B]$ by making a clique out of $B\cap B'$ for each other bag $B'$.} of each bag $B$ either consists of a bipartite graph $W_B$ together with $\mu$ additional vertices (bags of Type~1) or is $H'$-minor-free (bags of Type~2).
\end{itemize}

Furthermore, the following properties hold:
\begin{enumerate}
\item Bags of Type~2 appear only in the leaves of the tree decomposition,
\item if $B_2$ is a bag that is a child of a bag $B_1$ in the tree decomposition, then $|B_2\cap V(W_{B_1})|\leq 1$; and if $B_2$ is  of Type~1, then $|B_1\cap V(W_{B_2})|\leq 1$ as well,
\item the algorithm runs in time $\Ocal_H(|V(G)|^4)$, and
\item the $\mu$ additional vertices of the bags of Type~1, called \emph{apex} vertices,  can be computed within the same running time.
\end{enumerate}
\end{proposition}

It is worth mentioning that Condition \lipItem{2} of \autoref{taz} is slightly stronger than what is stated in \cite{Tazari12fast}, but it follows from the proof of \cite[Theorem 4.1]{DemaineHK10deco}.

\begin{figure}[h]
\center
\includegraphics[scale=0.7]{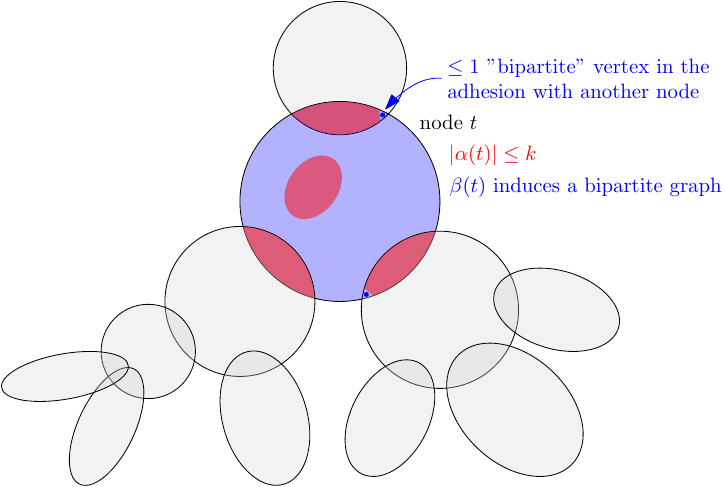}
\caption{Illustration of a bipartite tree decomposition. The blue part induces a bipartite graph, while the additional vertices are depicted in red.
This bipartite tree decomposition has width at most $k$ if there are at most $k$ additional vertices in each bag.}
\label{fig_bipartite_decomposition}
\end{figure}
If we can solve a problem on graphs admitting a tree decomposition as the one described in~\autoref{taz}, we should thus be able to solve this problem on odd-minor-closed graph classes.
Minor-closed graph classes are already quite well-understood.
Therefore, as a first step toward building a theory for solving  problems on odd-minor-closed graph classes, we study in this paper a new type of tree decomposition, which we call \emph{bipartite tree decomposition}, corresponding to the tree decompositions of \autoref{taz}, but where all bags are only of Type~1 (see \autoref{fig_bipartite_decomposition}). We also stress that this decomposition has also been implicitly used in \cite{KawarabayashiR10an(a} and is also introduced, under the same name, in \cite{Campbell23odd-}.
The width of such tree decompositions is defined as the maximum number of apex vertices in a bag of the decomposition. Naturally, the \emph{bipartite treewidth} of a graph $G$, denoted by $\btw(G)$, is the minimum width over all bipartite tree decompositions of $G$; a formal definition is given in \autoref{sec_tw}. It follows easily from the definition that $\btw(G)=0$ if and only if $G$ is bipartite (indeed, to prove the sufficiency, just take a single bag containing the whole bipartite graph, with no apex vertices). More generally, for every graph $G$ it holds that $\btw(G) \leq \oct(G)$, where $\oct$ denotes the size of a minimum \textsl{odd cycle transversal}, that is, a vertex set intersecting every odd cycle. On the other hand, since a bipartite tree decomposition is a tree decomposition whose width is not larger than the maximum size of a bag (in each bag, just declare all vertices as apices), for every graph $G$ it holds that $\btw(G) \leq \tw(G) + 1 $, where $\tw$ denotes treewidth. Thus, bipartite treewidth can be seen as a common generalization of treewidth and the odd cycle transversal number.  Hence, an \FPT-algorithm parameterized by $\btw$ should generalize {\sl both} \FPT-algorithms parameterized by $\tw$ and by $\oct$. Since our goal is to develop a theory for solving problems related to odd-minors, the first prerequisite is that bipartite treewidth is closed under odd-minors. Fortunately, this is indeed the case (cf. \autoref{lem:closed-under-om}). Interestingly, this would {\sl not} be true anymore if, in Condition \lipItem{2} of \autoref{taz}, the considered intersections were required to be upper-bounded by some integer larger than one (cf. \autoref{lem:not-closed-under-om}).

This type of tree decomposition has already been used implicitly in order to solve {\sc Odd Cycle Transversal} parameterized by the solution size by Kawarabayashi and Reed \cite{KawarabayashiR10an(a}.
Independently of our work, Campbell, Gollin, Hendrey, and Wiederrecht \cite{Campbell23odd-} are also currently  studying bipartite tree decompositions. In particular, they provide universal obstructions
characterizing  bounded
\btw\ in the form of a ``grid theorem” (actually the result of \cite{Campbell23odd-} apply in the much more general setting of undirected group labeled graphs). They also designed an \FPT-approximation algorithm that can construct a bipartite tree decomposition
in time $g(k)\cdot n^4\log n$ (cf. \autoref{prop_seb}).
This  \FPT-approximation is an important prerequisite for our
algorithmic results as it
permits us to assume that, for the implementation of  our algorithms, some (approximate) bipartite tree decomposition is
provided in advance.

Our aim is to provide a general framework for the design of  dynamic programming algorithms on bipartite tree decompositions and, more generally, on a broader type of decompositions that we call \emph{1-$\Hcal$-tree decompositions}. These decompositions generalize bipartite tree decompositions, in the sense that the role of bipartite graphs is replaced by a general graph class $\Hcal$.

\subparagraph{Our results.}
In this article we formally introduce bipartite treewidth and bipartite tree decompositions (noticing that they were implicitly already used before, as discussed above). We then focus on the complexity of various problems when the bipartite treewidth of the input graph is taken as a parameter. In particular, we show the following (cf. \autoref{table-results}):

\begin{itemize}
\item While a graph with \btw\ at most $k$ is $(k+2)$-colorable (\autoref{approx-col}), {\sc 3-Coloring} is \NP-complete even on graphs of \oct of size three (\autoref{lem-hard-coloring}), and thus \btw\ at most three.

\item {\sc $K_t$-Subgraph-Cover}, {\sc Weighted Vertex Cover/Independent Set}, {\sc Odd Cycle Transversal}, and {\sc Maximum Weighted Cut} are \FPT parameterized by \btw (cf. \autoref{sec-results-FPT}). In particular, our \FPT-algorithms extend the domain where these well-studied problems can be solved in polynomial time to graphs that are ``locally close to being bipartite''.
Furthermore, as $\btw(G)\leq\oct(G)$ for any graph $G$, we think that the fact that \textsc{Odd Cycle Transversal} is \FPT\ parameterized by \btw is  relevant by itself, as it generalizes the well-known \FPT-algorithms parameterized by the solution size~\cite{ReedSV04find,LokshtanovNRRS14}.

\item Let $H$ be a  2-connected graph.
For each of the {\sc $H$-(Induced-)Subgraph-Packing}, {\sc $H$-Scattered-Packing}, and \!{\sc $H$-Odd-Minor-Packing} problems (cf. \autoref{xp} for the definitions), we obtain the following complexity dichotomy: if $H$ is bipartite, then the problem is {\sf para-NP}-complete parameterized by \btw (in fact, even for $\btw=0$), and if $H$ is non-bipartite, then the problem is solvable in \XP-time. The hardness results are presented in \autoref{sec-packing-hardness} and the \XP-algorithms in \autoref{xp}.

\item In view of the definition of bipartite tree decompositions, it seems natural to consider any graph class $\Hcal$, instead of bipartite graphs, as the ``free part'' of the bags.
This leads to the more general definition of \emph{1-$\Hcal$-tree decomposition} and \emph{1-$\Hcal$-treewidth} (cf. \autoref{sec_tw}), with
1-$\{\emptyset\}$-treewidth being equivalent to the usual treewidth and
1-$\Bcal$-treewidth being the bipartite treewidth if $\Bcal$ is the class of bipartite graphs.
We introduce these more general definitions because our dynamic programming technologies easily extend to 1-$\Hcal$-treewidth.
It also seems natural to generalize by allowing any number $q$ of ``bipartite vertices'' in the adhesion, instead of just one. It corresponds, in \autoref{fig_bipartite_decomposition}, to replacing the ``$\le 1$'' with ``$\le q$''.
For $q=0$, this corresponds to the $\Hcal$-treewidth defined in \cite{EibenGHK21meas} (see also \cite{JansenK21fpta,AgrawalKLPRSZ22Deleting} on the  study of  $\Hcal$-treewidth for several instantiations of $\Hcal$).
However, as mentioned above, while 1-$\Bcal$-treewidth is closed under odd-minors (\autoref{lem:closed-under-om}), this is not the case anymore for $q\geq 2$ (\autoref{lem:not-closed-under-om}).
For $q \geq 2$, some problems remain intractable even when $H$ is not bipartite.
As an illustration of this phenomenon, we prove that {\sc $H$-Scattered-Packing} (where there cannot be an edge in $G$ among the copies of $H$ to be packed)
is {\sf para-NP}-complete parameterized by $q$-$\Bcal$-treewidth for $q\geq 2$ even if $H$ is {\sl not} bipartite (\autoref{lem-paraNP-non-bip}).
\end{itemize}

In the statements of the running time of our algorithms, we always let $n$ (resp. $m$) be the number of vertices (resp. edges) of the input graph of the considered problem.

\begin{table}[htb]
\centering
\scalebox{1}{\begin{tabular}{ |c|c|c|c| }
\hline
Problem & Complexity &  Constraints on $H$/Running time\\
\hline
{\sc $H$-Minor-Cover}~\cite{Yannakakis81node} &  & $H$ containing $P_3$ as a subgraph\\
{\sc $H$(-Induced)-Subgraph/} & \multirow{4}{*}{\textsf{para-}\NP-complete,} & \multirow{4}{*}{$H$ bipartite containing $P_3$ as a subgraph} \\
{\sc Odd-Minor-Cover}~\cite{Yannakakis81node}  and& \multirow{5}{*}{$k= 0$} &  \\
{\sc $H$(-Induced)-Subgraph/} & &\\
{\sc $H$(-Odd)-Minor-Packing} & &\\
{\sc $H$-Scattered-Packing} & & $H$ 2-connected bipartite with $|V(H)|\geq 2$ \\
\hline
\multirow{2}{*}{\sc 3-Coloring} & \textsf{para-}\NP-complete, & \\
& $k= 3$ & \\
\hline
{\sc $K_t$-Subgraph-Cover} & \multirow{4}{*}{\FPT} & $\Ocal(2^k\cdot (k^t\cdot(n+m)+m^{1+o(1)}))$\\
{\sc Weighted Independent Set} &  & $\Ocal(2^k\cdot (k\cdot (k+n)+m^{1+o(1)}))$\\
{\sc Odd Cycle Transversal} &  & $\Ocal(3^{k}\cdot (m+k^2\cdot n)^{1+o(1)})$\\
{\sc Maximum Weighted Cut} &  & $\Ocal(2^k\cdot (k\cdot (k+n)+n^{\Ocal(1)}))$\\
\hline
{\sc $H$-Subgraph-Packing} & \multirow{4}{*}{\XP} & \multirow{2}{*}{$H$ non-bipartite 2-connected}\\
{\sc $H$-Induced-Subgraph-Packing} & & \\
{\sc $H$-Scattered-Packing} & &  \multirow{2}{*}{$n^{\Ocal(k)}$}\\
{\sc $H$-Odd-Minor-Packing} & & \\
\hline
\end{tabular}}
\caption{Summary of the results obtained in this article, where $k$ is the bipartite treewidth of the input graph.}\label{table-results}
\end{table}

\subparagraph{Related results.}
Other types of tree decompositions integrating some ``free bipartite parts'' have been defined recently.
As we already mentioned, Eiben, Ganian,  Hamm, and  Kwon~\cite{EibenGHK21meas} defined \emph{$\Hcal$-treewidth} for a fixed graph class $\Hcal$.
The $\Hcal$-treewidth of a graph $G$ is essentially the minimum treewidth of the graph induced by some set $X\subseteq V(G)$ such that the connected components of $G\setminus X$ belong to $\Hcal$,  and is equal to $0$-$\Hcal$-treewidth minus one (cf. \autoref{sec_tw}).
In particular, when $\Hcal$ is the class of bipartite graphs $\Bcal$, Jansen and de Kroon \cite{JansenK21fpta} provided an \FPT-algorithm to test whether the $\Bcal$-treewidth of a graph is at most $k$.

Recently, as a first step to provide a systematic theory for odd-minors,
Gollin and Wiederrecht \cite{GollinW23odd-}
defined the \emph{$\Hcal$-blind-treewidth} of a graph $G$,
where $\Hcal$ is a property of annotated graphs\footnote{An \emph{annotated graph} is a pair $(G,X)$ where $G$ is a graph and $X\subseteq V(G)$.}.
Then the $\Hcal$-blind-treewidth is the smallest $k$ such that $G$ has a tree decomposition where every bag $B$ such that $(G,B)\notin\Hcal$ has size at most $k$.
For the case where  $\Ccal$ consists of every $(G,X)$
where every odd cycle in $H$ has \textsl{at most} one vertex in $X$,
we obtain the $\Ccal$-blind-treewidth, for which \cite{GollinW23odd-} gives
an  analogue of the Grid Exclusion Theorem~\cite{ChuzhoyT21,RobertsonST94} under the odd-minor relation.
Moreover, \cite{GollinW23odd-} provides an \FPT-algorithm for {\sc Independent Set} parameterized by $\Ccal$-blind-treewidth.
The authors of \cite{GollinW23odd-} prove that the \emph{bipartite-blind treewidth} of a graph $G$ is lower-bounded by a function of
the maximum treewidth over all non-bipartite blocks of $G$.
This immediately implies that bipartite-blind treewidth is lower-bounded by bipartite treewidth.
Hence, our \FPT-algorithm for {\sc Independent Set} is more general than the one in \cite{GollinW23odd-}. Independently of our work,~\cite{Campbell23odd-} presents an \FPT-algorithm to solve {\sc Odd Cycle Transversal} parameterized by \btw in time $f(\btw)\cdot n^4\log n$ (in fact, they solve a more general group labeled problem). Our algorithm for \textsc{Odd Cycle Transversal} (cf. \autoref{co-algo-OCT}), which runs in time $\Ocal(3^\btw\cdot\btw^3\cdot n\cdot m)$, is considerably faster.

\subparagraph{Organization of the paper.}
In \autoref{sec-overview} we provide an overview of our techniques.
In \autoref{def} we give basic definitions on graphs and we define $q$-$\Hcal$-treewidth.
In \autoref{fpt} we define a \emph{nice reduction}, give a general dynamic programming algorithm to obtain \FPT-algorithms, and apply it to {\sc $K_t$-Subgraph-Cover}, {\sc Weighted Independent Set/Vertex Cover}, {\sc Odd Cycle Transversal}, and {\sc Maximum Weighted Cut}.
In \autoref{xp} we provide another dynamic programming scheme to solve {\sc $H$-Subgraph-Packing} and {\sc $H$-Odd-Minor-Packing} in \XP-time. In \autoref{bip} we provide our hardness results.
Finally, we present several questions for further research in \autoref{sec-conclusions}.

\section{Overview of our techniques}
\label{sec-overview}

In this section we present an overview of the techniques that we use to obtain our results.

\subsection{Dynamic programming algorithms}
\label{sec-overview-dp}

Compared to dynamic programming on classical tree decompositions, there are two main difficulties for doing dynamic programming on (rooted) {\sl bipartite} tree decompositions. The first one is that the bags in a bipartite tree decomposition may be arbitrarily large, which prevents us from applying typical brute-force approaches to define table entries. The second one, and apparently more important,  is the lack of an upper bound on the number of children of each node of the decomposition. Indeed, unfortunately, a notion of “nice bipartite tree decomposition” preserving the width (even approximately) does not exist (cf. \autoref{no-nice}). We discuss separately the main challenges involved in our \FPT-algorithms and in our \XP-algorithms.

\subsubsection{\FPT-algorithms}
\label{sec-overview-FPT}

In fact, in order to obtain \FPT-algorithms parameterized by \btw, for most of the considered problems, it would be enough to bound the number of children as a function of \btw, but we were not able to come up with a general technique that achieves this property (cf. \autoref{no-nice}). For particular problems, however, we can devise ad-hoc solutions. Namely, for {\sc $K_t$-Subgraph-Cover}, {\sc Weighted Vertex Cover/Independent Set}, {\sc Odd Cycle Transversal}, and {\sc Maximum Weighted Cut} parameterized by \btw, we overcome the above issue by managing to replace the children with constant-sized bipartite gadgets.
More specifically, we guess an annotation of the ``apex'' vertices of each bag $t$, whose number is bounded by \btw, that essentially tells which of these vertices go to the solution or not (with some extra information depending on each particular problem; for instance, for \textsc{Odd Cycle Transversal}, we also guess the side of the bipartition of the non-solution vertices).
Having this annotation, each adhesion of the considered node $t$ with a child contains, by the definition of bipartite tree decompositions, at most one vertex $v$ that is not annotated.
At this point, we crucially  observe that, for the considered  problems, we can make a local computation for each child, independent from the computations at other children, depending only on the values of the optimum solutions at that child that are required to {\sl contain} or to {\sl exclude} $v$ (note that we need to be able to keep this extra information at the tables of the children). Using the information given by these local computations, we can replace the children of $t$ by constant-sized bipartite gadgets (sometimes empty) so that the newly built graph, which we call a \emph{nice reduction}, is an equivalent instance modulo some constant.
If a nice reduction can be efficiently computed for a problem $\Pi$, then we say that $\Pi$ is a \emph{nice problem} (cf. \autoref{sec:nice}).
The newly modified bag has bounded \oct, so we can then use an \FPT-algorithm parameterized by \oct\ to find the optimal solution with respect to the guessed annotation.

\subparagraph{An illustrative example.} Before entering into some more technical details and general definitions, let us illustrate this idea with the \textsc{Weighted Vertex Cover} problem. The formal definition of bipartite tree decomposition can be found in \autoref{sec_tw} (in fact, for the more general setting of 1-$\Hcal$-treewidth).
For this informal explanation, it is enough to suppose that we want to compute the dynamic programming tables at a bag associated with a node $t$ of the rooted tree decomposition, and that the vertices of the bag at $t$ are partitioned into two sets: $\beta(t)$ induces a bipartite graph and its complement, denoted by $\alpha(t)$, corresponds to the apex vertices, whose size is bounded by the parameter $\btw$. The first step is to guess, in time at most $2^{\btw}$, which vertices in $\alpha(t)$ belong to the desired minimum vertex cover. After such a guess, all the vertices in $\alpha(t)$ can be removed from the graph, by also removing the neighborhood of those that were {\sl not} taken into the solution. The definition of bipartite tree decomposition implies that, in each adhesion with a child of the current bag, there is at most one ``surviving'' vertex. Let $v$ be such a vertex belonging to the adhesion with a child $t'$ of $t$. Suppose that, inductively, we have computed in the tables for $t'$ the following two values, subject to the choice that we made for $\alpha(t)$: the minimum weight $w_v$ of a vertex cover in the graph below $t'$ that contains $v$, and the minimum weight $w_{\bar{v}}$ of a vertex cover in the graph below $t'$ that does {\sl not} contain $v$. Then, the trick is to observe that, having these two values at hand, we can totally forget the graph below $t'$: it is enough to delete this whole graph, except for $v$, and attach a new pendant edge $vu$, where $u$ is a new vertex, such that $v$ is given weight $w_v$ and $u$ is given weight $w_{\bar{v}}$. It is easy to verify that this gadget mimics, with respect to the current bag, the behavior of including vertex $v$ or not in the solution for the child $t'$. Adding this gadget for every child results in a bipartite graph, for which we can just solve \textsc{Weighted Vertex Cover} in polynomial time using a classical algorithm~\cite{ChenKLPGS25maxi}, and add the returned weight to our tables. The running time of this whole procedure, from the leaves to the root of the decomposition, is clearly \FPT parameterized by the bipartite treewidth of the input graph.

\subparagraph{Extensions and limitations.} The algorithm sketched above is problem-dependent, in particular the choice of the gadgets for the children, and the deletion of the neighborhood of the vertices chosen in the solution. Which type of replacements and reductions can be afforded in order to obtain an \FPT-algorithm for bipartite treewidth? For instance, concerning the gadgets for the children, as far as the considered problem can be solved in polynomial time on bipartite graphs, we could attach to the ``surviving'' vertices an arbitrary bipartite graph instead of just an edge. If we assume that the considered problem is \FPT parameterized by \oct (which is a reasonable assumption, as \btw generalizes \oct), then one could think that it may be sufficient to devise gadgets with bounded \oct. Unfortunately, this will not work in general. Indeed, even if each of the gadgets has bounded \oct (take, for instance, a triangle), since we do not have any upper bound, in terms of \btw, on the number of children (hence, the number of different adhesions), the resulting graph after the gadget replacement may have unbounded \oct. In order to formalize the type of replacements and reductions that can be allowed, we introduce in \autoref{fpt} the notions of \emph{nice reduction} and \emph{nice problem}.
Additional insights into these definitions, which are quite lengthy, are provided in \autoref{sec:nice}.

Another sensitive issue is that of ``guessing the vertices into the solution''. While this is quite simple for \textsc{Weighted Vertex Cover} (either a vertex is in the solution, or it is not), for some other problems we may have to guess a  richer structure in order to have enough information to combine the tables of the children into the tables of the current bag. This is the reason why, in the general dynamic programming scheme that we present in \autoref{fpt}, we deal with \emph{annotated problems}, i.e., problems that receive as input, apart from a graph,
a collection of annotated sets in the form of a partition
$\mathcal{X}$ of some $X\subseteq V(G)$.
For instance, for {\sc Weighted Vertex Cover}, we define its \emph{annotated extension}, which we call {\sc Annotated Weighted Vertex Cover}, that takes as input a graph $G$ and two disjoint sets $R$ and $S$ of vertices of $G$, and asks for a minimum vertex cover $S^\star$ such that $S\subseteq S^\star$ and $S^\star\cap R=\emptyset$.

\subparagraph{General dynamic programming scheme.}
Our general scheme
essentially says that if a problem $\Pi$ has an annotated extension $\Pi'$ that is
\begin{itemize}
\item a nice problem and
\item solvable in \FPT-time parameterized by \oct,
\end{itemize}
then $\Pi$ is solvable in \FPT-time parameterized by \btw.
More specifically, it is enough to prove that $\Pi'$ is solvable in time $f(|X|)\cdot n^{\Ocal(1)}$ on an instance $(G,\Xcal)$ such that $G\setminus X$ is bipartite, where $\Xcal$ is a partition of $X$ corresponding to the annotation.
This general dynamic programming algorithm
works in a wider setting, namely for a general graph class $\Hcal$ that plays the role of bipartite graphs, with the additional condition that the annotated extension $\Pi'$ is ``{$\Hcal$-nice}'''; cf. \autoref{DP} for the details.

\subparagraph{Applications.}
We then apply this general  framework to obtain parameterized algorithms for several problems parameterized by bipartite treewidth.
For each of {\sc $K_t$-Subgraph-Cover} (\autoref{sec-kt}), {\sc Weighted Vertex Cover} {\sc/Independent Set} (\autoref{sec-mis}), {\sc Odd Cycle Transversal} (\autoref{sec-algo-oct}), and {\sc Maximum Weighted Cut} (\autoref{sec-cut}), we prove that the problem has an annotated  extension that is 1) nice and 2) solvable in \FPT-time parameterized by \oct, as discussed above.

To prove that an annotated problem has a nice reduction, we essentially use two ingredients.
Given two compatible boundaried graphs $\bf F$ and $\bf G$ with boundary $X$ (a boundaried graph is essentially a graph along with some labeled vertices that form a boundary, see the formal definition in \autoref{sec_bound}), an annotated problem is usually nice if the following hold:
\begin{itemize}
\item (\emph{Gluing property}) Given that we have guessed the annotation $\Xcal$ in the boundary $X$, a solution compatible with the annotation is optimal in the graph ${\bf F}\oplus{\bf G}$ obtained by gluing $\bf F$ and $\bf G$ if and only if it is optimal in each of the two glued graphs. In this case, it means that the optimum on $({\bf F}\oplus{\bf G},\Xcal)$ is equal to the optimum on $(F,\Xcal)$ modulo some constant depending only on $G$ and $\Xcal$.
\item (\emph{Gadgetization}) Given that we have guessed the annotation in the boundary $X\setminus\{v\}$ for some vertex $v$ in $X$, there is a small boundaried graph $G'$, that is bipartite (maybe empty), such that the optimum on $({\bf F}\oplus{\bf G},\Xcal)$ is equal to the optimum on $({\bf F}\oplus{\bf G'},\Xcal)$ modulo some constant depending only on $G$ and $\Xcal$.
\end{itemize}

The gluing property seems critical to show that a problem is nice.
This explains why we solve {\sc $H$-Subgraph-Cover} only when $H$ is a clique.
For any graph $H$, {\sc Annotated $H$-Subgraph-Cover} is defined similarly to {\sc Annotated Weighted Vertex Cover} by specifying vertices that must or must not be taken in the solution.
If $H$ is a clique, then we crucially use the fact that $H$ is a subgraph of ${\bf F}\oplus{\bf G}$ if and only if it is a subgraph of either $F$ or $G$ to prove that {\sc Annotated $H$-Subgraph-Cover} has the gluing property.
However, we observe that if $H$ is not a clique, then {\sc Annotated $H$-Subgraph-Cover} does not have the gluing property (see \autoref{no-glu-h}). This is the main difficulty that we face to solve {\sc $H$-Subgraph-Cover} in the general case.

Note also that if we define in a similar fashion the annotated extension of {\sc Odd Cycle Transversal} (that is, a set $S$ of vertices contained in the solution and a set $R$ of vertices that do not belong to the solution), then we can prove that this annotated extension does not have the gluing property. However, if we define {\sc Annotated Odd Cycle Transversal} as the problem that takes as input a graph $G$ and three disjoint sets $S,X_1,X_2$ of vertices of $G$ and aims at finding an odd cycle transversal $S^\star$ of minimum size such that $S\subseteq S^\star$ and $X_1$ and $X_2$ are on different sides of the bipartition obtained after removing $S^\star$, then {\sc Annotated Odd Cycle Transversal} does have the gluing property (see \autoref{glu-oct}).

For {\sc Maximum Weighted Cut}, the annotation is pretty straightforward: we use two annotation sets $X_1$ and $X_2$, corresponding to the vertices that will be on each side of the cut.  It is easy to see that this annotated extension has the gluing property.

Finding the right gadgets is the main difficulty to prove that a problem is nice.
As explained above, for {\sc Annotated Weighted Vertex Cover}, we replace the boundaried graph $\bf G$ by an edge that simulates the behavior of $G$ with respect to $v$, which is the only vertex that interest us (see \autoref{obs-mis}).
For {\sc Annotated Maximum Weighted Cut}, if $\Xcal=(X_1,X_2)$, the behavior of $G$ can be simulated by an edge between $v$ and a vertex in $X_1$ of weight equal to the optimum on $(G,(X_1,X_2\cup\{v\}))$ and an edge between $v$ and a vertex in $X_2$ of weight equal to the optimum on $(G,(X_1\cup\{v\},X_2))$ (see \autoref{obs+cut}).
For {\sc Annotated $K_t$-Subgraph-Cover}, if $\Xcal=(R,S)$, depending on the optimum on $(G,(R\cup\{v\},S))$ and the one on $(G,(R,S\cup\{v\}))$, we can show that the optimum on $({\bf F}\oplus{\bf G},\Xcal)$ is equal to the optimum on $(F,\Xcal)$ or $(F\setminus\{v\},\Xcal)$ modulo some constant (see \autoref{obs+1}).
For {\sc Annotated Odd Cycle Transversal}, if $\Xcal=(S,X_1,X_2)$, we can show that the optimum on $({\bf F}\oplus{\bf G},\Xcal)$ is equal modulo some constant to the optimum on either $(F,\Xcal)$, or $(F\setminus\{v\},\Xcal)$, or $(F',\Xcal)$, where $F'$ is obtained from $F$ by adding an edge between $v$ and either a vertex of $X_1$ or a vertex of $X_2$ (see \autoref{obs+oct}).

Finally, let us now mention some particular ingredients used to prove that the considered annotated problems are solvable in time $f(|X|)\cdot n^{\Ocal(1)}$ on an instance $(G,\Xcal)$ such that $G\setminus X$ is bipartite, where $\Xcal$ is a partition of a vertex set $X$ corresponding to the annotation.
For {\sc Annotated $K_t$-Subgraph-Cover} and {\sc Annotated Weighted Vertex Cover}, this is simply a reduction to {\sc (Weighted) Vertex Cover} on bipartite graphs.
For {\sc Odd Cycle Transversal}, we adapt the algorithm of Reed, Smith, and Vetta~\cite{ReedSV04find} that uses iterative compression to solve {\sc \Annotated Odd Cycle Transversal} in \FPT-time parameterized by \oct, so that it takes annotations into account (\autoref{lem-oct-annot}).
As for {\sc Maximum Weighted Cut} parameterized by \oct, the most important trick is to reduce to a $K_5$-odd-minor-free graph, and then use known results of Gr{\"{o}}tschel and Pulleyblank \cite{GrotschelP81weak} and Guenin \cite{Guenin01acha} to solve the problem in polynomial time (\autoref{k5}).

\subsubsection{\XP-algorithms}
\label{sec-overview-XP}

We now sketch some of the basic ingredients of the \XP-algorithms that we present in \autoref{xp} for {\sc $H$-(Induced-)Subgraph/Odd-Minor/Scattered-Packing}. The main observation is that,
if $H$ is 2-connected and non-bipartite, since the ``non-apex'' part of each bag is bipartite and $H$ is non-bipartite, in any $H$-subgraph/induced/scattered/odd-minor-packing and every bag of the decomposition, there are at most $\btw$ occurrences of $H$ that intersect that bag.
We thus guess these occurrences, and how they intersect the children, which allows us to reduce the number of children by just deleting those not involved in the packing. The guess of these occurrences is the dominant term in the running time of the resulting \XP-algorithm using this method. Note that for {\sc $H$(-Induced)-Subgraph/Scattered-Packing}, we can indeed easily guess those occurrences in \XP-time parameterized by \btw, as the total size of the elements of the packing intersecting a given bag is bounded by a function of \btw and $H$. However, for {\sc $H$-Odd-Minor-Packing} this is not the case anymore, as an element of the packing may contain an arbitrary number of vertices in the bipartite part of a bag. We overcome this issue as follows. The existence of an $H$-odd-minor is equivalent to the existence of a so-called \emph{odd $H$-expansion}, which  is essentially a collection of trees connected by edges preserving the appropriate parities of the resulting cycles (cf. \autoref{odd-eq}). In an odd $H$-expansion, the \emph{branch vertices} are those that have degree at least three, or that are incident to edges among different trees (cf. \autoref{sec-algo-odd-packing}). Note that, in an odd $H$-expansion, the number of branch vertices depends only on $H$. Equipped with this property, we first guess, at a given bag, the branch vertices of the packing that intersect that bag. Note that this indeed yields an \XP number of choices, as required. Finally, for each such choice, we use an \FPT-algorithm of Kawarabayashi, Reed, and Wollan~\cite{KawarabayashiRW11} solving the {\sc Parity $k$-Disjoint Paths} to check whether the guessed packing exists or not. This approach is formalized in \autoref{lem-algo-odd-minor-packing}.

It is worth mentioning that, as will be discussed in \autoref{sec-conclusions}, we leave as an open problem the existence of  \FPT-algorithms for the above packing problems parameterized by \btw.

\subsection{Hardness results}
\label{sec-overview-hardness}

Finally, we discuss some of the tools that we use to obtain  the  \textsf{para-}\NP-completeness results summarized in \autoref{table-results}, which can be found in \autoref{bip}. We present a number of different reductions, some of them consisting of direct simple reductions, such as the one we provide for \textsc{$3$-Coloring} in \autoref{lem-hard-coloring}.

Except for \textsc{$3$-Coloring}, all the considered problems fall into two categories: covering or packing problems. For the first family (cf. \autoref{sec-covering}), the \textsf{para-}\NP-completeness is an immediate consequence of a result of Yannakakis~\cite{Yannakakis81node} that characterizes hereditary graph classes $\Gcal$ for which \PbCov\ on bipartite graphs is polynomial-time solvable and those for which \PbCov\ remains \NP-complete (cf. \autoref{yan})

For the packing problems (cf. \autoref{sec-packing-hardness}), we do not have such a general result as for the covering problems, and we provide several reductions for different problems. For instance, we prove in \autoref{pack-bip} that if $H$ is a bipartite graph containing $P_3$ as a subgraph, then {\sc $H$-Subgraph-Packing} and {\sc $H$-Induced-Subgraph-Packing} are \NP-complete on bipartite graphs. The proof of \autoref{pack-bip} consists of a careful analysis and a slight modification of a reduction of Kirkpatrick and Hell~\cite{KirkpatrickH83onth} for the problem of partitioning the vertex set of an input graph $G$ into subgraphs isomorphic to a fixed graph $H$. The hypothesis about containing $P_3$ is easily seen to be tight.
We derive from \autoref{pack-bip} a similar result for {\sc $H$-Minor-Packing} in \autoref{pack-bip-min}.
%
Given that odd-minors preserve cycle parity (\autoref{odd-eq}), when $H$ is bipartite,
{\sc $H$-Odd-Minor-Packing} and {\sc $H$-Minor-Packing} are the same problem on bipartite graphs (\autoref{pack-bip-odd}).
Hence, the same hardness result holds for {\sc $H$-Odd-Minor-Packing} when $H$ is bipartite and contains $P_3$ as a subgraph.

In \autoref{pack-bip-scat} we prove that, if
 $H$ is a 2-connected bipartite graph with at least one edge,
then {\sc $H$-Scattered-Packing} is \NP-complete on bipartite graphs, by a simple reduction from the  {\sc Induced Matching} on bipartite graphs, which is known to be \NP-complete~\cite{Cameron89indu}.

Finally, in \autoref{lem-paraNP-non-bip} we prove that if $H$ is a 2-connected graph containing an edge and $q\in\bN_{\geq 2}$,
then {\sc $H$-Scattered-Packing} is {\sf para-NP}-complete parameterized by $q$-$\Bcal$-treewidth.
In fact, this reduction is exactly the same as the one of \autoref{pack-bip-scat}, with the extra observation that, if $G'$ is the graph constructed in the reduction, then the ``bipartite'' treewidth of $G'$ is at most the one of $H$ for $q\geq 2$.

\section{Definitions}\label{def}

In this section we give some definitions and preliminary results.
Basic definitions are given in \autoref{sec_basic}.
Odd-minors are defined in \autoref{sec_odd}.
Treewidth, bipartite treewidth, and other treewidth-like parameters are defined in \autoref{sec_tw}.
Finally, boundaried graphs are defined in \autoref{sec_bound}.

\subsection{Basic definitions}\label{sec_basic}

\subparagraph{Sets and integers.}
We denote by $\bN$ the set of non-negative integers.
Given two integers $p$ and $q$, the set $[p,q]$ contains every integer $r$ such that $p\leq r\leq q$.
For an integer $p\geq 1$, we set $[p]=[1,p]$ and $\bN_{\geq p}=\bN\setminus [0,p-1]$.
For a set $S$, we denote by $2^{S}$ the set of all subsets of $S$ and, given an integer $r\in[|S|]$,
we denote by $\binom{S}{r}$ the set of all subsets of $S$ of size exactly $r$.

\subparagraph{Parameterized complexity.}
A parameterized problem is a language $L\subseteq\Sigma^\star\times\bN$, where $\Sigma^\star$ is a set of strings over a finite alphabet $\Sigma$.
An input of a parameterized problem is a pair $(x, k)$, where $x$ is a string over $\Sigma$ and $k\in\bN$ is a parameter.
A parameterized problem is \emph{fixed-parameter tractable} (or \FPT) if it can be solved
in time $f (k) \cdot |x|^{\Ocal(1)}$ for some computable function $f$.
A parameterized problem is \XP\ if it can be solved
in time $f (k) \cdot |x|^{g(k)}$ for some computable functions $f$ and $g$.
A parameterized problem is {\sf para-NP}-complete if it is \NP-complete for some fixed value $k$ of the parameter.

\subparagraph{Partitions.}
Given $p\in\bN$, a \emph{$p$-partition} of a set $X$ is a tuple $(X_1,\ldots,X_p)$ of pairwise disjoint subsets of $X$ such that $X=\bigcup_{i\in[p]}X_i$.
We denote by $\Pcal_p(X)$ the set of all $p$-partitions of $X$.
Given a partition $\Xcal\in P_p(X)$, its domain $X$ is also denoted as $\cup \Xcal$.
A \emph{partition} is a $p$-partition for some $p\in\bN$.
Note that here, we allow empty sets in a $p$-partition and that the order matters.
Given $Y\subseteq X$, $\Xcal=(X_1,\ldots,X_p)\in\Pcal_p(X)$, and $\Ycal=(Y_1,\ldots,Y_p)\in\Pcal_p(Y)$, we say that $\Ycal\subseteq \Xcal$ if $Y_i\subseteq X_i$ for each $i\in[p]$.
Given a set $U$, two subsets $X,A\subseteq U$, and $\Xcal=(X_1,\ldots,X_p)\in\Pcal_p(X)$, $\Xcal\cap A$ denotes the partition $(X_1\cap A,\ldots,X_p\cap A)$ of $X\cap A$.

\subparagraph{Functions.}
Given two sets $A$ and $B$, and two functions $f,g:A\to2^B$, we denote by $f\cup g$ the function that maps $x\in A$ to $f(x)\cup g(x)\in2^B$.
Let $f:A\to B$ be an injection. Let $K\subseteq B$ be the image of $f$. By convention, if $f$ is referred to as a bijection, it means that $f$ maps $A$ to $K$.
Given a function $w:A\to\bN$, and $A'\subseteq A$, $w(A')=\sum_{x\in A'}w(x)$.

\subparagraph{Basic concepts on graphs.}
All graphs considered in this paper are undirected, finite, and without loops or parallel edges.
We use standard graph-theoretic notation and we refer the reader to \cite{Diestel10grap} for any undefined terminology.
For convenience, we use $uv$ instead of $\{u,v\}$ to denote an edge of a graph.
Let $G$ be a graph. In the rest of this paper we always use $n$ for the cardinality of $V(G)$,
and $m$ for the cardinality of $E(G)$, where $G$ is the input graph of the problem under consideration.
For $S \subseteq V(G)$, we set $G[S]=(S,E\cap \binom{S}{2} )$ and use the shortcut $G \setminus S$ to denote $G[V(G) \setminus S]$.
Given a vertex $v\in V(G)$, we denote by $N_{G}(v)$ the set of vertices of $G$ that are adjacent to $v$ in $G$.
Moreover, given a set $A\subseteq V(G)$, $N_{G}(A)=\bigcup_{v\in A}N_{G}(v)\setminus A$.
For $k\in\bN$, we denote by $P_k$ the path with $k$ vertices, and we say that $P_k$ has length $k-1$ (i.e., the \emph{length} of a path is its number of edges).
We denote by $\cc(G)$ the set of connected components of a graph $G$.
For $A,B\subseteq V(G)$, $E(A,B)$ denotes the set of edges of $G$ with one endpoint in $A$ and the other in $B$.
We say that $E'\subseteq E(G)$ is an \emph{edge cut} of $G$ if there is a partition $(A,B)$ of $V(G)$ such that $E'=E(A,B)$.
We say that a pair $(L,R)\in 2^{V(G)}\times 2^{V(G)}$ is a {\em separation} of $G$
if $L\cup R=V(G)$ and $E(L\setminus R,R\setminus L)=\emptyset$.
The \emph{order} of $(L,R)$ is $|L\cap R|$.
$L\cap R$ is called a \emph{$|L\cap R|$-separator} of $G$.
A graph $G$ is \emph{$k$-connected} if, for any separation $(L,R)$ of $G$ of order at most $k-1$, either $L\subseteq R$ or $R\subseteq L$.
A \emph{cut vertex} in a connected graph $G$ is a vertex whose removal disconnects $G$.
A \emph{block} of $G$ is a maximal connected subgraph of $G$ without a cut vertex.
A graph class $\Hcal$ is \emph{monotone} if the subgraphs of graphs in $\Hcal$ are also in $\Hcal$.
A graph class $\Hcal$ is \emph{hereditary} if the induced subgraphs of graphs in $\Hcal$ are also in $\Hcal$.
The \emph{torso} of a set $X\subseteq V(G)$, denoted by $\torso_G(X)$, is the graph obtained from $G[X]$ by making $N_G(C)$ a clique for each $C\in\cc(G\setminus X)$.
Given two graphs $G_1$ and $G_2$, and $q\in\bN$, a \emph{$q$-clique-sum} of $G_1$ and $G_2$ is obtained from their disjoint union by identifying a $q$-clique of $G_1$ with a $q$-clique of $G_2$, and then possibly deleting some edges of that clique.
A graph class $\Gcal$ is \emph{closed under $q$-clique-sums} if for each $G_1,G_2\in\Gcal$, any $q$-clique-sum of $G_1$ and $G_2$ also belongs to $\Gcal$.

\subsection{Odd-minors}\label{sec_odd}

\subparagraph{Colorings.} A \emph{coloring} of a graph $G$ is a function $c:V(G)\to\bN$.
Given $v\in V(G)$, $c(v)$ is called the \emph{color} of $v$ by $c$.
Given $k\in\bN$, a \emph{$k$-coloring} is a coloring $c:V(G)\to[k]$.
Given a coloring $c$ of a graph $G$ and an edge $uv\in E(G)$, we say that $uv$ is \emph{monochromatic} if $c(u)=c(v)$. Otherwise, we say the $uv$ is \emph{bichromatic}.
A coloring $c$ of a graph $G$ is said to be \emph{proper} if every edge of $G$ is bichromatic.
We say that a graph $G$ is \emph{$k$-colorable} if there exists a proper $k$-coloring of $G$.

\subparagraph{Minors and odd-minors.}
Let $G$ and $H$ be two graphs.
An \emph{$H$-expansion} in $G$ is a function $\eta$ with domain $V(H)\cup E(H)$ such that:
\begin{itemize}
	\item for every $v\in V(H)$, $\eta(v)$ is a subgraph of $G$ that is a tree $T_v$, called \emph{node} of $\eta$, such that each leaf of $T_v$ is adjacent to a vertex of another node of $\eta$, and $\eta(v)$ is disjoint from $\eta(w)$ for distinct $v,w\in V(H)$, and
	\item for every $uv\in E(H)$, $\eta(uv)$ is an edge $u'v'$ in $G$, called \emph{edge} of $\eta$, such that $u'\in V(\eta(u))$ and $v'\in V(\eta(v))$.
\end{itemize}
We denote by $\bigcup\eta$ the subgraph $\bigcup_{x\in V(H)\cup E(H)}\eta(x)$ of $G$.
Given a cycle $C$ in $H$, we set $\eta(C)$ to be the unique cycle in $G$ intersecting exactly the edges $\eta(uv)$ of $\eta$ such that $uv\in E(C)$.

If there is an $H$-expansion in $G$, then we say that $H$ is a \emph{minor} of $G$.
The \emph{contraction} of an edge $uv$ in a simple graph $G$ results in a simple graph $G'$
obtained from $G \setminus \{u,v\}$ by adding a new vertex $w$ adjacent to all the vertices
in the set $N_G(\{u,v\})$.
Equivalently, a graph $H$ is a \emph{minor} of a graph $G$ if $H$ can be obtained from a subgraph of $G$ by contracting edges.
We call such a subgraph \emph{model} of $H$ in $G$.
Note that the image of $\eta$ is a model of $H$ in $G$.
We say that a graph $G$ is \emph{$H$-minor-free} if $G$ excludes the graph $H$ as a minor.

\begin{lemma}\label{odd-eq}
Let $G$ and $H$ be two graphs.
The following statements are equivalent.
\begin{enumerate}
\item There is an $H$-expansion $\eta$ in $G$ and a 2-coloring of $\bigcup\eta$ that is proper in each node of $\eta$ and such that each edge of $\eta$ is monochromatic.
\item There is an $H$-expansion $\eta$ in $G$ such that
every cycle in $\bigcup\eta$ has an even number of edges in $\bigcup_{v\in V(H)}\eta(v)$.
\item There is an $H$-expansion $\eta$ in $G$ such that the length of
every cycle $C$ in $H$ has the same parity as the length of the cycle $\eta(C)$ in $\bigcup\eta$.
\item $H$ can be obtained from a subgraph of $G$ by contracting each edge of an edge cut.
\end{enumerate}
\end{lemma}

\begin{proof}
See \autoref{fig_odd-minor} to get some intuition.

{\bf 1 $\Rightarrow$ 2:} Let $\eta$ be an $H$-expansion in $G$ with a 2-coloring $c$ of $\bigcup\eta$ that is proper in each node of $\eta$ and such that each edge of $\bigcup\eta$ is monochromatic.
$\bigcup\eta$ is a subgraph of $G$.
The edges of $\bigcup_{v\in V(H)}\eta(v)$ are exactly the bichromatic edges of $\eta$.
Let $C$ be a cycle in $\bigcup\eta$.
We transform $C$ into a directed cyclic graph $C'$.
The bichromatic edges in $C'$ have alternatively color 1-2 and color 2-1.
Thus, since $C'$ is a cycle, the number of edges 1-2 and 2-1 is equal.
Hence, the number of bichromatic edges in $C$ is even.

{\bf 2 $\Rightarrow$ 1:} Let $\eta$ be an $H$-expansion in $G$ such that
every cycle in $\bigcup\eta$ has an even number of edges in $\bigcup_{v\in V(H)}\eta(v)$.
Let $v$ be an arbitrary vertex of $H$. We color $\eta(v)$ greedily to obtain a proper 2-coloring of the node.
Since $\eta(v)$ is a tree, there is only one proper 2-coloring up to isomorphism.
We extend this isomorphism greedily to the entire $\bigcup\eta$ so that each edge of $\eta$ is monochromatic and each node of $\eta$ is properly 2-colored.
Assume that there is a vertex $v$ of $\bigcup\eta$ that is not colorable by this greedy approach.
Then $v$ is part of a cycle $C$ in $\eta$ such that each other vertex of $C$ is colored, but the neighbors $u$ and $w$ of $v$ in $C$ give contradictory instructions for the coloring of $v$.
If $u$ has color $c_u$ and $w$ has color $c_w$ with $c_u\ne c_w$ (resp. $c_u=c_w$), then, given that $C$ has an even number of bichromatic edges, this implies that exactly one of $uv$ and $vw$ is bichromatic (resp. $uv$ and $vw$ are either both monochromatic or both bichromatic).
Thus, $v$ can be colored greedily.
Therefore, $\eta$ is an $H$-expansion in $G$ with a 2-coloring $c$ of $\bigcup\eta$ that is proper in each node of $\eta$ and such that each edge of $\eta$ is monochromatic.

{\bf 2 $\Leftrightarrow$ 3:} Let $\eta$ be an $H$-expansion in $G$.
By definition, there is a one-to-one correspondence between the edges of $\eta$ and the edges of $H$.
Therefore, $C$ is a cycle of $H$ if and only if $\eta(C)$ is a cycle of $\bigcup\eta$, and
there are as many edges of $\eta$ in $\eta(C)$ as the number of edges in $C$.
The other edges of $\eta(C)$ are in the nodes of $\eta$, i.e., in $\bigcup_{v\in V(H)}\eta(v)$.
Thus,
every cycle in $\bigcup\eta$ has an even number of edges in $\bigcup_{v\in V(H)}\eta(v)$ if and only if the length of
every cycle $C$ in $H$ has the same parity as the length of the cycle $\eta(C)$ in $\bigcup\eta$.

{\bf 1 $\Rightarrow$ 4:} Let $\eta$ be an $H$-expansion in $G$ with a 2-coloring $c$ of $\bigcup\eta$ that is proper in each node of $\eta$ and such that each edge of $\eta$ is monochromatic.
$\bigcup\eta$ is a subgraph of $G$.
Let $X_1$ and $X_2$ be the sets of vertices of $\eta$ with color 1 and 2 respectively.
By the definition of $c$, the set $E'=E(X_1,X_2)$ contains exactly the edges that are inside the nodes of $\eta$, and none of the edges joining two nodes of $\eta$.
Then $E'$ is an edge cut (between $X_1$ and $X_2$) of $\bigcup\eta$ and, by contracting $E'$ in $\bigcup\eta$, we obtain $H$.

{\bf 4 $\Rightarrow$ 1:} Let $G'$ be a subgraph of $G$ and $E'$ be an edge cut of $G'$ such that $H$ can be obtained by contracting $E'$ in $G'$.
Let $E''=E(G')\setminus E'$.
Let $G''$ be a graph obtained from $G'$ by removing edges in $E'$ such that every connected component of $G''\setminus E''$ is a spanning tree.
Then there is an $H$-expansion $\eta$ in $G$ such that $\bigcup\eta=G''$ and the edges of $\eta$ are exactly the edges in $E''$.
Let $(X_1,X_2)$ be the partition of $V(G')$ witnessing the edge cut $E'$.
Then, if we give color 1 to the vertices of $X_1$ and color 2 to the vertices of $X_2$, there is a proper 2-coloring of every node of $\eta$ and each edge of $\eta$ is monochromatic.
 \end{proof}

\begin{figure}[h]
\center
\includegraphics[scale=0.7]{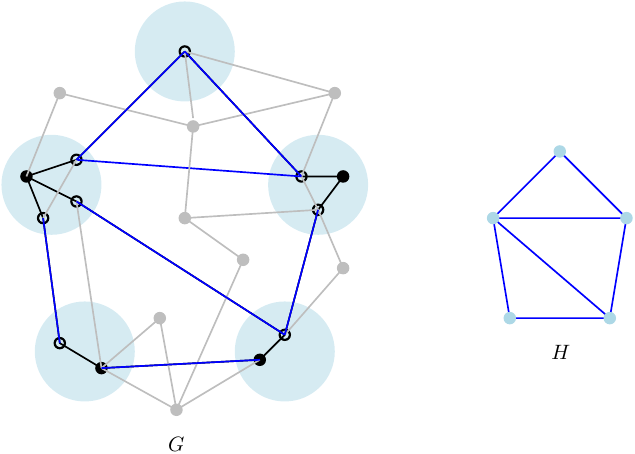}
\caption{An odd $H$-expansion $\eta$ in a graph $G$. The nodes of $\eta$ are the subgraphs in the blue disks, and the edges of $\eta$ are the blue edges in $G$.}
\label{fig_odd-minor}
\end{figure}

An $H$-expansion for which one of the statements of \autoref{odd-eq} is true is called an \emph{odd $H$-expansion} (see \autoref{fig_odd-minor} for an illustration).
Using Statement 3, we say that such an $H$-expansion in $G$ \emph{preserves cycle parity}.
If there is an $H$-expansion in $G$, then we say that $H$ is an \emph{odd-minor} of $G$.
Note that if $H$ is an odd-minor of $G$, then $H$ is a minor of $G$. However, the opposite does not always hold. For instance, $K_3$ is a minor of $C_4$, but is not an odd-minor of $C_4$.
We say that a graph $G$ is \emph{$H$-odd-minor-free} if $G$ excludes the graph $H$ as an odd-minor.
In particular, observe that bipartite graphs are exactly the $K_3$-odd-minor-free graphs and that the forests are exactly the $\{K_3,C_4\}$-odd-minor-free graphs.

\subsection{Treewidth-like parameters and preliminary results}\label{sec_tw}

\subparagraph{Treewidth.}
A \emph{tree decomposition} of a graph~$G$
is a pair~$(T,\chi)$ where $T$ is a tree and $\chi: V(T)\to 2^{V(G)}$
such that
\begin{itemize}
	\item $\bigcup_{t \in V(T)} \chi(t) = V(G)$,
	\item for every $e\in E(G)$, there is a $t\in V(T)$ such that $\chi(t)$ contains both endpoints of~$e$, and
	\item for every~$v \in V(G)$, the subgraph of~$T$ induced by $\{t \in V(T)\mid {v \in \chi(t)}\}$ is connected.
\end{itemize}

The {\em width} of $(T,\chi)$ is equal to $\max\big\{\left|\chi(t)\right|-1 \bigmid t\in V(T)\big\}$
and the {\em treewidth} of $G$, denoted by $\tw(G)$, is the minimum width over all tree decompositions of $G$.

For every node $t\in V(T)$, $\chi(t)$ is called \emph{bag} of $t$.
Given $tt'\in E(T)$, the \emph{adhesion} of $t$ and $t'$, denoted by $\adh(t,t')$, is the set $\chi(t)\cap\chi(t')$.

A \emph{rooted tree decomposition} is a triple $(T,\chi,r)$ where $(T,\chi)$ is a tree decomposition and $(T,r)$ is a \emph{rooted tree} (i.e., $T$ is a tree and $r\in V(T)$).
Given $t\in V(T)$, we denote by $\ch_r(t)$ the set of children of $t$ and by $\Par_r(t)$ the parent of $t$ (if $t\neq r$).
We set $\delta_t^r=\adh(t,\Par_r(t))$, with the convention that $\delta_r^r=\emptyset$.
Moreover, we denote by $G_t^r$ the graph induced by $\bigcup_{t'\in V(T_t)}\chi(t')$ where $(T_t,t)$ is the rooted subtree of $(T,r)$ containing all descendants of t.
We may use $\delta_t$ and ${G_t}$ instead of $\delta_t^r$ and ${G_t^r}$ when there is no risk of confusion.\bigskip

While our goal in this article is to study bipartite treewidth, defined below, we provide the following definition in a more general way, namely of a parameter that we call 1-$\Hcal$-treewidth, with the hope of it finding some application in future work.
We use the term $1$-$\Hcal$-treewidth to signify that the $\Hcal$-part of each bag intersects each neighboring bag in at most one vertex. This also has the benefit of avoiding confusion with $\Hcal$-treewidth defined in \cite{EibenGHK21meas}, which would be another natural name for this class of parameters.

\subparagraph{$1$-$\Hcal$-treewidth.}
Let $\Hcal$ be a graph class.
A \emph{1-$\Hcal$-tree decomposition} of a graph $G$ is a triple $(T,\alpha,\beta)$, where $T$ is a tree and $\alpha,\beta:V(T)\to 2^{V(G)}$, such that
\begin{itemize}
\item $(T,\alpha\cup\beta)$ is a tree decomposition of $G$, where $\alpha\cup \beta$ maps each $t\in V(T)$ to $\alpha(t)\cup\beta(t)$,
\item for every $t\in V(T)$, $\alpha(t)\cap\beta(t)=\emptyset$,
\item for every $t\in V(T)$, $G[\beta(t)]\in\Hcal$, and
\item for every $tt'\in E(T)$, $|(\alpha\cup\beta)(t')\cap\beta(t)|\leq 1$.
\end{itemize}
The vertices in $\alpha(t)$ are called \emph{apex vertices} of the node $t\in V(T)$.

The \emph{width} of $(T,\alpha,\beta)$ is equal to $\max\big\{\left|\alpha(t)\right|\bigmid t\in V(T)\big\}$.
The \emph{1-$\Hcal$-treewidth} of $G$, denoted by $(1,\Hcal)\text{-}\tw(G)$, is the minimum width over all 1-$\Hcal$-tree decompositions of $G$.

A \emph{rooted 1-$\Hcal$-tree decomposition} is a tuple $(T,\alpha,\beta,r)$ where $(T,\alpha,\beta)$ is a 1-$\Hcal$-tree decomposition and $(T,r)$ is a rooted tree.

Given that $(T,\alpha\cup\beta)$ is a tree decomposition, we naturally extend all definitions and notations concerning treewidth to 1-$\Hcal$-treewidth.

Observe also that a tree decomposition $(T,\chi)$ is also a 1-$\Hcal$-tree decomposition for every graph class $\Hcal$, in the sense that $(T,\chi,o)$ is a 1-$\Hcal$-tree decomposition, where $o:V(T)\to~\emptyset$.
Therefore, for every graph class $\Hcal$ and every graph $G$, $(1,\Hcal)\text{-}\tw(G)\leq\tw(G)+1$.

If $\Hcal$ is the graph class containing only the empty graph, then a 1-$\Hcal$-tree decomposition is exactly a tree decomposition.

\subparagraph{Remark.}
In \cite{EibenGHK21meas}, a parameter with a similar name is defined, namely \emph{$\Hcal$-treewidth}.
The $\Hcal$-treewidth of a graph $G$ is essentially the minimum treewidth of the graph induced by some set $X\subseteq V(G)$ such that the connected components of $G\setminus X$ belong to $\Hcal$.
This actually corresponds to the 0-$\Hcal$-treewidth (minus one), which is defined by replacing the ``1'' by a ``0'' in  the last item of the definition of a 1-$\Hcal$-tree decomposition above.
Indeed, let $(T,\alpha,\beta)$ be a 0-$\Hcal$-tree decomposition of a graph $G$ of width $k$.
Note that, for each distinct $t,t'\in V(T)$, $\beta(t)\cap\beta(t')=\emptyset$.
Let $X=\bigcup_{t\in V(T)}\alpha(t)$.
Then $(T,\alpha)$ is a tree decomposition of $X$ of width $k-1$.
Moreover, for each $t\in V(T)$, $G[\beta(t)]\in\Hcal$, and therefore the connected components of $G\setminus X$ belong to $\Hcal$.

\begin{figure}
\center
\includegraphics{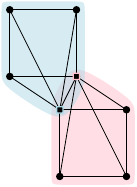}
\caption{A graph of bipartite treewidth one. A corresponding bipartite tree decomposition of width one is depicted, with two bags (one blue and one pink). The apex vertex of each bag is the squared vertex of the same color.}
\end{figure}

\subparagraph{Bipartite treewidth [adapted from \cite{DemaineHK10deco,Tazari12fast}].}
A graph $G$ is \emph{bipartite} if there is a partition $(A,B)$ of $V(G)$ such that $E(G)=E(A,B)$.
We denote the class of bipartite graphs by~$\Bcal$.
We focus here on the case where $\Hcal=\Bcal$. Then, we use the term \emph{bipartite treewidth} instead of 1-$\Hcal$-treewidth, and denote it by $\btw$. As mentioned in the introduction, this definition had already been used (more or less implicitly) in~\cite{DemaineHK10deco,Tazari12fast}.

Given that the bipartite graphs are closed under 1-clique-sums, we have the following.
\begin{observation}\label{obs-biptw-zero}
A graph has bipartite treewidth zero if and only if it is bipartite.
\end{observation}

Moreover, Campbell, Gollin, Hendrey, and Wiederrecht \cite{Campbell23odd-} recently announced an \FPT-approximation algorithm to construct a bipartite tree decomposition.

\begin{proposition}[\cite{Campbell23odd-}]\label{prop_seb}
There exist computable functions $f_1,f_2,g:\bN\to\bN$ and an algorithm that, given a graph $G$ and $k\in\bN$, outputs, in time $g(k)\cdot n^4\log n$, either  a report that $\btw(G)\geq f_1(k)$, or a bipartite tree decomposition of $G$ of width at most $f_2(k)$.
\end{proposition}

Bipartite treewidth is not closed under minors, given that contracting an edge in a bipartite graph (which has bipartite treewidth zero) may create a non-bipartite graph (which has positive bipartite treewidth). However, bipartite treewidth is closed under odd-minors, which is a desirable property to deal with odd-minor related problems.

\begin{lemma}\label{lem:closed-under-om}
Bipartite treewidth is closed under odd-minor containment.
\end{lemma}

\begin{proof}
Let $G$ be a graph and $H$ be an odd-minor of $G$. We want to prove that $\btw(H)\leq\btw(G)$.
By \autoref{odd-eq}, there is a subgraph $G'$ of $G$ and an edge cut $E'$ such that $H$ is obtained from $G'$ by contracting every edge in $E'$.

Since we only removed vertices and edges to obtain $G'$ from $G$, $\btw(G')\leq\btw(G)$.
It remains to show that $\btw(H)\leq\btw(G')$.
Let $(T,\alpha',\beta')$ be a bipartite tree decomposition of $G'$.
We transform $(T,\alpha',\beta')$ to a bipartite tree decomposition $(T,\alpha,\beta)$ of $H$ as follows.
For each $e=uv\in E'$ and for each $t\in V(T)$ such that $\{u,v\}\cap(\alpha'\cup\beta')(t)\neq\emptyset$,
\begin{itemize}
\item if $\{u,v\}\cap\alpha'(t)\neq\emptyset$, then the vertex $v_e$ resulting from contracting $e$ is placed in $\alpha(t)$,
\item otherwise, $v_e$ is placed in $\beta(t)$.
\end{itemize}
For each $v\in V(G')$ that is not involved in any contraction and for each $t\in V(T)$, if $v\in\alpha'(t)$ (resp. $v\in\beta'(t)$), then $v\in\alpha(t)$ (resp. $v\in\beta(t)$).

Let us show that $(T,\alpha,\beta)$ is indeed a bipartite tree decomposition of $H$.
For simplicity, we identify the vertices in the bags with the vertices of $H$.
It is easy to see that $(T,\alpha\cup\beta)$ is a tree decomposition of $H$, since it is obtained from $(T,\alpha'\cup\beta')$ by contracting the edge set $E'$.
The only difficult point is to check that, for each $v\in V(H)$ the subgraph of $T$ induced by the bags containing $v$ is connected. This trivially true if $v\in V(G)$. Otherwise, $v=v_{uw}$ for some $u,w\in V(G)$.
In this case, given that the subgraph of $T$ induced by the bags of $(T,\alpha'\cup\beta')$ containing $u$ (resp. $w$) is connected and that there is a bag of $(T,\alpha'\cup\beta')$ containing both $u$ and $w$, we also conclude that the subgraph of $T$ induced by the bags of $(T,\alpha\cup\beta)$ containing $v_{uw}$ is indeed connected.

Moreover, given that an edge with at least one endpoint in $\alpha'(t)$ contracts to a vertex in $\alpha(t)$, no new vertex is added to $\beta(t)$, and therefore, for any $t'\in V(T)\setminus\{t\}$, $|(\alpha\cup\beta)(t')\cap\beta(t)|\leq 1$.

It remains to prove that, for each $t\in V(T)$,  $H[\beta(t)]$ is bipartite.
Let $t\in V(T)$.
Let $E_t$ be the set of edges of $E'$ with both endpoints in $\beta'(t)$.
We have to prove  that the bipartite graph induced by $\beta'(t)$ in $G'$ remains bipartite after contracting $E_t$.
$E_t$ is an edge cut of $G'[\beta'(t)]$, witnessed by some vertex partition $(A_1,A_2)$.
Given a proper 2-coloring $(B_1,B_2)$ of $G'[\beta'(t)]$, which is bipartite, keep the same color for the vertices in $A_1$, and change the color of the vertices in $A_2$, i.e., define the coloring $(C_1,C_2)=((B_1\cap A_1)\cup(B_2\cap A_2),(B_2\cap A_1)\cup(B_1\cap A_2))$.
Thus, the monochromatic edges are exactly the edges of $E_t$.
Therefore, contracting $E_t$ gives a proper 2-coloring of $H[\beta(t)]$, so $H[\beta(t)]$ is bipartite.
Thus, $(T,\alpha,\beta)$ is a bipartite tree decomposition of $H$.

Moreover, since the contraction of an edge with both endpoints in $\beta'(t)$ is a vertex in $\beta(t)$, it follows that $|\alpha(t)|\leq|\alpha'(t)|$ for every $t\in V(T)$. Therefore, $\btw(H)\leq\btw(G')$.
 \end{proof}

A natural generalization of bipartite treewidth can be made by replacing the ``1'' in the last item of the definition of 1-$\Hcal$-tree decomposition by any $q\in\bN$, hence defining \emph{$q$-$\Hcal$-tree decompositions} and \emph{$q$-$\Hcal$-treewidth}, denoted by $(q,\Hcal)$-$\tw(G)$.
For $q\geq 2$, however, $q$-$\Bcal$-treewidth is not closed under odd-minor containment, as we prove in \autoref{lem:not-closed-under-om} below.
Additionally, given that for $q\in\{0,1\}$, $\torso_{G\setminus\alpha(t)}(\beta(t))=G[\beta(t)]$, we could replace the third item by the property ``for every $t\in V(T)$, $\torso_{G\setminus\alpha(t)}(\beta(t))\in\Hcal$'', hence defining what we call \emph{$q$-torso-$\Hcal$-tree decompositions} and \emph{$q$-torso-$\Hcal$-treewidth}, denoted by $(q,\Hcal)^{\star}$-$\tw(G)$.
However, we prove in \autoref{lem:not-closed-under-om} that, for $q\geq 2$, $(q,\Bcal)^{\star}$-$\tw$ is also not closed under odd-minors.
These facts, in our opinion, provide an additional justification for the choice of $q\le1$ in the definition of bipartite treewidth.

\begin{lemma}\label{lem:not-closed-under-om}
For $q\geq 2$, $q$(-torso)-$\Bcal$-treewidth is not closed under odd-minor containment.
In particular, for any $t\geq 3$, there exist a graph $G$ and an odd-minor $H$ of $G$, such that
$(q,\Bcal)$-$\btw(G)=0$ and $(q,\Bcal)$-$\btw(H)=t-2$, and
$(q,\Bcal)^{\star}$-$\btw(G)\leq 1$ and $(q,\Bcal)^{\star}$-$\btw(H)=t-2$.
\end{lemma}

\begin{proof}
Let $t\in\bN_{\geq 3}$ and let $K_t'$ (resp. $K_t''$) be the graph obtained from $K_t$ by subdividing every edge once (resp. twice).
Let $V'=\{v_1,\ldots,v_t\}$ be the set of vertices of $K_t''$ that are the original vertices of $K_t$.
$K_t$ is an odd-minor of $K_t''$ since $K_t$ can be obtained from $K_t''$ by contracting the edge cut $E(V',V(K_t'')\setminus V')$. Note also that $K_t'$ is bipartite.
We show that taking $G = K_t''$ and $H = K_{t}$ satisfies the statement of the lemma.

Given that $K_t$ is a complete graph, it has to be fully contained in one bag of
any tree decomposition, so in particular of any $q$(-torso)-$\Bcal$-tree decomposition.
Since the smallest odd cycle transversal of $K_t$ has size $t-2$, we have that
$(q,\Bcal)^{\star}$-$\tw(K_t)=(q,\Bcal)$-$\tw(K_t)=t-2$.

Let us first prove that $(q,\Bcal)$-$\tw(K_t'')= 0$.
For $i,j\in[t]$ with $i<j$, let $e_{i,j}$ be the path of length three between $v_i$ and $v_j$.
Let $T$ be a tree with one central vertex $x_0$ and, for each $i,j\in[t]$ with $i<j$, a vertex $x_{i,j}$ only adjacent to $x_0$ (thus, $T$ is a star).
Let $\beta(x_0)=V'$ and $\beta(x_{i,j})=V(e_{i,j})$ for each $i,j\in[t]$ with $i<j$.
Let $\alpha(x)=\emptyset$  for each $x\in V(T)$.
$V'$ is an independent set, so $G[V']$ is bipartite. Note that paths are bipartite.
Moreover, each adhesion contains at most two vertices of $\beta(x_0)$ and two vertices of $\beta(x_{i,j})$.
Hence, $(T,\alpha,\beta)$ is a $q$-$\Bcal$-tree decomposition of $K_t''$, for $q\ge2$, and has width zero.

Let us now prove that $(q,\Bcal)^{\star}$-$\btw(K_t'')\leq 1$.
Let $u_{i,j}$ and $w_{i,j}$ be the internal vertices of $e_{i,j}$, such that $u_{i,j}$ is adjacent to $v_i$.
Let $V_1$ (resp. $V_2$) be the set of vertices $u_{i,j}$ (resp. $w_{i,j}$).
We construct a $q$-torso-$\Bcal$-tree decomposition $(T,\alpha',\beta')$ of $K_t''$ as follows.
We set $\alpha'(x_0)=\emptyset$ and $\beta'(x_0)=V'\cup V_2$.
For each $i,j\in[t]$ with $i<j$, we set $\alpha'(x_{i,j})=\{v_i\}$ and $\beta'(x_{i,j})=\{u_{i,j},w_{i,j}\}$.
Observe that $\torso_{K_t''\setminus\alpha'(x_0)}(\beta'(x_0))=K_t'$,
since each path $v_i$-$u_{i,j}$-$w_{i,j}$ is replaced by an edge $v_iw_{i,j}$.
Thus, it is bipartite.
Similarly, the torso at each other node of $T$ is an edge, and hence is bipartite as well.
Moreover, each adhesion contains at most two vertices of $\beta'(x_0)$ and one vertex of $\beta'(x_{i,j})$.
Hence, $(T,\alpha',\beta')$ is indeed a $q$-torso-$\Bcal$-tree decomposition of $K_t''$, for $q\ge2$, and has width one.
Therefore, $(q,\Bcal)^{\star}$-$\btw(K_t'')\leq 1$.

Hence, $q$(-torso)-$\Bcal$-treewidth is not closed under odd-minor containment and the gap between a graph and an odd-minor of this graph can be arbitrarily large.
 \end{proof}

As mentioned in \autoref{sec-overview-dp}, one of the main difficulties for doing dynamic programming on (rooted) bipartite tree decompositions is the lack of a way to upper-bound the number of children of each node of the decomposition.
As shown in the next lemma, the notion of ``nice tree decomposition'' is not generalizable to bipartite tree decompositions.

\begin{lemma}\label{no-nice}
For any $t\in\bN$, there exists a graph $G$ such that $\btw(G)=1$ and any bipartite tree decomposition of $G$ whose nodes all have at most $t$ neighbors has width at least $t-1$.
\end{lemma}

\begin{proof}
Let $G$ be the graph obtained from $K_{t,t}$ by gluing a new pendant triangle $H_v$ to each vertex $v$ of $K_{t,t}$ (that is, $v$ is identified with one vertex of its pendant triangle).
Let $T$ be the star $K_{1,2t}$, with vertex set $\{t_0\}\cup\{t_v\mid v\in V(K_{t,t})\}$.
Let $\alpha(t_0)=\emptyset$, $\beta(t_0)=V(K_{t,t})$, and for every $v\in V(K_{t,t})$, $\alpha(t_v)=\{v\}$ and $\beta(t_v)=V(H_v)\setminus \{v\}$. It can be easily verified that
$\Tcal=(T,\alpha,\beta)$ is a bipartite tree decomposition of $G$ of width one.
Given that $G$ is not bipartite, \autoref{obs-biptw-zero} implies that $\btw(G)=1$.
Note that node $t_0$ has $2t$ neighbors.
For any bipartition $(A,B)$ of $V(K_{t,t})$ such that $A,B\neq\emptyset$, we have $|E(A,B)|\geq t$.
Indeed, if $a_1$ and $a_2$ are the vertices of $A$ on each side of $K_{t,t}$, where we can assume without loss of generality $a_1\ge 1$ and $a_2\le k-1$ (since $A,B\neq\emptyset$), then $E(A,B)$ has size $a_1(k-a_2)+a_2(k-a_1)=(k-a_1)a_1+(k-a_2)a_2+(a_1-a_2)^2\ge k-a_1+a_2+(a_1-a_2)^2\ge k$.
Hence, for any bipartite tree decomposition $\Tcal'$ of $G$ such that $V(K_{t,t})$ is not totally contained in one bag, there is an adhesion of two bags of size at least $t$, so the width of $\Tcal'$ is at least $t-1$.
If $V(K_{t,t})$ is fully contained in one bag, however, the only way to reduce the number of children to $k$, for some integer $k$, is to add $2t-k$ of the pendant triangles inside the same bag. But then this bag has odd cycle transversal number at least $2t-k$, so the obtained bipartite tree decomposition has width at least $2t-k$. Hence, if we want that $k \leq t$, then the corresponding width is at least $2t-k \geq t$.
 \end{proof}

\subsection{Boundaried graphs} \label{sec_bound}

To design dynamic programming algorithms, we need to first define boundaried graphs.\medskip

Let $t\in \bN$. A \emph{$t$-boundaried graph} is a triple ${\bf G}=(G,B,\rho)$ where $G$ is a graph, $B\subseteq V(G)$, $|B|=t$, and $\rho:B\to\bN$ is an injective function.
We say that $B$ is the \emph{boundary} of ${\bf G}$ and we write $B=\bd(G)$.
We call ${\bf G}$ \emph{trivial} if all its vertices belong to the boundary.
We say that two $t$-boundaried graphs ${\bf G}_1=(G_1,B_1,\rho_1)$ and ${\bf G}_2=(G_2,B_2,\rho_2)$ are \emph{isomorphic} if $\rho_1(B_1)=\rho_2(B_2)$ and there is an isomorphism from $G_1$ to $G_2$ that extends the bijection $\rho_2^{-1}\circ\rho_1$.
A triple $(G,B,\rho)$ is a \emph{boundaried graph} if it is a $t$-boundaried graph for some $t\in\bN$.
A boundaried graph $\bf F$ is a \emph{boundaried induced subgraph} (resp. \emph{boundaried subgraph}) of $\bf G$ if $\bf F$ can be obtained from $\bf G$ by removing vertices (resp. and edges).
A boundaried graph $\bf F$ is a \emph{boundaried odd-minor} of $\bf G$ if $\bf F$ can be obtained from a boundaried subgraph ${\bf G}'$ of $\bf G$ by contracting an edge cut such that every vertex in $\bd({\bf G}')$ is on the same side of the cut.
We say that two boundaried graphs ${\bf G}_1=(G_1,B_1,\rho_1)$ and ${\bf G}_2=(G_2,B_2,\rho_2)$ are \emph{compatible} if $\rho_1(B_1)=\rho_2(B_2)$ and $\rho_2^{-1}\circ\rho_1$ is an isomorphism from $G_1[B_1]$ to $G_2[B_2]$.

\begin{figure}[h]
\center
\includegraphics{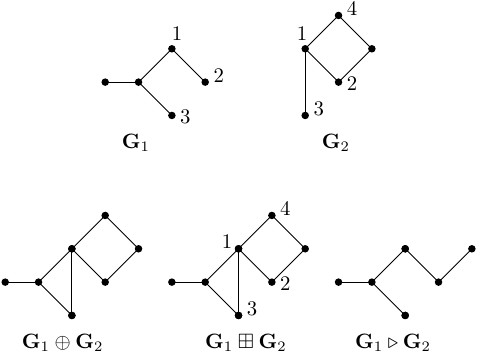}
\caption{Examples of use of the operations $\oplus$, $\boxplus$, and $\triangleright$.}
\label{fig_gluing_operations}
\end{figure}
Given two boundaried graphs ${\bf G}_1=(G_1,B_1,\rho_1)$ and ${\bf G}_2=(G_2,B_2,\rho_2)$, we define ${\bf G}_1\oplus {\bf G}_2$ as the unboundaried graph obtained if we take the disjoint union of $G_1$ and $G_2$ and, for every $i\in\rho_1(B_1)\cap\rho_2(B_2)$, we identify vertices $\rho_1^{-1}(i)$ and $\rho_2^{-1}(i)$.
Note that we do not ask ${\bf G}_1$ and ${\bf G}_2$ to be compatible.
If two vertices are adjacent in at least one boundary, then there are adjacent in ${\bf G}_1\oplus {\bf G}_2$.
If $v$ is the result of the identification of $v_1:=\rho_1^{-1}(i)$ and $v_{2}:=\rho_2^{-1}(i)$ then we say that $v$ is the \emph{heir} of
$v_{i}$ from ${\bf G}_i, i\in[2]$. If $v$ is either a vertex
in $B_{1}$ where $ρ_{1}(v)\not\in ρ_{1}(B_{1})\cap ρ_{2}(B_{2})$
or a vertex
in $B_{2}$ where $ρ_{2}(v)\not\in ρ_{1}(B_{1})\cap ρ_{2}(B_{2})$,
then $v$ is also a (non-identified) vertex of ${\bf G}_1\oplus {\bf G}_2$  and is an \emph{heir} of itself (from ${\bf G}_1$ or ${\bf G}_2$ respectively). For $i\in[2]$,
and given an edge $vu$ in ${\bf G}_1\oplus {\bf G}_2$,
we say that $vu$ is the \emph{heir} of an edge $v'u'$ from ${\bf G}_{i}$
if $v'$ is the heir of $v$ from $\bf G_{i}$, $u'$ is the heir of $u$ from $\bf G_{i}$, and $v'u'$
is an edge of $G_{i}$.
If $x'$ is an heir of $x$ from ${\bf G}=(G,B,\rho)$ in ${\bf G}'$, then we write $x=\heir_{\bf G,G'}(x')$.
Given $B'\subseteq B$, we write $\heir_{\bf G,G'}(B')=\bigcup_{v\in B'}\heir_{\bf G,G'}(x')$.
We also define ${\bf G}_1\boxplus {\bf G}_2$ as the {\sl boundaried} graph $({\bf G}_1\oplus {\bf G}_2,B,\rho)$, where $B$ is the set of all heirs from $\bf G_1$ and $\bf G_2$ and $\rho:B\to\bN$ is the union of $\rho_1$ and $\rho_2$ after identification.
Note  that in circumstances where  $\boxplus$ is repetitively applied, the heir relation is maintained due to its transitivity.
Moreover, we define ${\bf G}_1\triangleright {\bf G}_2$ as the unboundaried graph $G$ obtained from ${\bf G}_1\oplus {\bf G}_2$ by removing all heirs from  ${\bf G}_{2}$ that are not heirs from ${\bf G}_{1}$
and all heirs of edges from ${\bf G}_2$ that are not heirs of edges from ${\bf G}_1$.
Informally, ${\bf G}_1\triangleright {\bf G}_2$ is obtained from ${\bf G}_1\oplus {\bf G}_2$ by removing vertices and edges of $G [B_2]$ that are not in $G [B_1]$.
Note that $\triangleright$ is not commutative.
See \autoref{fig_gluing_operations} for an illustration of the operations $\oplus$, $\boxplus$, and $\triangleright$.

\section{General dynamic programming to obtain \FPT-algorithms}\label{fpt}

In this section, we introduce a framework for obtaining \FPT-algorithms for problems
parameterized by the width of a given bipartite tree decomposition of the input graph.
In \autoref{sec:nice} we introduce the main technical notion of a \emph{nice problem}
and the necessary background,
and in \autoref{sec:nice:dp} we provide dynamic programming algorithms for nice problems.
In \autoref{sec-results-FPT},
we give applications to concrete problems.

\subsection{Nice problems}\label{sec:nice}
\label{sec-nice-reductions}

All algorithms we give for problems on graphs of bounded \btw\
follow the same strategy.
To avoid unnecessary repetition, we introduce a framework that captures
the features that the problems have in common with respect to their algorithms using bipartite tree decompositions.
Naturally, the algorithms use dynamic programming along a rooted bipartite tree decomposition.
However, as the bags in the decomposition can now be large,
we cannot apply brute-force approaches to define table entries as we can, for instance,
on standard tree decompositions when dealing with treewidth.

Suppose we are at a node $t$ with children $t_1, \ldots, t_d$.
Since the size of $\alpha(t)$ is bounded by the width,
we can store all possible ways in which solutions interact with $\alpha(t)$.
Moreover, since each adhesion has at most one vertex from $\beta(t)$,
the size of each adhesion is still bounded in terms of the bipartite treewidth.
Therefore, we can store one table entry
for each way in which a solution can interact with the adhesion $\delta_t$ of $t$ and its parent.
However, since the size of $\beta(t)$ is unbounded,
there can now be an exponential (in $n$)
number of choices of table entries that are
``compatible'' with the choice $\Xcal$ made for $\delta_t$,
so we cannot simply brute-force them to determine the optimum value corresponding to $\Xcal$.
To overcome this, we apply the following strategy:
First, since the size of $\alpha(t)$ is bounded in terms of the bipartite treewidth,
we guess which choice $\Acal$
of the interaction of the solution with $\alpha(t) \cup \delta_t$ that extends $\Xcal$
leads to the optimum partial solution.
For each $i \in [d]$, there may be a vertex $v_{t_i} \in \delta_{t_i} \cap \beta(t)$
whose interaction with the partial solution remained undecided.
We replace, for each $i \in [d]$,
the subgraph $G_{t_i} \setminus \delta_{t_i}$
with a simply structured subgraph that simulates the behaviour
of the table at $t_i$ when it comes to the decision of how $v_{t_i}$ interacts with the solution,
under the choice of $\Acal$ for $\alpha(t) \cup \delta_t$.
The crux is that the resulting graph will have an odd cycle transversal
that is bounded in terms of the size of $\alpha(t)$,
so we can apply known \FPT-algorithms parameterized by odd cycle transversal
to determine the value of the table entry.
These notions can be formalized not only for bipartite treewidth, but for any 1-$\Hcal$-treewidth,
so we present them in full generality here.
We also depart from using tree decompositions explicitly,
and state them in an equivalent manner in the language of boundaried graphs.

First, let us formalize the family of problems we consider
which we refer to as \emph{optimization problems}.
Here, solutions correspond in some sense to partitions of the vertex set,
and we want to optimize some property of such a partition.
For instance, if we consider \textsc{Odd Cycle Transversal},
then this partition has three parts, one for the vertices in the solution,
and one part for each part of the bipartition of the vertex set of the graph obtained by removing the solution vertices, and we want to minimize the size of the first part of the partition.
(It will become clear later why we keep one separate part for each part of the bipartition.)
In \textsc{Maximum Cut},
the partition simply points to which side of the cut each vertex is on,
and we want to maximize the number of edges going across.

A \emph{$p$-partition-evaluation function on graphs} is a function $f$ that receives as input a graph $G$ along with a $p$-partition $\mathcal{P}$ of its vertices and outputs a non-negative integer.
Given such a function $f$ and some choice $\mathsf{opt}\in\{\max,\min\}$
we define the associated graph parameter $\mathsf{p}_{f,\mathsf{opt}}$ where, for every graph $G$,
$$\mathsf{p}_{f,\mathsf{opt}}(G)=\mathsf{opt}\{f(G,\mathcal{P})\mid \text{ $\mathcal{P}$ is a $p$-partition of $V(G)$}\}.$$
An \emph{optimization graph partition problem} (\emph{optimization problem} for short) is a problem that can be expressed as follows.
\begin{center}
	\fbox{
		\begin{minipage}{5cm}
			\noindent\textbf{Input}:~~A graph $G$.\\
			\textbf{Objective}:~~Compute $\mathsf{p}_{f,\mathsf{opt}}(G)$.
		\end{minipage}
	}
\end{center}
To represent the case when we made some choices for the (partial) solution
to an optimization problem, such as $\Acal$ above,
we consider \emph{annotated} versions of such problems.
They extend the function $\mathsf{p}_{f,\mathsf{opt}}$
so to receive as input, apart from a graph,
a set of annotated sets in the form of a partition
$\mathcal{X}\in\Pcal_p(X)$ of some $X\subseteq V(G)$.
More formally, the \emph{annotated extension} of  $\mathsf{p}_{f,\mathsf{opt}}$
is the parameter  $\hat{\mathsf{p}}_{f,\mathsf{opt}}$  such that
$$\hat{\mathsf{p}}_{f,\mathsf{opt}}(G,\mathcal{X})=\mathsf{opt}\{f(G,\mathcal{P})\mid \text{ $\mathcal{P}$ is a $p$-partition of $V(G)$ with $\mathcal{X}\subseteq \mathcal{P}$}\}.$$
Observe that  ${\mathsf{p}}_{f,\mathsf{opt}}(G)=\hat{\mathsf{p}}_{f,\mathsf{opt}}(G,\emptyset^p)$, for every graph $G$.
The problem $\Pi'$ is a \emph{$p$-annotated extension} of the optimization problem $\Pi$ if $\Pi$ can be expressed by some $p$-partition-evaluation function $f$ and some choice $\mathsf{opt}\in\{\max,\min\}$, and if $\Pi'$ can be expressed as follows.
\begin{center}
	\fbox{
		\begin{minipage}{9cm}
			\noindent\textbf{Input}:~~A graph $G$ and $\Xcal\in\Pcal_p(X)$ for some $X\subseteq V(G)$.\\
			\textbf{Objective}:~~Compute $\hat{\mathsf{p}}_{f,\mathsf{opt}}(G,\Xcal)$.
		\end{minipage}
	}
\end{center}
We also say that $\Pi'$ is a \emph{$p$-annotated problem}.

Let us turn to the main technical tool introduced in this section
that formalizes the above idea, namely \emph{nice problems}.
First, we may assume that the vertices of $G$ are labeled injectively via $\sigma$.
Then, each graph $G_{t_i}$, for $i \in [d]$,
naturally corresponds to a boundaried graph
${\bf G_i} = (G_{t_i}, \delta_{t_i}, \sigma_{|\delta_{t_i}})$;
from now on let $X_i = \delta_{t_i}$.
The part of $G_t$ that will be modified
can be viewed as a boundaried graph ${\bf G}$
which is essentially obtained as $\boxplus_{i \in [d]} {\bf G_i}$.
However, as we want to fix a choice of how the partial solution interacts with $\delta_t$,
we include these corresponding vertices in ${\bf G}$ as well,
modeled as a trivial boundaried graph ${\bf X}$,
making ${\bf G} = {\bf X} \boxplus (\boxplus_{i \in [d]} {\bf G_i})$ (see \autoref{fig_collection} for an illustration).

Denote the boundary of ${\bf G}$ by $X$.
The set $X$ is partitioned into $(A, B)$,
corresponding to
$(X \cap (\alpha(t) \cup \delta(t)), X \cap (\beta(t) \setminus \delta(t)))$,
and the fact that the adhesion between $t_i$ and $t$ had at most one vertex in common with $\beta(t)$
now materializes as the fact that for each $i \in [d]$,
${\bf G}$ has at most one vertex outside of $A$
that is an heir of a vertex in~${\bf G_i}$.
Fixing a choice for $\alpha(t) \cup \delta_t$
now corresponds to choosing a partition $\Acal$ of the set $A$.
As we assume that all table entries at the children have been computed,
we assume knowledge of all values
$\hat{\mathsf{p}}_{f,\mathsf{opt}}(G_{t_i},\mathcal{X}_i)$,
for all $i \in [d]$ and $\Xcal_i \in \Pcal_p(X_i)$.
This finishes the motivation of the input of a nice problem.

Given the pair $({\bf G}, \Acal)$,
a nice problem outputs a tuple $({\bf G'} = (G', X', \rho'), \Acal', s')$, called \emph{nice reduction},
with the following desired properties.
${\bf G'}$ can be constructed by gluing $d'$ boundaried graphs
plus one trivial one (for some $d'\in\bN$), similarly to ${\bf G}$.
$\Acal'$ is a $p$-partition of a set $A' \subseteq V(G')$
whose size is at most the size of $A$ plus a constant.
No matter what the structure of the graph of the
vertices in $(\alpha \cup \beta)(t)$ looked like
(remember, so far we carved out only the adhesions),
the solutions are preserved, up to an offset of $s'$.
This is modeled by saying that for each boundaried graph ${\bf F}$
(which corresponds to the remainder of the bag at $t$)
compatible with ${\bf G}$,
 $\hat{\mathsf{p}}_{f,\mathsf{opt}}({\bf G}\oplus {\bf F},\mathcal{A})
 = \hat{\mathsf{p}}_{f,\mathsf{opt}}({\bf G}'\triangleright{\bf F},\Acal')+s'.$
The reason why we use the $\triangleright$-operator in the right-hand side of the equation
is the gadgeteering happening in the later sections.
To achieve the ``solution-preservation'',
we might have to add or remove vertices, or change adjacencies between vertices in $X_i$ (for {\sc Odd Cycle Transversal} and {\sc Maximum Cut} for instance).

The last condition corresponds to our aim that
if the bag at $t$ induces a graph of small \oct (now, a small modulator to a graph class $\Hcal$),
then the entire graph resulting from the operation (${\bf G}'\triangleright{\bf F}$) should
have a small modulator to $\Hcal$ (namely $A'$).
All remaining conditions are related to the efficiency of the nice reduction.


\begin{figure}
\center
\includegraphics[scale=1]{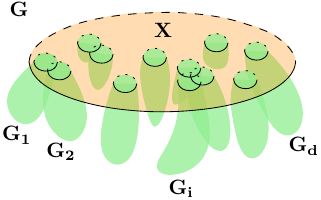}
\caption{Illustration of ${\bf G}={\bf X}\boxplus(\boxplus_{i\in[d]}{\bf G}_i)$.}
\label{fig_collection}
\end{figure}

\subparagraph{Nice problem and nice reduction.}
	Let $p\in \bN$,
	let $\Hcal$ be a graph class,
	and let $\Pi$ be a $p$-annotated problem corresponding to some choice of a $p$-partition-evaluation function $f$ and some $\mathsf{opt}\in\{\max,\min\}$.
	We say that $\Pi$ is a \emph{$\Hcal$-nice problem} if there exists an algorithm that receives as input
	\begin{itemize}
	\item a boundaried  graph ${\bf G}=(G,X,ρ)$, 	
	\item  a trivial boundaried graph ${\bf X}=(G[X],X,\rho_X)$
		and a collection $\{{\bf G}_i\mid i\in[d]\}$ of boundaried graphs with ${\bf G}_i=(G_{i},X_{i},ρ_{i})$ for $i\in[d]$,
		such that $d\in\bN$ and ${\bf G}={\bf X}\boxplus(\boxplus_{i\in[d]}{\bf G}_i)$ (see \autoref{fig_collection}),
	\item a partition $(A,B)$  of $X$ such that
for all $i \in [d]$,
	$|\heir_{\bf G_i,G}(X_i)\setminus A|\leq 1$,
	\item some $\Acal\in \Pcal_p(A)$, and
\item for every $i\in[d]$ and each $\mathcal{X}_i \in \Pcal_p(X_i)$, the value $\hat{\mathsf{p}}_{f,\mathsf{opt}}(G_i,\mathcal{X}_i)$,
	\end{itemize}
	and outputs, in time $\Ocal(|A|\cdot d),$ a tuple
	$({\bf G}'=(G',X',ρ'), \mathcal{A}', s')$, called \emph{$\Hcal$-nice reduction of the pair $(\mathbf{G},\Acal)$ with respect to $\Pi$}, such that the following hold (see \autoref{fig-nice-prob} for an illustration).
	\begin{itemize}
 		\item[(i)] There are sets $A'\subseteq V(G')$ and $\Acal'\in\Pcal_p(A')$ with $|A'|=|A|+\Ocal(1)$.
		\item[(ii)] There is a trivial boundaried graph ${\bf X}'=(G[X'],X',\rho_{X'})$ and a collection $\{{\bf G}_i'=(G_{i}',X_{i}',ρ_{i}')\mid i\in[d']\}$, where $d'\in\bN$, of boundaried graphs such that ${\bf G}' = {\bf X}'\boxplus(\boxplus_{i \in [d']} {\bf G}_i')$ and
$|V(G')| \le |X| + \Ocal(|B|)$, $|E(G')| \le |E(G[X])| + \Ocal(|B|)$.
		\item[(iii)]
			For any boundaried graph ${\bf F}$ compatible with ${\bf G}$, it holds that
\begin{eqnarray}
 \hat{\mathsf{p}}_{f,\mathsf{opt}}({\bf G}\oplus {\bf F},\mathcal{A})
 & = &\hat{\mathsf{p}}_{f,\mathsf{opt}}({\bf G}'\triangleright{\bf F},\Acal')+s'.
\nonumber
\end{eqnarray}
\item[(iv)] For any boundaried graph ${\bf F}=(F,X_F,\rho_F)$ compatible with ${\bf G}$,
if $\bar{F}\setminus A_F\in\Hcal$, where $\bar{F}=({\bf F}\oplus {\bf G})[\heir_{{\bf F},{\bf G}\oplus{\bf F}}(V(F))]$ and $A_F=\heir_{{\bf G},{\bf G}\oplus{\bf F}}(A)$,
then $({\bf G}'\triangleright{\bf F})\setminus A'\in\Hcal$.
\end{itemize}

\begin{figure}
	\centering
	\includegraphics[scale=0.5]{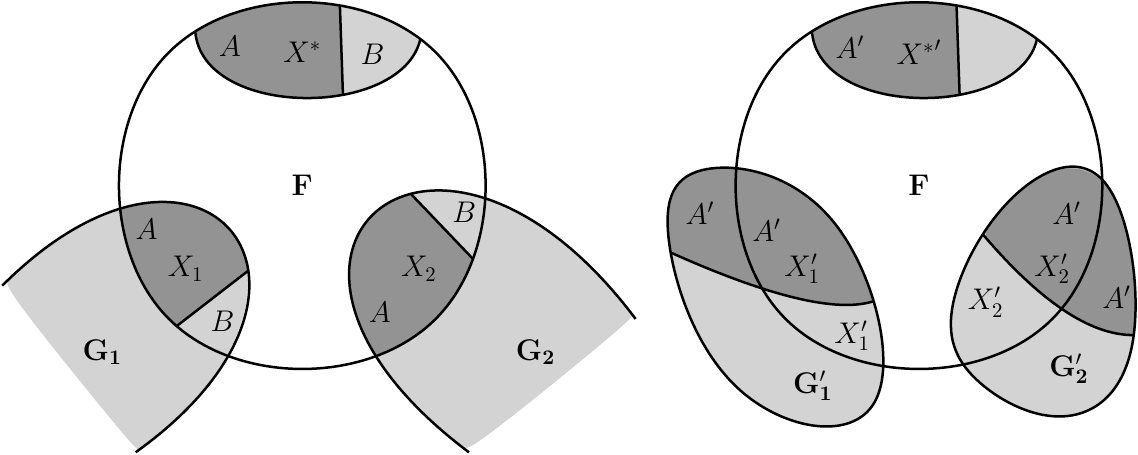}
	\caption{Illustration of the setting of the nice problem and reduction.
		The shaded area on the left is ${\bf G}$
		where $X = X_1 \cup X_2 \cup X^{\star}$,
		and the shaded area on the right is ${\bf G'}$
		where $X' = X_1' \cup X_2' \cup {X^{\star}}'$.}
	\label{fig-nice-prob}
\end{figure}

All the definitions of this section are naturally generalizable to graphs with weights on the vertices and/or edges.
Given such a weight function $w$, we extend $f(G,\Pcal)$, $\p_{f\opt}(G)$, $\hat{\p}_{f\opt}(G,\Xcal)$, $({\bf G},\Acal)$, and $({\bf G}',\Acal',s')$ to $f(G,\Pcal,w)$, $\p_{f\opt}(G,w)$, $\hat{\p}_{f\opt}(G,\Xcal,w)$, $({\bf G},\Acal,w)$, and $({\bf G}',\Acal',s',w')$, respectively.

\subsection{General dynamic programming scheme}\label{sec:nice:dp}

We now have all the ingredients for our general dynamic programming scheme on bipartite tree decompositions.
We essentially prove that if a problem $\Pi$ has an annotated extension that is $\Bcal$-nice and solvable in \FPT-time parameterized by \oct, then $\Pi$ is solvable in \FPT-time parameterized by \btw.
This actually holds for more general $\Hcal$.

\begin{lemma}\label{DP}
Let $p\in\bN$.
Let $\Hcal$ be a hereditary graph class.
Let $\Pi$ be an optimization problem.
Let $\Pi'$ be a problem that is:
\begin{itemize}
\item a $p$-\annotatedextension of $\Pi$ corresponding to some choice of $p$-partition-evaluation function $g$ and some $\mathsf{opt}\in\{\max,\min\}$,
\item $\Hcal$-nice, and
\item solvable, given instances $(G,\Xcal)$ such that $G\setminus{\cup\Xcal}\in\Hcal$, in time $f(|{\cup\Xcal}|)\cdot n^c\cdot m^d$, for some $c,d\in\bN$.
\end{itemize}
Then, there is an algorithm that, given a graph $G$ and a 1-$\Hcal$-tree decomposition of $G$ of width $k$, computes $\p_{f,\opt}(G)$ in time $\Ocal(p^k\cdot f(k+\Ocal(1))\cdot (k\cdot n)^c\cdot m^d)$ (or $\Ocal(p^k\cdot f(k+\Ocal(1))\cdot(m+k^2\cdot n)^d)$ if $c=0$).
\end{lemma}

\begin{proof}
Let {\sf Alg} be the algorithm that solves instances $(G,\Xcal)$ such that $G\setminus\cup \Xcal\in\Hcal$ in time $f(|\cup \Xcal|)\cdot n^c\cdot m^d$.

Let $(T,\alpha,\beta,r)$ be a rooted $1$-$\Hcal$-tree decomposition of $G$ of width at most $k$.
Let $\sigma:V(G)\to\bN$ be an injection.
For $t\in V(T)$, we set the following.
Let ${\bf G}_t =(G_t,\delta_t,\sigma_{|\delta_t})$ be a boundaried graph whose underlying graph $G_t$ contains the vertices of $G$ appearing below $t$ in the tree decomposition and whose boundary is the adhesion $\delta_t$ of $t$ with its parent.
Let $A_t=\alpha(t)\cup\delta_t$ be the vertices of $\alpha(t)$ plus the vertex in $\delta_t\setminus \alpha(t)$, if it exists.
Let $X_t$ be the union of $A_t$ and the adhesions $\delta_{t'}$ for $t'\in\ch_r(t)$ and let $B_t=X_t\setminus A_t$, that is, containing the vertex in $\delta_{t'}\setminus \alpha(t)$ for each $t'\in\ch_r(t)$, if it exists.
Let also ${\bf X_t}=(G[X_t],X_t,\sigma_{|X_t})$ be a trivial boundaried graph with underlying graph $X_t$ and ${\bf H}_t={\bf X}_t\boxplus(\boxplus_{t'\in\ch_r(t)}{\bf G}_{t'})$ (as in \autoref{fig_collection}).
Let also ${\bf F}_t$ be such that $G_t={\bf F}_t\oplus{\bf H}_t$, that is, essentially, the graph induced by the bag at node $t$ with boundary $X_t$.
Recall that $|\bd({\bf G}_{t'})\setminus A_t|\leq 1$ for $t'\in\ch_r(t)$.

We proceed in a bottom-up manner to compute $\ph(G_t,\Xcal)$, for each $t\in V(T)$,
for each $\Xcal\in\Pcal_p(\delta_t)$.
Hence, given that $\delta_r=\emptyset$, it implies that $\ph(G_r,\emptyset)=\p_{g,\opt}(G)$.

We fix $t\in V(T)$.
By induction, for each $t'\in\ch_r(t)$ and for each $\Xcal_{t'}\in\Pcal_p(\delta_{t'})$, we compute the value $\ph(G_{t'},\Xcal_{t'})$.
Let $\Xcal\in\Pcal_p(\delta_t)$,
let $\Qcal$ be the set of all $\Acal\in\Pcal_p(A_t)$ such that $\Acal\cap\delta_t=\Xcal$, and
let $\Acal\in\Qcal$.
Since $\Pi'$ is $\Hcal$-nice, there is an $\Hcal$-nice reduction $({\bf H}_\Acal,\Acal',s_\Acal)$ of $({\bf H}_t,\Acal)$ with respect to $\Pi'$.
Hence, $\ph(G_t,\Acal)=\ph({\bf H}_\Acal\triangleright{\bf F}_t,\Acal')+s_\Acal$ by item (iii).
Let us compute $\ph({\bf H}_\Acal\triangleright{\bf F}_t,\Acal')$.

Given that $(T,\alpha,\beta)$ is a 1-$\Hcal$-tree decomposition, we have $F_t\setminus \alpha(t)=G[\beta(t)]\in\Hcal$, where $F_t$ is the underlying graph of $\bf F_t$. Thus, given that $\Hcal$ is hereditary, we have $F_t\setminus A_t\in\Hcal$.
Therefore, by item (iv) of the definition of an $\Hcal$-reduction, $({\bf H}_\Acal\triangleright{\bf F}_t)\setminus (\cup\Acal')\in\Hcal$.
Hence, we can compute $\ph({\bf H}_\Acal\triangleright{\bf F}_t,\Acal')$, and thus $\ph(G_t,\Acal)$, using {\sf Alg} on the instance $({\bf H}_\Acal\triangleright{\bf F}_t,\Acal')$.
Finally, $\ph(G_t,\Xcal)=\opt_{\Acal\in\Qcal}\ph(G_t,\Acal)$.

It remains to calculate the complexity.
Throughout, we make use of the fact that $p$ is a fixed constant.
We can assume that $T$ has at most $n$ nodes: for any pair of nodes $t$ and $t'$ with $(\alpha\cup\beta)(t)\subseteq(\alpha\cup\beta)(t')$, we can contract the edge $tt'$ of $T$ to a new vertex $t''$ with $\alpha(t'')=\alpha(t')$ and $\beta(t'')=\beta(t')$.
This defines a valid 1-$\Hcal$-tree decomposition of the same width.
For any leaf $t$ of $T$, there is a vertex $u\in V(G)$ that only belongs to the bag of $t$.
From this observation, we can inductively associate each node of $T$ to a distinct vertex of $G$.
Hence, if $c_t=|\ch_r(t)|$, then we have $\sum_{t\in V(T)}c_t\leq n$.
Let also $n_t=|(\alpha\cup\beta)(t)|$ and $m_t=|E(G[(\alpha\cup\beta)(t)])|$.
Note that $|A_t| = |\alpha(t)|+|\delta_t\cap\beta(t)|\leq k+1$ and that $|B_t|=|\bigcup_{t'\in V(T)}\delta_{t'}\cap\beta(t)|\leq c_t$, so $|X_t|\leq k+1+c_t$.
Moreover, the properties of the tree decompositions imply that the vertices in $\beta(t)\setminus X_t$ are only present in node $t$.
Then,
$\sum_{t\in V(T)}n_t=\sum_{t\in V(T)}(|X_t|+|\beta(t)\setminus X_t|)=\Ocal(k\cdot n)$.
Also, let $\bar m_t$ be the number of edges only present in the bag of node $t$.
The edges that are present in several bags are those in the adhesion of $t$ and its neighbors.
$t$ is adjacent to its $|c_t|$ children and its parent, and an adhesion has size at most $k+1$.
Thus,
$\sum_{t\in V(T)}m_t\leq\sum_{t\in V(T)}({\bar m_t}+k^2(1+c_t))=\Ocal(m+k^2\cdot n)$.

There are $p^{|A_t|}\leq p^{k+1}=\Ocal(p^k)$ partitions of $\Pcal_p(A_t)$.
For each of them, we compute in time $\Ocal(k\cdot c_t)$ an $\Hcal$-nice reduction $({\bf H}_\Acal,\Acal',s_\Acal)$ with $|\cup \Acal'|=|A_t|+\Ocal(1)=k+\Ocal(1)$ (by item (i)) and with $\Ocal(|B_t|)=\Ocal(c_t)$ additional vertices and edges (by item (ii)).
We thus solve $\Pi'$ on $({\bf H}_\Acal\triangleright{\bf F}_t,\Acal')$ in time
$f(k+\Ocal(1))\cdot \Ocal((n_t+c_t)^c\cdot(m_t+c_t)^d)$.
Hence, the running time is $\Ocal(p^k\cdot f(k+\Ocal(1))\cdot (k\cdot n)^c\cdot m^d)$ (or $\Ocal(p^k\cdot f(k+\Ocal(1))\cdot(m+k^2\cdot n)^d)$ if $c=0$).
 \end{proof}

\subsection{Generalizations}

For the sake of simplicity, we assumed in \autoref{DP} that the problem $\Pi$ under consideration takes as input just a graph.
However, a similar statement still holds if we add labels/weights on the vertices/edges of the input graph.
This is in particular the case for {\sc Weighted Independent Set} (\autoref{sec-mis}) and {\sc Maximum Weighted Cut} (\autoref{sec-cut}) where the vertices or edges are weighted. Furthermore, while we omit the proof here, with some minor changes to the definition of a nice problem, a similar statement would also hold for $q$(-torso)-$\Hcal$-treewidth.\medskip

Moreover, again for the sake of simplicity, we assumed that $\Pi'$ is solvable in \FPT-time, while other complexities such as \XP-time could be considered.
Similarly, in the definition of the nice reduction, the constraints $|A'|=|A|+\Ocal(1)$, $|V(G')| \le |X| + \Ocal(|B|)$, $|E(G')| \le |E(G[X])| + \Ocal(|B|)$ can be modified.
For instance, we could change the first constraint to $|A'|=\Ocal(1)$.
In those cases, the dynamic programming algorithm still applies, but the running time of \autoref{DP} is different.\medskip

To give a precise running time for {\sc $K_t$-Subgraph-Cover} (\autoref{sec-kt}), {\sc Weighted Independent Set} (\autoref{sec-mis}), and {\sc Maximum Weighted Cut} (\autoref{sec-cut}) below,
let us observe that, if $\Pi'$ is solvable in time $f(|\cup \Xcal|)\cdot (n')^c\cdot (m')^d$, where $G'=G\setminus\cup \Xcal$, $n'=|V(G')|$, and $m'=|E(G')|$, then the running time of \autoref{DP} is better.
Indeed, in the proof of the complexity of \autoref{DP},
we now solve $\Pi'$ on $({\bf H}_\Acal\triangleright{\bf F},\Acal')$ in time
$f(k+\Ocal(1))\cdot \Ocal((n'_t+c_t)^c\cdot(m'_t+c_t)^d)$, where $n'_t=|\beta(t)|$ and $m'_t=|E(G[\beta(t)])|$.
We have
$\sum_{t\in V(T)}n'_t=\sum_{t\in V(T)}(|B|+|\beta(t)\cap\delta_t|+|\beta(t)\setminus X|)=\Ocal(n)$ and
$\sum_{t\in V(T)}m'_t\leq m$.
Hence, the total running time is $\Ocal(p^k\cdot (k\cdot n+ f(k+\Ocal(1))\cdot n^c\cdot m^d))$.

\subsection{Applications}
\label{sec-results-FPT}

We now apply the above framework to give \FPT-algorithms for several problems parameterized by bipartite treewidth,
that is, $1$-$\Bcal$-treewidth where $\Bcal$ is the class of bipartite graphs.
Thanks to \autoref{DP},
this now reverts to showing that the problem under consideration has a $\Bcal$-nice annotated extension
that is solvable in \FPT\xspace time when parameterized by \oct.
Several of the presented results actually hold for other graph classes $\Hcal$, not necessarily only bipartite graphs.

All of the problems of this section have the following property, that seems critical to show that a problem is $\Hcal$-nice.

\subparagraph{Gluing property.}
Let $\Pi$ be a $p$-\annotated problem corresponding to some choice of $p$-partition-evaluation function $f$ and some $\mathsf{opt}\in\{\max,\min\}$.
We say that $\Pi$ has the \emph{gluing property} if, given two compatible boundaried graphs ${\bf F}=(F,X,\rho)$ and ${\bf G}=(G,X,\rho)$, $\Xcal\in\Pcal_p(X)$, and
$\Pcal\in\Pcal_p(V({\bf F}\oplus{\bf G}))$  such that $\Xcal\subseteq\Pcal$, then
$\phf({\bf F}\oplus{\bf G},\Xcal)=f({\bf F}\oplus{\bf G},\Pcal)$ if and only if $\phf(F,\Xcal)=f(F,\Pcal\cap V(F))$ and $\phf(G,\Xcal)=f(G,\Pcal\cap V(G))$, where $F$ and $G$ are the underlying graphs of $\bf F$ and $\bf G$ (see \autoref{fig_gluing_property}).
\medskip

\begin{figure}[h]
\center
\includegraphics[scale=0.7]{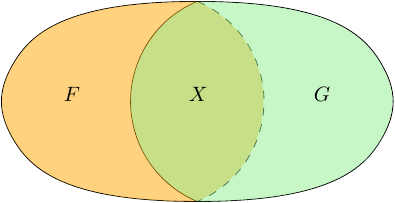}
\caption{If a problem $\Pi$ has the gluing property, then a solution is optimal on ${\bf F}\oplus{\bf G}$ if and only if its restriction to $F$ and its restriction to $G$ are both optimal.}
\label{fig_gluing_property}
\end{figure}

For the sake of simplicity, with a slight abuse of notation, we identify in this section a vertex with its heir.
\medskip

Let $\Pi'$ be an annotated extension of some problem $\Pi$.
Given an instance $({\bf G}={\bf X}\boxplus(\boxplus_{i\in[d]}{\bf G}_i),(A,B),\Acal)$ for a $\Bcal$-nice reduction with respect to $\Pi'$, we know that the boundary of each $G_i$ contains at most one vertex of $B$, and hence which is not annotated (the annotated vertices are those in $A=\cup\Acal$).
To show that $\Pi'$ is $\Bcal$-nice, we thus essentially need to show how to reduce a graph ${\bf F}\oplus{\bf G}$ to a graph $F'$ when the boundary of ${\bf F}$ and ${\bf G}$ is totally annotated (and that $\Pi'$ has the gluing property), and when the boundary of ${\bf F}$ and ${\bf G}$ has a single vertex $v$ that is not annotated.
To show that $\Pi$ is \FPT\ parameterized by \btw, it then suffices to prove that $\Pi'$ is \FPT\ parameterized by \oct\ on instances where a minimal odd cycle transversal is annotated.

\subsubsection{$K_t$-Subgraph-Cover}\label{sec-kt}

Let $\Gcal$ be a graph class.
The problem \PbCov\ is defined as follows.

\begin{center}
	\fbox{
		\begin{minipage}{11.5cm}
			\noindent{{\sc (Weighted)} \PbCov}\\
			\noindent\textbf{Input}:~~A graph $G$ (and a weight function $w:V(G)\to\bN$).\\
			\textbf{Objective}:~~Find the set $S\subseteq V(G)$ of minimum size (resp. weight) such \phantom{\textbf{Objective}:~~} that $G\setminus S\in\Gcal$.
		\end{minipage}
	}
\end{center}

If $\Gcal$ is the class of edgeless (resp. acyclic, planar, bipartite, (proper) interval, chordal) graphs, then we obtain the {\textsc{Vertex Cover}} (resp. {\textsc{Feedback Vertex Set}}, {\textsc{Vertex Planarization}}, {\textsc{Odd Cycle Transversal}}, {\sc (proper) Interval Vertex Deletion}, {\sc Chordal Vertex Deletion}) problem.
Also, given a graph $H$, if $\Gcal$ is the class of graphs that do not contain $H$ as a subgraph (resp. a minor/odd-minor/induced subgraph), then the corresponding problem is called {\sc $H$-Subgraph-Cover} (resp. {\sc $H$-Minor-Cover}/{\sc $H$-Odd-Minor-Cover}/{\sc $H$-Induced-Subgraph-Cover}).\medskip

Let $H$ be a graph and $w:V(G)\to\bN$ be a weight function (assigning one to every vertex in the unweighted case).
We define $f_ {H}$  as the $2$-partition-evaluation function where, for every graph $G$, for every $(R,S)\in\Pcal_2(V(G))$,
\begin{equation*}
   f_ {H}(G,(R,S))=
   \begin{cases}
     +\infty & \text{if } H \text{ is a subgraph of } G\setminus S, \\
     w(S) &  \text{otherwise.}
   \end{cases}
\end{equation*}

Seen as an optimization problem, {\sc (Weighted) $H$-Subgraph-Cover} is the problem of computing $\p_{f_ {H},\min}(G)$.
We call its annotated extension {\sc (Weighted) \Annotated $H$-Subgraph-Cover}.
In other words, {\sc (Weighted) \Annotated $H$-Subgraph-Cover} is defined as follows.

\begin{center}
	\fbox{
		\begin{minipage}{11.5cm}
			\noindent{{\sc (Weighted) \Annotated $H$-Subgraph-Cover}}\\
			\noindent\textbf{Input}:~~A graph $G$, two disjoint sets $R,S\subseteq V(G)$ (and a weight function $w:~V(G)\to\bN$).\\
			\textbf{Objective}:~~Find, if it exists, the minimum size (resp. weight) of a set $S^\star \subseteq V(G)$ such that $R\cap S^\star=\emptyset$, $S\subseteq S^\star$, and $G\setminus S^\star$ does not contain $H$ as a subgraph.
		\end{minipage}
	}
\end{center}

Given a graph $G$ and a set $X\subseteq V(G)$, note that $X$ is a vertex cover if and only if $V(G)\setminus X$ is an independent set.
Hence, the size/weight of a minimum vertex cover is equal to the size/weight of a maximum independent set.
Thus, seen as optimization problems, {\sc (Weighted) Vertex Cover} and {\sc (Weighted) Independent Set} are equivalent problems.\medskip

In order to prove that {\sc (Weighted) \Annotated $K_t$-Subgraph-Cover} is a nice problem, we first prove that {\sc (Weighted) \Annotated $K_t$-Subgraph-Cover} has the gluing property.

\begin{lemma}[Gluing property]\label{glu-kt}
{\sc (Weighted) \Annotated $K_t$-Subgraph-Cover} has the gluing property.
More precisely, given two boundaried graphs ${\bf F}=(F,B_F,\rho_F)$ and ${\bf G}=(G,B_G,\rho_G)$, a weight function $w:V({\bf F}\oplus{\bf G})\to\bN$, a set $X\subseteq V({\bf F}\oplus{\bf G})$ such that $B_F\cap B_G\subseteq X$, and $\Xcal=(R,S)\in\Pcal_2(X)$, we have $$\hat{\p}_{f_{K_t},\min}({\bf F}\oplus{\bf G},\Xcal,w)=\hat{\p}_{f_{K_t},\min}(F,\Xcal\cap V(F),w)+\hat{\p}_{f_{K_t},\min}(G,\Xcal\cap V(G),w)-\bar{w},$$ where $\bar{w}=w(S\cap B_F\cap B_G)$.
\end{lemma}

\begin{proof}
Observe that $K_t$ is a subgraph of ${\bf F}\oplus{\bf G}$ if and only if $K_t$ is a subgraph of $F$ or of $G$.

Let $\Pcal=(R^\star,S^\star)\in\Pcal_2(V({\bf F}\oplus{\bf G}))$ be such that $\Xcal\subseteq\Pcal$ and $\hat{\p}_{f_{K_t},\min}({\bf F}\oplus{\bf G},\Xcal,w)=f_{K_T}({\bf F}\oplus{\bf G},\Pcal,w).$
Then $K_t$ is neither a subgraph of $F\setminus (S^\star\cap V(F))$ nor of $G\setminus(S^\star\cap V(G))$.
Therefore,
\begin{align*}
\hat{\p}_{f_{K_T},\min}({\bf F}\oplus{\bf G},&\Xcal,w)=w(S^\star)\\
&=w(S^\star\cap V(F))+w(S^\star\cap V(G))-w(S^\star\cap B_F\cap B_G)\\
&\geq \hat{\p}_{f_{K_t},\min}(F,\Xcal\cap V(F),w)+\hat{\p}_{f_{K_t},\min}(G,\Xcal\cap V(G),w)-\bar{w}.
\end{align*}

Reciprocally, let $\Pcal_H=(R_H,S_H)\in\Pcal_2(V(H))$ be such that $\Xcal\cap V(H)\subseteq\Pcal_H$ and $\hat{\p}_{f_{K_t},\min}(H,\Xcal\cap V(H))=f_{K_t}(H,\Pcal_H)$ for $H\in\{F,G\}$.
Then $K_t$ is not a subgraph of $({\bf F}\oplus{\bf G})\setminus (S_F\cup S_G)$, so
\begin{align*}
\hat{\p}_{f_{K_t},\min}({\bf F}\oplus{\bf G},&\Xcal,w)\leq w(S_F\cup S_G)\\
&=w(S_F)+w(S_G)-\bar{w}\\
&=\hat{\p}_{f_{K_t},\min}(F,\Xcal\cap V(F),w)+\hat{\p}_{f_{K_t},\min}(G,\Xcal\cap V(G),w)-\bar{w}.
\end{align*}
 \end{proof}

The main obstacle to find an \FPT-algorithm parameterized by $(1,\Hcal)$-\tw\ for {\sc (Weighted) \Annotated $H$-Subgraph-Cover}, for $H$ that is not a clique, is the fact that the problem does not have the gluing property.

\begin{lemma}\label{no-glu-h}
If $H$ is not a complete graph, then
{\sc (Weighted) \Annotated $H$-Subgraph-Cover} does not have the gluing property.
\end{lemma}

\begin{proof}
Since $H$ is not a clique, there are two vertices $u,v\in V(H)$ that are not adjacent.
Let $V'=V(H)\setminus\{u,v\}$ and
let $\sigma:V'\to\bN$ be an injection.
Let ${\bf F}=(H\setminus\{u\}, V',\sigma)$ and ${\bf G}=(H\setminus\{v\},V',\sigma)$.
Then ${\bf F}\oplus{\bf G}$ is isomorphic to $H$.
Let $\Xcal=(V',\emptyset)$ and $\Pcal=(H\setminus\{u\},\{u\})\supseteq\Xcal$.
Then we have $\hat{\p}_{f_H,\min}({\bf F}\oplus{\bf G},\Xcal)=f_H({\bf F}\oplus{\bf G},\Pcal)=1$.
However, $\hat{\p}_{f_H,\min}(G,\Xcal\cap V(G))=0<f_H(G,\Pcal\cap V(G))=1$.
 \end{proof}

We now show how to reduce a graph ${\bf F}\oplus{\bf G}$ when the boundary of ${\bf F}$ and ${\bf G}$ has a single vertex $v$ that is not annotated.
Essentially, given an annotation $(R,S)$, we compare the size $s^+$ of an optimal solution when $v$ is added to $S$ and the size $s^-$ of an optimal solution when $v$ is added to $R$.
If $s^+\le s^-$, then it is always better too add $v$ to $S$, so we do so and delete $V(G)\setminus V(F)$ from ${\bf F}\oplus{\bf G}$.
Otherwise, we still delete $V(G)\setminus V(F)$ from ${\bf F}\oplus{\bf G}$, but there is no need to annotate $v$, whose addition to $S$ or not will be determined by its behaviour in $F$.
See \autoref{fig_subgraph_cover} for an illustration.

\begin{figure}[h]
\center
\includegraphics[scale=0.7]{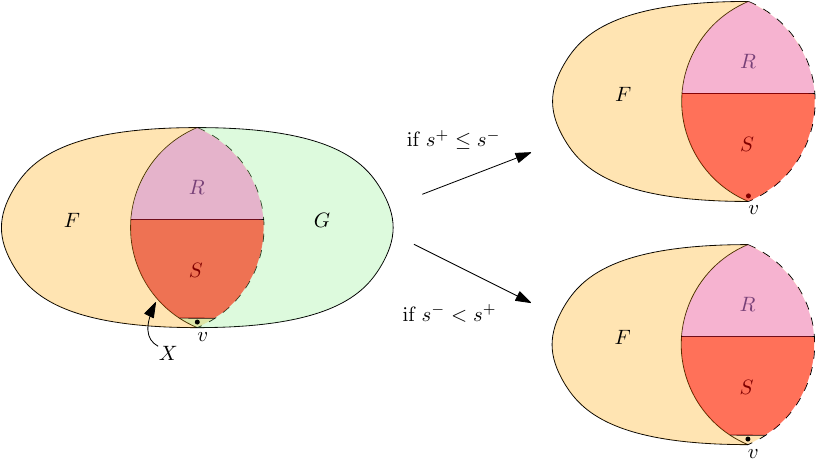}
\caption{Illustration of the gadgetization for {\sc $K_t$-Subgraph-Cover}.}
\label{fig_subgraph_cover}
\end{figure}

\begin{lemma}[Gadgetization]\label{obs+1}
Let ${\bf F}=(F,B_F,\rho_F)$ and ${\bf G}=(G,B_G,\rho_G)$ be two boundaried graphs.
Let $X\subseteq V({\bf F}\oplus{\bf G})$ be such that $B_F\cap B_G\subseteq X$.
Let also $v\in B_F\cap B_G$ and let $\Xcal=(R,S)\in\Pcal_2(X\setminus\{v\})$.
We define $\Xcal^+=(R,S\cup\{v\})$ and $\Xcal^-=(R\cup\{v\},S)$.
Furthermore, for $a\in\{+,-\}$, we set $s^a=\hat{\p}_{f_{K_t},\min}(G,\Xcal^a\cap V(G))$, and we set $\bar{s}=|S\cap B_F\cap B_G|$.
Then, for any $t\in\bN$,
\begin{equation*}
   \hat{\p}_{f_{K_t},\min}({\bf F}\oplus{\bf G},\Xcal)=
   \begin{cases}
     s^++\hat{\p}_{f_{K_t},\min}(F,\Xcal^+\cap V(F))-\bar{s}-1 & \text{if } s^+\leq s^-, \\
     s^-+\hat{\p}_{f_{K_t},\min}(F,\Xcal\cap V(F))-\bar{s} &  \text{otherwise.}
   \end{cases}
\end{equation*}
\end{lemma}

\begin{proof}
For $H\in\{F,G\}$ and for $a\in\{+,-\}$, let $S^a_H\subseteq V(H)$ be such that $\Xcal^a\cap V(H)\subseteq(V(H)\setminus S^a_H,S^a_H)$ and $\hat{\p}_{f_{K_t},\min}(H,\Xcal^a\cap V(H))=f_{K_t}(H,(V(H)\setminus S^a_H,S^a_H))=|S^a_H|$.
We deduce that $s^+=|S_G^+|$ and $s^-=|S_G^-|$.
Let us note similarly $t^+=|S_F^+|$ and $t^-=|S_F^-|.$

Note that $S^+_F\cap S^+_G=(S\cap B_F\cap B_G)\cup\{v\}$ and $S^-_F\cap S^-_G=S\cap B_F\cap B_G$.
Hence, using \autoref{glu-kt}, we have that
\begin{align*}
\hat{\p}_{f_{K_t},\min}({\bf F}\oplus{\bf G},\Xcal)
&=\min\{\hat{\p}_{f_{K_t},\min}({\bf F}\oplus{\bf G},\Xcal^+),\hat{\p}_{f_{K_t},\min}({\bf F}\oplus{\bf G},\Xcal^-)\}\\
&=\min\{t^++s^+-1,t^-+s^-\}-\bar{s}.
\end{align*}

Note that we always have $s^+\leq s^-+1$ since $G\setminus (S_G^-\cup\{v\})$ does not contain $K_t$ as a subgraph, and thus $|S_G^-\cup\{v\}|\geq p_{f_{K_t},\min}(G,\Xcal^+\cap V(G))$.
Similarly, $t^+\leq t^-+1$.

Thus, if $s^+\leq s^-$, then $t^++s^+-1\leq t^-+s^+\leq t^-+s^-.$
Given that $t^+=\hat{\p}_{f_{K_t},\min}(F,\Xcal^+\cap V(F))$, it follows that
$$\hat{\p}_{f_{K_t},\min}({\bf F}\oplus{\bf G},\Xcal)=\hat{\p}_{f_{K_t},\min}(F,\Xcal^+\cap V(F))+s^+-\bar{s}-1.$$
And if $s^-<s^+$, then $s^+=s^-+1$, so $\min\{t^++s^+-1,t^-+s^-\}=\min\{t^+,t^-\}+s^-=\hat{\p}_{f_{K_t},\min}(F,\Xcal\cap V(F))+s^-$.
It follows that
$$\hat{\p}_{f_{K_t},\min}({\bf F}\oplus{\bf G},\Xcal)=\hat{\p}_{f_{K_t},\min}(F,\Xcal\cap V(F))+s^--\bar{s}.$$
 \end{proof}

Contrary to \autoref{glu-kt}, observe that \autoref{obs+1} only holds in the unweighted case.
Indeed, in the weighted case, we now have $s^+\leq s^-+w(v)$, and thus, when
$s^+\in[s^-+1,s^-+w(v)-1]$, we do not know what happens.

Using \autoref{glu-kt} and \autoref{obs+1}, we can now prove that {\sc \Annotated $K_t$-Subgraph-Cover} is $\Hcal$-nice.
Essentially, given an instance $({\bf G}={\bf X}\boxplus(\boxplus_{i\in[d]}{\bf G}_i),(A,B),(R,S))$,
we reduce ${\bf G}$ to ${\bf X}$ and further remove some vertices of $B$ that can be optimally added to $S$, and show that the resulting boundaried graph is equivalent to $\bf G$ modulo some constant $s$.

\begin{lemma}[Nice problem]\label{ktvc-nice}
Let $\Hcal$ be a hereditary graph class.
Let $t\in\bN$.
{\sc \Annotated $K_t$-Subgraph-Cover} is $\Hcal$-nice.
\end{lemma}

\begin{proof}
Let ${\bf G}=(G,X,ρ)$ be a boundaried  graph, 	
let ${\bf X}=(G[X],X,\rho_X)$ be a trivial boundaried graph and
let $\{{\bf G}_i=(G_{i},X_{i},ρ_{i})\mid i\in[d]\}$ be a collection of boundaried graphs,
such that ${\bf G}={\bf X}\boxplus(\boxplus_{i\in[d]}{\bf G}_i)$,
let $(A,B)$ be a partition of $X$ such that
for all $i \in [d]$,
$|X_i\setminus A|\leq 1$,
and
let $\Acal=(R,S)\in\Pcal_2(A)$.
Suppose that we know,
for every $i\in[d]$ and each $\mathcal{X}_i \in \Pcal_2(X_i)$, the value $\hat{\mathsf{p}}_{f_{K_t},\mathsf{min}}(G_i,\mathcal{X}_i)$.

Let us show that item (iii) holds (recall that $\Hcal$-niceness is defined in \autoref{sec:nice}).
Let $({\bf H}_0,S_0,s_0)=({\bf G},S,0)$.
For $i$ going from 1 up to $d$, we construct $({\bf H}_i,S_i,s_i)$ from $({\bf H}_{i-1},S_{i-1},s_{i-1})$
such that for any boundaried graph $\bf F$ compatible with $\bf G$,
$$\hat{\p}_{f_{K_t},\min}({\bf G}\oplus{\bf F},\Acal)=\hat{\p}_{f_{K_t},\min}({\bf H}_i\oplus{\bf F},\Acal_i)+s_i,$$
where $\Acal_i=(R,S_i)$.
This is obviously true for $i=0$.

Let $i\in[d]$.
Let ${\bf H}_i$ be the boundaried graph such that ${\bf H}_{i-1}={\bf H}_i\boxplus{\bf G}_i$.
See \autoref{fig_nice_subgraph} for an illustration.
\begin{figure}[h]
\center
\includegraphics[scale=0.8]{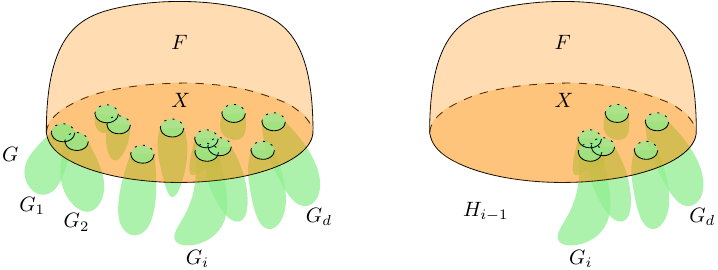}
\caption{Illustration of ${\bf G}$ and $\bf H_{i-1}$ in the proof of \autoref{ktvc-nice}.}
\label{fig_nice_subgraph}
\end{figure}
By induction,
$\hat{\p}_{f_{K_t},\min}({\bf G}\oplus{\bf F},\Acal)=
\hat{\p}_{f_{K_t},\min}({\bf H}_{i-1}\oplus{\bf F},\Acal_{i-1})+s_{i-1}$.

{Suppose first that $X_i\subseteq R\cup S_{i-1}$.}
Let $\Pcal_i=(R^i,S^i)\in\Pcal_2(X)$ be such that $\Acal_{i-1}\subseteq \Pcal_i$ and
$\hat{\p}_{f_{K_t},\min}({\bf H}_{i-1}\oplus{\bf F},\Acal_{i-1})=\hat{\p}_{f_{K_t},\min}({\bf H}_{i-1}\oplus{\bf F},\Pcal_i).$
Let $F$ be the underlying graph of $\bf F$.
According to \autoref{glu-kt},
\begin{align*}
\hat{\p}_{f_{K_t},\min}({\bf H}_{i-1}\oplus{\bf F},\Pcal_i) &= \hat{\p}_{f_{K_t},\min}(F,\Pcal_i)+\hat{\p}_{f_{K_t},\min}(H_{i-1},\Pcal_i)-|X\cap S^i| \\
        &= \hat{\p}_{f_{K_t},\min}(F,\Pcal_i)+\hat{\p}_{f_{K_t},\min}(H_i,\Pcal_i)\\
        &\ \ \ +\hat{\p}_{f_{K_t},\min}(G_i,\Pcal_i\cap X_i)-|X_i\cap S^i|-|X\cap S^i| \\
        &= \hat{\p}_{f_{K_t},\min}({\bf H}_i\oplus{\bf F},\Pcal_i)+\hat{\p}_{f_{K_t},\min}(G_i,\Acal_{i-1}\cap X_i)-|S_{i-1}\cap X_i|.
\end{align*}
Since this is the case for all such $\Pcal_i$, it implies that
\begin{align*}
\hat{\p}_{f_{K_t},\min}({\bf H}_{i-1}\oplus{\bf F},\Acal_{i-1})&=
\hat{\p}_{f_{K_t},\min}({\bf H}_i\oplus{\bf F},\Acal_{i-1})+\hat{\p}_{f_{K_t},\min}(G_i,\Acal_{i-1}\cap X_i)-|S_{i-1}\cap X_i|.
\end{align*}
Therefore, $$\hat{\p}_{f_{K_t},\min}({\bf G}\oplus{\bf F},\Acal)=\hat{\p}_{f_{K_t},\min}({\bf H}_i\oplus{\bf F},\Acal_i)+s_i,$$ where $\Acal_i=\Acal_{i-1}$ and $s_i=s_{i-1}+\hat{\p}_{f_{K_t},\min}(G_i,\Acal_{i-1}\cap X_i)-|S_{i-1}\cap X_i|$.

{Otherwise, there is $v_i\in V(G_{i-1})$ such that $X_i\setminus (R\cup S_{i-1})=\{v_i\}$.}
Let $\Xcal_i^+=(R,S_{i-1}\cup\{v_i\})\cap X_i$ and $\Xcal_i^-=(R\cup\{v_i\},S_{i-1})\cap X_i$.
Let $\Pcal_i=(R^i,S^i)\in\Pcal_2(X\setminus\{v_i\})$ be such that $\Acal_{i-1}\subseteq \Pcal_i$ and
$\hat{\p}_{f_{K_t},\min}({\bf H}_{i-1}\oplus{\bf F},\Acal_{i-1})=\hat{\p}_{f_{K_t},\min}({\bf H}_{i-1}\oplus{\bf F},\Pcal_i).$
Note that $${\bf H}_{i-1}\oplus{\bf F}=({\bf H}_i\boxplus{\bf G}_i)\oplus{\bf F}=({\bf H}_i\boxplus{\bf F})\oplus{\bf G}_i.$$
For $a\in\{+,-\}$, let $s^a_i=\p_{f_{K_t},\min}(G,\Xcal_i^a)$.
Then, using \autoref{obs+1}, we have the following case distinction.
\begin{align*}
\hat{\p}_{f_{K_t},\min}&({\bf H}_{i-1}\oplus{\bf F},\Pcal_i)
=\hat{\p}_{f_{K_t},\min}(({\bf H}_i\boxplus{\bf F})\oplus{\bf G}_i\oplus{\bf F},\Pcal_i)\\
&\hspace{-0.4cm}=
\begin{cases}
s_i^++\hat{\p}_{f_{K_t},\min}({\bf H}_i\oplus{\bf F},(R^i,S^i\cup\{v_i\}))-|S_{i-1}\cap X_i|-1 & \text{if } s_i^+\leq s_i^-\\
s_i^-+\hat{\p}_{f_{K_t},\min}({\bf H}_i\oplus{\bf F},\Pcal_i)-|S_{i-1}\cap X_i| & \text{otherwise.}
\end{cases}
\end{align*}
Since this is the case for every such $\Pcal_i$, by setting
$({\bf H}_i,S_i,s_i)=({\bf H}_i,S_{i-1}\cup\{v_i\},s_{i-1}+s^+-|S_{i-1}\cap X_i|-1)$ if $s^+\leq s^-$ and
$({\bf H}_i,S_i,s_i)=({\bf H}_i,S_{i-1},s_{i-1}+s^--|S_{i-1}\cap X_i|)$ otherwise, we have that
$$\hat{\p}_{f_{K_t},\min}({\bf G}\oplus{\bf F},\Acal)=\hat{\p}_{f_{K_t},\min}({\bf H}_i\oplus{\bf F},\Acal_i)+s_i.$$

The boundaried graph $\bf H_d$ obtained at the end is isomorphic to ${\bf X}$.
Let $S_B=S_d\setminus S\subseteq B$.
Observe that
\begin{align*}
\hat{\p}_{f_{K_t},\min}({\bf H}_d\oplus{\bf F},\Acal_d) &=\hat{\p}_{f_{K_t},\min}(({\bf H}_d\oplus{\bf F})\setminus S_B,\Acal_d\setminus S_B)+|S_B|\\
&=\hat{\p}_{f_{K_t},\min}(({\bf H}_d\setminus S_B)\triangleright{\bf F},\Acal)+|S_B|
\end{align*}
Hence, $$\hat{\p}_{f_{K_t},\min}({\bf G}\oplus{\bf F},\Acal)=\hat{\p}_{f_{K_t},\min}(({\bf X}\setminus S_B)\triangleright{\bf F},\Acal)+|S_B|+s_d.$$

Let ${\bf H}={\bf X}\setminus S_B$.
Item (i) holds trivially.
Item (ii) also holds, given that $|V(H)|\leq |X|$ and $|E(H)|\leq E(G[X])$, where $H$ is the underlying graph of $\bf H$.
Given that $S_B\subseteq B\subseteq X$, it implies that ${\bf H}\triangleright{\bf F}$ is isomorphic to ${F}\setminus S_B$.
Thus, since $\Hcal$ is hereditary, if $F\setminus A$ belongs to $\Hcal$, then so does $({\bf H}\triangleright{\bf F})\setminus A$. So item (iv) holds.
Hence, $({\bf X}\setminus S_B,\Acal,|S_B|+s_d)$ follows every conditions so that it is an $\Hcal$-nice reduction of $(G,\Acal)$  with respect to {\sc \Annotated $K_t$-Subgraph-Cover}.

At each step $i$, we compute $|S_{i-1}\cap X_i|$, and thus $s_i$, in time $\Ocal(|A|)$ (since $\hat{\mathsf{p}}_{f_{K_t},\mathsf{min}}(G_i,\mathcal{X}_i)$ is supposed to be known).
${\bf H}_i$ and $S_i$ are then constructed in time $\Ocal(1)$.
Hence, the computation takes time $\Ocal(|A|\cdot d)$, and thus, {\sc \Annotated $K_t$-Subgraph-Cover} is $\Hcal$-nice.
 \end{proof}

We now solve {\sc Weighted \Annotated $K_t$-Subgraph-Cover} for $t\geq 3$ parameterized by \oct.
Note that 
{\sc Weighted Vertex Cover} can be solved on bipartite graphs in time $\Ocal(m^{1+o(1)})$ using a flow algorithm \cite{ChenKLPGS25maxi}.
Indeed, let $G=(A,B)$ be a bipartite graph and $w:V(G)\to\bN$ be a weight function. We construct a flow network $N$ by connecting a source $s$ to each vertex in $A$ and a sink $t$ to each vertex in $B$.
We give infinite capacity to the original edges of $G$, and capacity $w(v)$ to each edge connecting a vertex $v$ and a terminal vertex.
Every $s-t$ cut in $N$ corresponds to exactly one vertex cover and every vertex cover corresponds to an $s-t$ cut. Thus a minimum cut of $N$ gives a minimum weight vertex cover of $G$.

\begin{lemma}\label{MIS-oct}
Let $t\in\bN_{\geq 2}$.
There is an algorithm that, given a graph $G$ and two disjoint sets $R,S\subseteq V(G)$ such that $G'=G\setminus (R\cup S)$ is bipartite,
solves {\sc Weighted \Annotated $K_t$-Subgraph-Cover} on $(G,R,S)$ in time $\Ocal(k^t\cdot(n'+m')+(m')^{1+o(1)})$, where $k=|R|$, $n'=|V(G')|$, and $m'=|E(G')|$.
\end{lemma}

\begin{proof}
Observe that we can assume that $S=\emptyset$, since $S^\star$ is an optimal solution for $(G,R,S)$ if and only if $S^\star\setminus S$ is an optimal solution for $(G\setminus S,R,\emptyset)$.
Thus, $G\setminus R$ is bipartite, so for any occurrence of $K_t$ contained in $G$ (as a subgraph), at most two of its vertices belong to $G\setminus R$.
Hence, enumerating the occurrences of $K_t$ takes time $\Ocal(k^t+k^{t-1}\cdot n'+k^{t-2}\cdot m')$.
If $G[R]$ contains an occurrence of $K_t$, then {\sc Weighted \Annotated $K_t$-Subgraph-Cover} has no solution.
Let us thus assume that $G[R]$ contains no $K_t$.
For each occurrences of $K_t$ in $G$ that contains $t-1$ vertices of $R$ and one vertex $v\in V(G)\setminus R$,
we add $v$ to $S^\star$ and remove $v$ from $G$, since $v$ has to be taken in the solution.
Hence, all that remains are occurrences of $K_t$ with $t-2$ vertices in $R$ and the two others in $G\setminus R$.
Let $H$ be the graph induced by the edges of the occurrences of $K_t$ in $G$ with both endpoints in $G\setminus R$.
Each edge of $H$ intersects any solution $\bar{S}$ on $G$ for {\sc Weighted \Annotated $K_t$-Subgraph-Cover}.
Hence, $\bar{S}$ is the union of $S^\star$ and a minimum weighted vertex cover $C$ of $H$.
Thus, $C$~can be computed in time $\Ocal((m')^{1+o(1)})$.
The running time of the algorithm is hence $\Ocal(k^t+k^{t-1}\cdot n'+k^{t-2}\cdot m'+(m')^{1+o(1)})$.
 \end{proof}

We apply \autoref{ktvc-nice} and \autoref{MIS-oct} to the dynamic programming algorithm of \autoref{DP} to obtain the following result.

\begin{corollary}\label{cor-Kt-cover}
Let $t\in\bN_{\geq 2}$.
Given a graph $G$ and a bipartite tree decomposition of $G$ of width $k$, there is an algorithm that solves {\sc $K_t$-Subgraph-Cover} on $G$ in time $\Ocal(2^k\cdot (k^t\cdot(n+m)+m^{1+o(1)}))$.
\end{corollary}

\subsubsection{Weighted Vertex Cover/Weighted Independent Set}\label{sec-mis}

Given that \autoref{obs+1} only holds for {\sc $K_t$-Subgraph-Cover} in the unweighted case, we propose here an analogous result that holds in the weighted case, when we restrict ourselves to $t=2$, i.e., {\sc Weighted Vertex Cover}.
We already know that {\sc Weighted Vertex Cover} has the gluing property (\autoref{glu-kt}).
We now show how to reduce a graph ${\bf F}\oplus{\bf G}$ to a graph $F'$ when the boundary of ${\bf F}$ and ${\bf G}$ has a single vertex $v$ that is not annotated. Recall that this reduction was sketched in \autoref{sec-overview}, and see \autoref{fig_weighted_cover} for an illustration.

\begin{figure}[h]
\center
\includegraphics[scale=0.7]{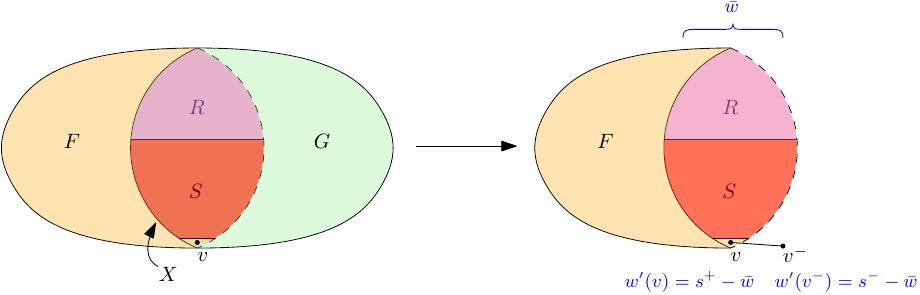}
\caption{Illustration of the gadgetization for {\sc Weighted Vertex Cover}.}
\label{fig_weighted_cover}
\end{figure}

\begin{lemma}[Gadgetization]\label{obs-mis}
Let ${\bf F}=(F,B_F,\rho_F)$ and ${\bf G}=(G,B_G,\rho_G)$ be two boundaried graphs.
Let $w:V({\bf F}\oplus{\bf G})\to\bN$ be a weight function, let $X\subseteq V({\bf F}\oplus{\bf G})$ be such that $B_F\cap B_G\subseteq X$, let $v\in B_F\cap B_G$, and let $\Xcal=(R,S)\in\Pcal_2(X\setminus\{v\})$.
We define $\Xcal^+=(R,S\cup\{v\})$ and $\Xcal^-=(R\cup\{v\},S)$.
Furthermore, for $a\in\{+,-\}$,
we set $s^a=\hat{\p}_{f_{K_t},\min}(G,\Xcal^a\cap V(G),w)$.
We also set $\bar{w}=w(S\cap B_F\cap B_G)$,
${\bf G}'=(G',\{v\},\rho_{|\{v\}})$, where $G'$ is an edge $vv^-$ for some new vertex $v^-$,
and $w':V({\bf F}\oplus{\bf G}')\to\bN$ such that $w'(v)=s^+-\bar{w}$, $w'(v^-)=s^--\bar{w}$, and $w'(x)=w(x)$ otherwise.
Then $$\hat{p}_{f{K_2},\min}({\bf F}\oplus{\bf G},\Xcal,w)=\hat{p}_{f{K_2},\min}({\bf F}\oplus{\bf G}',\Xcal,w').$$
\end{lemma}

\begin{proof}
For $a\in\{+,-\}$, let $t^a=\hat{\p}_{f_{K_t},\min}(F,\Xcal^a\cap V(F),w)$.

Note that
\begin{align*}
s'^-:=\hat{\p}_{f_{K_2},\min}(G',\Xcal^-\cap V(G'),w')&=\hat{\p}_{f_{K_2},\min}(G',(\{v\},\emptyset),w')\\
&=w'(v^-)\\
&=s^--\bar{w},\\
s'^+:=\hat{\p}_{f_{K_2},\min}(G',\Xcal^+\cap V(G'),w')&=\hat{\p}_{f_{K_2},\min}(G',(\emptyset,\{v\}),w')\\
&=w'(v)\\
&=s^+-\bar{w},\\
t'^-:=\hat{\p}_{f_{K_2},\min}(F,\Xcal^-\cap V(F),w')&=t^-, \text{ and }\\
t'^+:=\hat{\p}_{f_{K_2},\min}(F,\Xcal^+\cap V(F),w')&=t^++w'(v)-w(v).
\end{align*}

Hence, using \autoref{glu-kt}, we have that
\begin{align*}
\hat{\p}_{f_{K_2},\min}({\bf F}\oplus{\bf G},&\Xcal,w)\\
&=\min\{\hat{\p}_{f_{K_2},\min}({\bf F}\oplus{\bf G},\Xcal^+,w),\hat{\p}_{f_{K_2},\min}({\bf F}\oplus{\bf G},\Xcal^-,w)\}\\
&=\min\{t^++s^+-w(v),t^-+s^-\}-\bar{w}\\
&=\min\{t^++s'^+-w(v),t^-+s'^-\}\\
&=\min\{t'^++s'^+-w'(v),t'^-+s'^-\}\\
&=\min\{\hat{\p}_{f_{K_2},\min}({\bf F}\oplus{\bf G}',\Xcal^+,w'),\hat{\p}_{f_{K_2},\min}({\bf F}\oplus{\bf G}',\Xcal^-,w')\}\\
&=\hat{\p}_{f_{K_2},\min}({\bf F}\oplus{\bf G}',\Xcal,w').
\end{align*}
\end{proof}

Using \autoref{glu-kt} and \autoref{obs-mis}, we can now prove that {\sc Weighted \Annotated Vertex Cover} is $\Hcal$-nice.
Essentially, given an instance $({\bf G}={\bf X}\boxplus(\boxplus_{i\in[d]}{\bf G}_i),(A,B),(R,S),w)$,
we reduce ${\bf G}$ to ${\bf X}$ where we glue an edge to some vertices in $B$.
We then show that if the appropriate weight is given to each new vertex, then the resulting boundaried graph is equivalent to $\bf G$ modulo some constant $s$.

\begin{lemma}[Nice problem]\label{mis-nice}
Let $\Hcal$ be a graph class that is closed under 1-clique-sums and contains $K_2$.
Then \!{\sc Weighted Annotated Vertex Cover} is $\Hcal$-nice.
\end{lemma}

\begin{proof}
Let ${\bf G}=(G,X,ρ)$ be a boundaried  graph, 	
let $w:V(G)\to\bN$ be a weight function,
let ${\bf X}=(G[X],X,\rho_X)$ be a trivial boundaried graph and
let $\{{\bf G}_i=(G_{i},X_{i},ρ_{i})\mid i\in[d]\}$ be a collection of boundaried graphs,
such that ${\bf G}={\bf X}\boxplus(\boxplus_{i\in[d]}{\bf G}_i)$,
let $(A,B)$ be a partition of $X$ such that
for all $i \in [d]$,
$|X_i\setminus A|\leq 1$,
and
let $\Acal=(R,S)\in\Pcal_2(A)$.
Suppose that we know,
for every $i\in[d]$ and each $\mathcal{X}_i \in \Pcal_2(X_i)$, the value $\hat{\mathsf{p}}_{f_{K_t},\mathsf{min}}(G_i,\mathcal{X}_i,w)$.
Recall the $\Hcal$-niceness was defined in \autoref{sec:nice}.

Let $v_1,\ldots,v_{|B|}$ be the vertices of $B$.
For $i\in[|B|]$, let $I_i=\{j\in[d]\mid X_j\setminus A=\{v_i\}\}$.
Let $I_0=\{j\in[d]\mid X_j\subseteq A\}$.
Obviously, $(I_i)_{i\in[0,|B|]}$ is a partition of $[d]$.
Let ${\bf G}_i'=\boxplus_{j\in I_i}{\bf G}_j$.
See \autoref{fig_nice_wvc} for an illustration.
\begin{figure}[h]
\center
\includegraphics{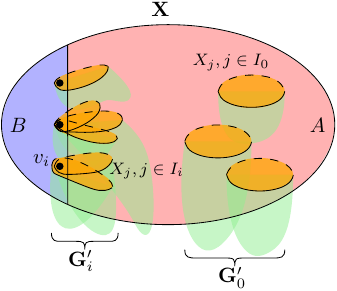}
\caption{Illustration of $G_i'$ in the proof of \autoref{mis-nice}.}
\label{fig_nice_wvc}
\end{figure}

Let us show that item (iii) holds.
Let $({\bf H_{-1}},w_{-1},s_{-1})=({\bf G},w,0)$.
For $i$ going from 0 up to $|B|]$, we will construct $({\bf H}_i,w_i,s_i)$ from $({\bf H}_{i-1},w_{i-1},s_{i-1})$
such that $(w_i)_{|V(F)\setminus\{v_j\mid j\leq i\}}=w_{|V(F)\setminus\{v_j\mid j\leq i\}}$ and $(w_i)_{|V(G_j')}=w_{|V(G_j')}$, for $j>i$, and for any boundaried graph $\bf F$ with underlying graph $F$ and compatible with $\bf G$,
$$\hat{\p}_{f_{K_2},\min}({\bf G}\oplus{\bf F},\Acal,w)=\hat{\p}_{f_{K_2},\min}({\bf H}_i\oplus{\bf F},\Acal,w_i)+s_i.$$
This is obviously true for $i=-1$.

Let $i\in[0,|B|]$.
Let ${\bf H}_i'$ be the boundaried graph with underlying graph $H_i'$
such that ${\bf H}_{i-1}={\bf H}_i'\boxplus{\bf G}_i'$.
By induction, we have
$$\hat{\p}_{f_{K_2},\min}({\bf G}\oplus{\bf F},\Acal,w)=
\hat{\p}_{f_{K_2},\min}({\bf H}_{i-1}\oplus{\bf F},\Acal,w_{i-1})+s_{i-1}.$$

\medskip
\noindent{\bf Suppose first that $i=0$.}
Let $\Pcal_0=(R^0,S^0)\in\Pcal_2(X)$ be such that $\Acal\subseteq \Pcal_0$ and
$\hat{\p}_{f_{K_2},\min}({\bf G}\oplus{\bf F},\Acal,w)=\hat{\p}_{f_{K_2},\min}({\bf G}\oplus{\bf F},\Pcal_0,w).$
According to \autoref{glu-kt},
\begin{align*}
\hat{\p}_{f_{K_2},\min}({\bf G}\oplus{\bf F},\Pcal_0,w) &= \hat{\p}_{f_{K_2},\min}(F,\Pcal_0,w)+\hat{\p}_{f_{K_2},\min}(G,\Pcal_0,w)-w(X\cap S^0)\\
&=\hat{\p}_{f_{K_2},\min}(F,\Pcal_0,w)-w(X\cap S^0)+\hat{\p}_{f_{K_2},\min}(H_0',\Pcal_0,w)\\
&\ \ \ +\sum_{j\in I_0}(\hat{\p}_{f_{K_2},\min}(G_j,\Pcal_0\cap X_j,w)-w(S^0\cap X_j))\\
&= \hat{\p}_{f_{K_2},\min}({\bf H}_0'\oplus{\bf F},\Pcal_0,w)+\sum_{j\in I_0}(\hat{\p}_{f_{K_2},\min}(G_j,\Acal\cap X_j,w)-w(S\cap X_j)).
\end{align*}
Since this is the case for all such $\Pcal_0$, it implies that
\begin{align*}
\hat{\p}_{f_{K_2},\min}({\bf G}\oplus{\bf F},\Acal,w)&=
\hat{\p}_{f_{K_2},\min}({\bf H}_0'\oplus{\bf F},\Acal,w)+\sum_{j\in I_0}(\hat{\p}_{f_{K_2},\min}(G_j,\Acal\cap X_j,w)-w(S\cap X_j)).
\end{align*}
Therefore, if ${\bf H}_0={\bf H}_0'$, $w_0=w$ and $$s_0=\sum_{j\in I_0}(\hat{\p}_{f_{K_2},\min}(G_j,\Acal\cap X_j,w)-w(S\cap X_j)),$$ then
$$\hat{\p}_{f_{K_2},\min}({\bf G}\oplus{\bf F},\Acal,w)=\hat{\p}_{f_{K_2},\min}({\bf H}_0\oplus{\bf F},\Acal,w_0)+s_0.$$

\noindent{\bf Otherwise, $i\in[|B|]$ and $X_j\setminus A=\{v_i\}$ for each $j\in I_i$.}
Let $X_i'=\bigcup_{j\in I_i}X_j$.
Let $\Xcal_i^+=(R,S\cup\{v_i\})\cap X_i'$ and $\Xcal_i^-=(R\cup\{v_i\},S)\cap X_i'$.
Let ${\bf H}_i''=(H_i'',\{v_i\},\rho_{|\{v_i\}})$ be the boundaried graph where $H_i''$ is an edge $v_iv_i^-$ for some new vertex $v_i^-$.
Let ${\bf H}_i={\bf H}_i'\boxplus{\bf H}_i''$.
Let $s_i^+=\hat{\p}_{f_{K_t},\min}(G_i',\Xcal_i^+,w)$ and $s_i^-=\hat{\p}_{f_{K_t},\min}(G_i',\Xcal_i^-,w)$.
By \autoref{glu-kt}, $$s_i^+=\sum_{j\in I_i}(\hat{\p}_{f_{K_t},\min}(G_j,\Xcal_i^+\cap X_j,w)-w(S\cap X_j)-w(v_i))+w(S\cap X_i')+w(v_i),$$
and $$s_i^-=\sum_{j\in I_i}(\hat{\p}_{f_{K_t},\min}(G_j,\Xcal_i^-\cap X_j,w)-w(S\cap X_j))+w(S\cap X_i').$$
Since the $\hat{\p}_{f_{K_t},\min}(G_j,\Xcal_i^a\cap X_j,w)$ are given,
$s_i^+$ and $s_i^-$ can be computed in time $\Ocal(|A|\cdot |I_i|)$.
Let $w_i:V({\bf F}\oplus{\bf H}_i'')\to\bN$ be such that $w_i(v_i)=s_i^+-w(S\cap B_F\cap B_G)$, $w_i(v_i^-)=s_i^--w(S\cap B_F\cap B_G)$, and $w_i(x)=w(x)$ otherwise.
Let $\Pcal_i=(R^i,S^i)\subseteq\Pcal_2(X\setminus\{v_i\})$ be such that $\Acal\subseteq \Pcal_i$ and
$\hat{\p}_{f_{K_2},\min}({\bf H}_{i-1}\oplus{\bf F},\Acal,w_{i-1})=\hat{\p}_{f_{K_2},\min}({\bf H}_{i-1}\oplus{\bf F},\Pcal_i,w_{i-1}).$
Then, using \autoref{obs-mis},
\begin{align*}
\hat{\p}_{f_{K_2},\min}({\bf H}_{i-1}\oplus{\bf F},\Pcal_i,w_{i-1})
&=\hat{\p}_{f_{K_2},\min}(({\bf H}_{i}'\boxplus{\bf F})\oplus{\bf G}_i',\Pcal_i,w_{i-1})\\
&=\hat{\p}_{f_{K_2},\min}(({\bf H}_{i}'\boxplus{\bf F})\oplus{\bf H}_i'',\Pcal_i,w_{i})\\
&=\hat{\p}_{f_{K_2},\min}({\bf H}_i\oplus{\bf F},\Pcal_i,w_{i})
\end{align*}
Since this is the case for all such $\Pcal_i$, it implies that
$$\hat{\p}_{f_{K_2},\min}({\bf H}_{i-1}\oplus{\bf F},\Acal,w_{i-1})=
\hat{\p}_{f_{K_2},\min}({\bf H}_i\oplus{\bf F},\Acal,w_{i}).$$
Therefore, given $s_i=s_{i-1}$,
$$\hat{\p}_{f_{K_2},\min}({\bf G}\oplus{\bf F},\Acal,w)=\hat{\p}_{f_{K_2},\min}({\bf H}_i\oplus{\bf F},\Acal,w_i)+s_i.$$

Observe that ${\bf H_{|B|}}={\bf X}\boxplus(\boxplus_{i\in[|B|]}{\bf H}_i'')$ and ${\bf H_{|B|}}\oplus{\bf F}={\bf F}\oplus(\boxplus_{i\in[|B|]}{\bf H}_i'')$.
Suppose that $F\setminus A\in\Hcal$.
Given that $\Hcal$ is closed under 1-clique-sums and contains edges, that each $H_i''$ is an edge, and that $|\bd({\bf H}_i'')|=1$,
it follows that $({\bf H_{|B|}}\oplus{\bf F})\setminus A\in\Hcal$.
So item (iv) holds.
Moreover, item (i) trivially holds and item (ii) holds given that
$|V(H_{|B|})|=|X|+|B|$ and $|E(H_{|B|})|=|E(G[X])|+|B|$.
Hence, $({\bf H}_{|B|},\Acal,s_{|B|},w_{|B|})$ is an $\Hcal$-nice reduction of $({\bf G},\Acal,w)$  with respect to {\sc Weighted \Annotated Vertex Cover}.

At each step $i$, $s_i$ is computable in time $\Ocal(|A|\cdot |I_i|)$, and ${\bf H}_{i}$ and $w_{i}$ are computable in time $\Ocal(1)$.
Hence, the computation takes time $\Ocal(|A|\cdot d)$.
Therefore, {\sc Weighted \Annotated Vertex Cover} is $\Hcal$-nice.
 \end{proof}
 
We now solve {\sc Weighted Vertex Cover/Independent Set} parameterized by \oct with a better running time than \autoref{MIS-oct}.

\begin{observation}\label{obs-mis-oct}
Let $\Hcal$ be a hereditary graph class such that {\sc Weighted Vertex Cover} can be solved on instances $(G,w)$ where $G\in\Hcal$ in time $\Ocal(n^c\cdot m^d)$ for some $c,d\in\bN$.
Then
{\sc Weighted \Annotated Vertex Cover} is solvable on instance $(G,R,S,w)$ such that $G'=G\setminus(R\cup S)\in\Hcal$ in time $\Ocal(k\cdot(k+ n')+n'^c\cdot m'^d)$, where $n'=|V(G')|$, $m'=|E(G')|$, and $k=|R|$.
\end{observation}

\begin{proof}
Let $G$ be a graph, $w$ be a weight function, and $R,S\subseteq V(G)$ be two disjoint sets such that $G\setminus(R\cup S)\in\Hcal$.
If $R$ is not an independent set, then $(G,R,S,w)$ has no solution.
Hence, we assume that $R$ is an independent set.
Then $S^\star\subseteq V(G)$ is a solution of {\sc Weighted \Annotated Vertex Cover} on $(G,R,S,w)$ if and only if $S^\star=S_B\cup S\cup N_G(R)$ where $S_B$ is a solution of {\sc Weighted Vertex Cover} on $(G\setminus(R\cup S),w)$.
Checking that $R$ is an independent set takes time $\Ocal(k^2)$ and then finding $N_G(R)$ takes time $\Ocal(k\cdot n')$, hence the result.
 \end{proof}

We apply \autoref{mis-nice} and \autoref{obs-mis-oct} to the dynamic programming algorithm of \autoref{DP} to obtain the following result.
Note that any graph class that is hereditary and does not contain $K_2$ is edgeless. Given that the problem is trivial on edgeless graphs, we can remove the condition of containing $K_2$.

\begin{corollary}\label{cor-mis}
Let $\Hcal$ be a hereditary graph class that is closed under 1-clique-sum.
Suppose that {\sc Weighted Vertex Cover} can be solved on instances $(G,w)$ where $G\in\Hcal$ in time $\Ocal(n^c\cdot m^d)$.
Then, given a graph $G$, a $1$-$\Hcal$-tree decomposition of $G$ of width $k$, and a weight function $w$, there is an algorithm that solves {\sc Weighted Vertex Cover}/{\sc Weighted Independent Set} on $(G,w)$ in time $\Ocal(2^k\cdot (k\cdot (k+n)+n^c\cdot m^d))$.
\end{corollary}

Given that the class $\Bcal$ of bipartite graphs is hereditary, closed under 1-clique-sums, and that {\sc Weighted Vertex Cover} can be solved on bipartite graphs in time $\Ocal(m^{1+o(1)})$~\cite{ChenKLPGS25maxi}, we obtain the following result concerning bipartite treewidth using \autoref{cor-mis}.

\begin{corollary}\label{mis-btw}
Given a graph $G$, a bipartite tree decomposition of $G$ of width $k$, and a weight function $w$, there is an algorithm that solves {\sc Weighted Vertex Cover}/{\sc Weighted Independent Set} on $(G,w)$ in time $\Ocal(2^k\cdot (k\cdot (k+n)+m^{1+o(1)}))$.
\end{corollary}

\subsubsection{Odd Cycle Transversal}
\label{sec-algo-oct}

Let $H$ be a graph.
We define $f_ {\oct}$  as the $3$-partition-evaluation function where, for every graph $G$ and for every $(S,X_1,X_2)\in\Pcal_3(V(G))$,
\begin{equation*}
   f_ {\oct}(G,(S,X_1,X_2))=
   \begin{cases}
     |S| & \text{if } G\setminus S\in\Bcal, \text{with bipartition $(X_1,X_2)$}, \\
     +\infty &  \text{otherwise.}
   \end{cases}
\end{equation*}

Hence, seen as an optimization problem, {\sc Odd Cycle Transversal} is the problem of computing $\p_{f_ {\oct},\min}(G)$.
We call its annotated extension {\sc \Annotated Odd Cycle Transversal}.
In other words, {\sc \Annotated Odd Cycle Transversal} is defined as follows.

\begin{center}
	\fbox{
		\begin{minipage}{11.5cm}
			\noindent{\sc (Weighted) \Annotated Odd Cycle Transversal}\\
			\noindent\textbf{Input}:~~A graph $G$, three disjoint sets $S,X_1,X_2\subseteq V(G)$ (and a weight function $w:V(G)\to\bN$).\\
			\textbf{Objective}:~~Find, if it exists, a set $S^\star$ of minimum size (resp. weight) such that $S\subseteq S^\star$, $(X_1\cup X_2)\cap S^\star=\emptyset$, and $G\setminus S^\star$ is bipartite with $X_1$ and $X_2$ on different sides of the bipartition.
		\end{minipage}
	}
\end{center}

We first prove that {\sc (Weighted) \Annotated Odd Cycle Transversal} has the gluing property.
\begin{lemma}[Gluing property]\label{glu-oct}
{\sc (Weighted) \Annotated Odd Cycle Transversal} has the gluing property.
More precisely, given two boundaried graphs ${\bf F}=(F,B_F,\rho_F)$ and ${\bf G}=(G,B_G,\rho_G)$, a function $w:V({\bf F\oplus G})\to\bN$, a set $X\subseteq V({\bf F}\oplus{\bf G})$ such that $B_F\cap B_G\subseteq X$, and $\Xcal=(S,X_1,X_2)\in\Pcal_3(X)$, we have
$$\hat{\p}_{f_{\oct},\min}({\bf F}\oplus{\bf G},\Xcal,w)=\hat{\p}_{f_{\oct},\min}(F,\Xcal\cap V(F),w)+\hat{\p}_{f_{\oct},\min}(G,\Xcal\cap V(G),w)-\bar{w},$$ where $\bar{w}=w(S\cap B_F\cap B_G)$.
\end{lemma}

\begin{proof}
Let $\Pcal=(S^\star,X_1^\star,X^\star_2)\in\Pcal_3(V({\bf F}\oplus{\bf G}))$ be such that $\Xcal\subseteq\Pcal$ and $\hat{\p}_{f_{\oct},\min}({\bf F}\oplus{\bf G},\Xcal,w)=f_{\oct}({\bf F}\oplus{\bf G},\Pcal,w).$
Then, for $H\in\{F,G\}$, $H\setminus (S^\star\cap V(H))$ is bipartite, witnessed by the 2-partition $(X_1^\star\cap V(H),X_2^\star\cap V(H))$.
Therefore,
\begin{align*}
\hat{\p}_{f_{\oct},\min}({\bf F}\oplus{\bf G},&\Xcal,w)=w(S^\star)\\
&=w(S^\star\cap V(F))+w(S^\star\cap V(G))-w(S^\star\cap B_F\cap B_G)\\
&\geq \hat{\p}_{f_{\oct},\min}(F,\Xcal\cap V(F),w)+\hat{\p}_{f_{\oct},\min}(G,\Xcal\cap V(G),w)-\bar{w}.
\end{align*}

Reciprocally, let $\Pcal_H=(S_H,X_1^H,X_2^H)\in\Pcal_3(V(H))$ be such that $\Xcal\cap V(H)\subseteq \Pcal_H$ and $\hat{\p}_{f_{\oct},\min}(H,\Xcal\cap V(H),w)=f_{\oct}(H,\Pcal_H,w)$ for $H\in\{F,G\}$.
Since $\Pcal_F\cap B_F\cap B_G=\Pcal_G\cap B_F\cap B_G$, it follows that $X_1^F\cup X_1^G$ and $X_2^F\cup X_2^G$ are two independent sets of $({\bf F\oplus G})\setminus (S_F\cup S_G)$.
Therefore, $({\bf F\oplus G})\setminus (S_F\cup S_G)$ is a bipartite graph witnessed by $(X_1^F\cup X_1^G,X_2^F\cup X_2^G)$. Thus,
\begin{align*}
\hat{\p}_{f_{\oct},\min}({\bf F}\oplus{\bf G},&\Xcal,w)\leq w(S_F\cup S_G)\\
&=w(S_F)+w(S_G)-\bar{w}\\
&=\hat{\p}_{f_{\oct},\min}(F,\Xcal\cap V(F),w)+\hat{\p}_{f_{\oct},\min}(G,\Xcal\cap V(G),w)-\bar{w}.
\end{align*}
 \end{proof}

We now show how to reduce a graph ${\bf F}\oplus{\bf G}$ to a graph $F'$ when the boundary of ${\bf F}$ and ${\bf G}$ has a single vertex $v$ that is not annotated.
See \autoref{fig_gadget_OCT} for an illustration.
Similarly to \autoref{obs+1}, the proof of \autoref{obs+oct} only holds in the unweighted case.

\begin{figure}[h]
\center
\includegraphics[scale=0.7]{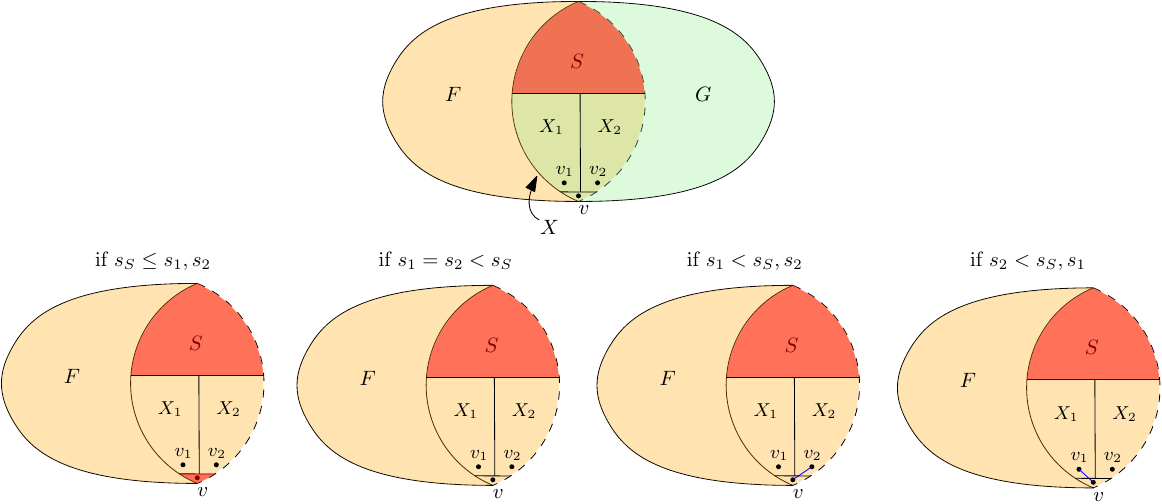}
\caption{Illustration of the gadgetization for {\sc Odd Cycle Transversal}.}
\label{fig_gadget_OCT}
\end{figure}

\begin{lemma}[Gadgetization]\label{obs+oct}
Let ${\bf F}=(F,B_F,\rho_F)$ and ${\bf G}=(G,B_G,\rho_G)$ be two boundaried graphs.
Let $X\subseteq V({\bf F}\oplus{\bf G})$ be such that $B_F\cap B_G\subseteq X$.
Let also $v\in B_F\cap B_G$, let $\Xcal=(S,X_1,X_2)\in\Pcal_3(X\setminus\{v\})$ with $X_1,X_2\neq\emptyset$, and
let $v_1\in X_1$ and $v_2\in X_2$.
We set $\Xcal_S=(S\cup\{v\},X_1,X_2)$, $\Xcal_1=(S,X_1\cup\{v\},X_2)$, and $\Xcal_2=(S,X_1,X_2\cup\{v\})$.
We also set, for $a\in\{S,1,2\}$, $s_a=\hat{\p}_{f_{\oct},\min}(G,\Xcal_a\cap V(G))$, and we set
$\bar{s}=|S\cap B_F\cap B_G|$.
For $i\in[2]$, we note $F_i$ the graph obtained from $F$ by adding the edge $vv_i$.
Then we have the following case distinction.
\begin{equation*}
   \hat{\p}_{f_{\oct},\min}({\bf F}\oplus{\bf G},\Xcal)=
   \begin{cases}
     s_S+\hat{\p}_{f_{\oct},\min}(F,\Xcal_S\cap V(F))-\bar{s}-1 & \text{if } s_S\leq s_1,s_2, \\
     s_1+\hat{\p}_{f_{\oct},\min}(F,\Xcal\cap V(F))-\bar{s} &  \text{if } s_1=s_2<s_S,\\
s_1+\hat{p}_{f_\oct,\min}(F_2,\Xcal\cap V(F))-\bar{s} & \text{if } s_1<s_S,s_2,\\
s_2+\hat{p}_{f_\oct,\min}(F_1,\Xcal\cap V(F))-\bar{s} & \text{otherwise.}
   \end{cases}
\end{equation*}
\end{lemma}

\begin{proof}
For $H\in\{F,G\}$ and $a\in\{S,1,2\}$, let $\Pcal_H^a=(S_H^a,X_{1,H}^a,X_{2,H}^a)$
be a partition of $V(H)$ such that $\Xcal_a\cap V(H)\subseteq\Pcal_H^a$ and $\hat{\p}_{f_\oct,\min}(H,\Xcal_a\cap V(H))=f_\oct(H,\Pcal_H^a)=|S_H^a|$.
We therefore have $s_a=|S_G^a|$. Let us similarly define $t_a=|S_F^a|$ for $a\in\{S,1,2\}$.

Note that $S^S_F\cap S^S_G=(S\cap B_F\cap B_G)\cup\{v\}$ and for $a\in\{1,2\}$, $S^a_F\cap S^a_G=S\cap B_F\cap B_G$.
Hence, using \autoref{glu-kt}, we have that
\begin{align*}
\hat{\p}_{f_{\oct},\min}({\bf F}\oplus{\bf G},\Xcal)
&=\min\{\hat{\p}_{f_{\oct},\min}({\bf F}\oplus{\bf G},\Xcal_a)\mid a\in\{S,1,2\}\}\\
&=\min\{t_S+s_S-1,t_1+s_1,t_2+s_2\}-\bar{s}.
\end{align*}
Note that $s_S\leq s_1+1$, since $G\setminus (S_G^1\cup\{v\})$ is bipartite, witnessed by the 2-partition $(X_1\setminus\{v\},X_2)$, and thus $|S_G^1\cup\{v\}|\geq \hat{\p}_{f_\oct,\min}(G,\Xcal_S\cap V(G))$.
Similarly, $s_S\leq s_2+1$, $t_S\leq t_1+1$, and $t_S\leq t_2+1$.

Hence, if $s_S\leq s_1,s_2$, then $\min\{t_S+s_S-1,t_1+s_1,t_2+s_2\}=t_S+s_S-1$, so
$$\hat{\p}_{f_{\oct},\min}({\bf F}\oplus{\bf G},\Xcal)=s_S+t_S-\bar{s}-1.$$

If $s_1=s_2<s_S$, then $s_1=s_2=s_S-1$. Thus,
$\min\{t_S+s_S-1,t_1+s_1,t_2+s_2\}=\min\{t_S,t_1,t_2\}+s_1=\hat{\p}_{f_{\oct},\min}(F,\Xcal\cap V(F))+s_1$, so
$$\hat{\p}_{f_{\oct},\min}({\bf F}\oplus{\bf G},\Xcal)=s_1+\hat{\p}_{f_{\oct},\min}(F,\Xcal\cap V(F))-\bar{s}.$$

If $s_1<s_S,s_2$, then $s_1+1=s_S\leq s_2$.
We have $t_2+s_2\geq t_S+s_2\geq t_S+s_S-1.$
Thus,
$\min\{t_S+s_S-1,t_1+s_1,t_2+s_2\}=\min\{t_S,t_1\}+s_1=s_1+\min\{\hat{\p}_{f_{\oct},\min}(F,\Xcal_S\cap V(F)),\hat{\p}_{f_{\oct},\min}(F,\Xcal_1\cap V(F))\}$.
Hence, we just need to ensure that $v$ cannot be added to $X_2$, which is done by adding an edge between $v$ and $v_2\in X_2$.
Therefore, $$\hat{\p}_{f_{\oct},\min}({\bf F}\oplus{\bf G},\Xcal)=s_1+\hat{p}_{f_\oct,\min}(F_2,\Xcal\cap V(F))-\bar{s}.$$

Otherwise, $s_2<s_S,s_1$
By symmetry, we similarly obtain
$$\hat{\p}_{f_{\oct},\min}({\bf F}\oplus{\bf G},\Xcal)=s_2+\hat{p}_{f_\oct,\min}(F_1,\Xcal\cap V(F))-\bar{s}.$$
 \end{proof}

Using \autoref{glu-oct} and \autoref{obs+oct}, we can now prove that {\sc \Annotated Odd Cycle Transversal} is $\Hcal$-nice.
Essentially, given an instance $({\bf G}={\bf X}\boxplus(\boxplus_{i\in[d]}{\bf G}_i),(A,B),\Acal)$,
we reduce ${\bf G}$ to ${\bf X}$ where we add two new vertices $u_1$ and $u_2$ in $\Acal$ and add edges between $u_i$ and some vertices in $B$, for $i\in[2]$, and show that the resulting boundaried graph is equivalent to $\bf G$ modulo some constant $s$.

\begin{lemma}[Nice problem]\label{oct-nice}
Let $\Hcal$ be a hereditary graph class.
{\sc \Annotated Odd Cycle Transversal} is $\Hcal$-nice.
\end{lemma}

\begin{proof}
Let ${\bf G}=(G,X,ρ)$ be a boundaried  graph, 	
let ${\bf X}=(G[X],X,\rho_X)$ be a trivial boundaried graph and
let $\{{\bf G}_i=(G_{i},X_{i},ρ_{i})\mid i\in[d]\}$ be a collection of boundaried graphs,
such that ${\bf G}={\bf X}\boxplus(\boxplus_{i\in[d]}{\bf G}_i)$,
let $(A,B)$ be a partition of $X$ such that
for all $i \in [d]$,
$|X_i\setminus A|\leq 1$,
and
let $\Acal=(S,X^1,X^2)\in\Pcal_3(A)$.
Suppose that we know,
for every $i\in[d]$ and each $\mathcal{X}_i \in \Pcal_3(X_i)$, the value $\hat{\mathsf{p}}_{f_{\oct},\mathsf{min}}(G_i,\mathcal{X}_i)$.

Let ${\bf G}'$ and ${\bf X}'$ be the boundaried graphs obtained from $\bf G$ and $\bf X$ respectively, by adding two new isolated vertices $u_1$ and $u_2$ in the boundary (with unused labels).
Let $\Acal'=(S,X^1\cup\{u_1\},X^2\cup\{u_2\})$. This operation is done to ensure that $X_i'=X_i\cup\{u_i\}$ is non-empty for $i\in[2]$.
Obviously, $\hat{\p}_{f_{\oct},\min}({\bf G}\oplus{\bf F},\Acal)=\hat{\p}_{f_{\oct},\min}({\bf G}'\oplus{\bf F},\Acal')$.

Let $v_1,\ldots,v_{|B|}$ be the vertices of $B$.
For $i\in[|B|]$, let $I_i=\{j\in[d]\mid X_j\setminus A=\{v_i\}\}$.
Let $I_0=\{j\in[d]\mid X_j\subseteq A\}$.
Obviously, $(I_i)_{i\in[0,|B|]}$ is a partition of $[d]$.

Let us show that item (iii) holds.
Let $({\bf H}_{-1},S_{-1},s_{-1},E_{-1})=({\bf G}',S,0,\emptyset)$.
For $i$ going from 0 up to $|B|$, we construct $({\bf H}_i,S_{i-1},s_i,E_i)$ from $({\bf H}_{i-1},S_i,s_{i-1},E_{i-1})$
such that for any boundaried graph $\bf F$ with underlying graph $F$ and compatible with~$\bf G$,
$$\hat{\p}_{f_{\oct},\min}({\bf G}\oplus{\bf F},\Acal)=\hat{\p}_{f_{\oct},\min}({\bf H}_i\oplus{\bf F}_i,\Acal_i)+s_i,$$
where ${\bf F}_i$ is the boundaried graph obtained from $\bf F$ by adding the edges in $E_i$ and $\Acal_i=(S_i,X^1\cup\{u_1\},X^2\cup\{u_2\})$.
This is obviously true for $i=-1$.

Let $i\in[0,|B|]$.
By induction, we have $$\hat{\p}_{f_{\oct},\min}({\bf G}\oplus{\bf F},\Acal)=\hat{\p}_{f_{\oct},\min}({\bf H}_{i-1}\oplus{\bf F}_{i-1},\Acal_{i-1})+s_{i-1}.$$
Let ${\bf G}'_i=\boxplus_{j\in I_i} {\bf G}_j$.
Let ${\bf H}_i'$ be the boundaried graph with underlying graph $H_i'$ such that ${\bf H}_{i-1}={\bf H}_i'\boxplus{\bf G}_i'$.

\medskip
\noindent{\bf Suppose first that $i=0$.}
Let $\Pcal_0=(S^0,X_1^0,X_2^0)\in\Pcal_3(X\cup\{u_1,u_2\})$ be such that $\Acal'\subseteq\Pcal_0$ and
$\hat{\p}_{f_\oct,\min}({\bf G}'\oplus{\bf F},\Acal')=\hat{\p}_{f_\oct,\min}({\bf G}'\oplus{\bf F},\Pcal_0).$
According to \autoref{glu-oct},
\begin{align*}
\hat{\p}_{f_\oct,\min}({\bf G}'\oplus{\bf F},\Pcal_0)
&=\hat{\p}_{f_\oct,\min}(F,\Pcal_0)+\hat{\p}_{f_\oct,\min}(G',\Pcal_0)-|S^0\cap X|\\
&=\hat{\p}_{f_\oct,\min}(F,\Pcal_0)-|S^0\cap X|+\hat{\p}_{f_\oct,\min}({H}_0',\Pcal_0)\\
&\ \ \ \ +\sum_{j\in I_0}(\hat{\p}_{f_\oct,\min}(G_j,\Pcal_0\cap X_j)-|S^0\cap  X_j|)\\
&=\hat{\p}_{f_\oct,\min}({\bf H}_0'\oplus{\bf F},\Pcal_0)+\sum_{j\in I_0}(\hat{\p}_{f_\oct,\min}(G_j,\Pcal_0\cap X_j)-|S\cap  X_j|),\\
&\ \ \ \ \ \ \ \text{ because }S^0\setminus S \subseteq B\text{ and }X_j\cap B=\emptyset.
\end{align*}
Since this is the case for all such $\Pcal_0$, it implies that
\begin{align*}
\hat{\p}_{f_\oct,\min}({\bf G}'\oplus{\bf F},\Acal')
&=\hat{\p}_{f_\oct,\min}({\bf H}_0'\oplus{\bf F},\Acal')+\sum_{j\in I_0}(\hat{\p}_{f_\oct,\min}(G_j,\Pcal_0\cap X_j)-|S\cap  X_j|).
\end{align*}
Therefore, if ${\bf H}_0={\bf H}_0'$, $S_0=S$, $s_0=\sum_{j\in I_0}(\hat{\p}_{f_\oct,\min}(G_j,\Pcal_0\cap X_j)-|S\cap  X_j|)$, and $E_0=\emptyset$, then
$$\hat{\p}_{f_{\oct},\min}({\bf G}\oplus{\bf F},\Acal)=\hat{\p}_{f_{\oct},\min}({\bf H}_0\oplus{\bf F}_0,\Acal_0)+s_0.$$

\noindent{\bf Otherwise, $i\in[|B|]$ and $X_j\setminus A=\{v_i\}$ for each $j\in I_i$.}
Let $X_i'=\bigcup_{j\in I_i}X_j$.
Let $\Xcal_i^S=(S_{i-1}\cup\{v_i\},X^1\cup\{u_1\},X^2\cup\{u_2\})$, $\Xcal_i^1=(S_{i-1},X^1\cup\{u_1,v_i\},X^2\cup\{u_2\})$, and $\Xcal_i^2=(S_{i-1},X^1\cup\{u_1\},X^2\cup\{u_2,v_i\})$.
For $a\in\{S,1,2\}$, let $s_i^a=\hat{\p}_{f_\oct,\min}(G_i',\Xcal_i^a\cap X_i')$.
By \autoref{glu-oct},
$$s_i^S=\sum_{j\in I_i}(\hat{p}_{f_\oct,\min}(G_j,\Xcal_i^S\cap X_j)-|S_{i-1}\cap X_j|-1)+|S_{i-1}\cap X_i'|+1$$
and, for $a\in\{1,2\}$,
$$s_i^a=\sum_{j\in I_i}(\hat{p}_{f_\oct,\min}(G_j,\Xcal_i^a\cap X_j)-|S_{i-1}\cap X_j|)+|S_{i-1}\cap X_i'|.$$
Therefore, $s_i^a$ can be computed for $a\in\{S,1,2\}$ in time $\Ocal(|A|\cdot |I_i|)$.
For $a\in\{1,2\}$, let ${\bf H}_i^a$ and ${\bf F}_i^a$ be the boundaried graphs obtained from ${\bf H}_i'$ and ${\bf F}_{i-1}$, respectively, by adding the edge $v_iu_a$.
Let $\Pcal_i=(S^i,X_1^i,X_2^i)\in\Pcal_3(X\cup\{u_1,u_2\}\setminus\{v_i\})$ be such that $\Acal_{i-1}\subseteq\Pcal_i$ and
$\hat{\p}_{f_\oct,\min}({\bf H}_{i-1}\oplus{\bf F}_{i-1},\Acal_{i-1})=\hat{\p}_{f_\oct,\min}({\bf H}_{i-1}\oplus{\bf F}_{i-1},\Pcal_i).$
Let $\Pcal_i^S=(S^i\cup\{v_i\},X_1^i,X_2^i)$, $\Pcal_i^1=(S^i,X_1^i\cup\{v_i\},X_2^i)$, and $\Pcal_i^2=(S^i,X_1^i,X_2^i\cup\{v_i\})$.
Note that for $a\in\{S,1,2\}$, we have $s_i^a=\hat{\p}_{f_\oct,\min}(G_i',\Pcal_i^a\cap X_i')$.
Then, using \autoref{obs+oct},
\begin{align*}
\hat{\p}_{f_\oct,\min}&({\bf H}_{i-1}\oplus{\bf F}_{i-1},\Pcal_i)
=\hat{\p}_{f_\oct,\min}(({\bf H}_{i}'\boxplus{\bf F}_{i-1})\oplus{\bf G}_i',\Pcal_i)\\
&=
\begin{cases}
     s_S+\hat{\p}_{f_{\oct},\min}({\bf H}_{i}'\oplus{\bf F}_{i-1},\Pcal_i^S)-\bar{s}-1 & \text{ if } s_i^S\leq s_i^1,s_i^2, \\
     s_1+\hat{\p}_{f_{\oct},\min}({\bf H}_{i}'\oplus{\bf F}_{i-1},\Pcal_i)-\bar{s} & \text{ if } s_i^1=s_i^2<s_i^S,\\
s_1+\hat{p}_{f_\oct,\min}({\bf H}_i^2\oplus{\bf F}_{i}^2,\Pcal_i)-\bar{s} & \text{ if } s_i^1<s_i^S,s_i^2, \text{and }\\
s_2+\hat{p}_{f_\oct,\min}({\bf H}_i^1\oplus{\bf F}_{i}^1,\Pcal_i)-\bar{s} &\text{ otherwise.}
\end{cases}
\end{align*}
Since this is the case for any such $\Pcal_i$, we have that
\begin{align*}
\hat{\p}_{f_{\oct},\min}({\bf G}\oplus{\bf F},\Acal)
&=\hat{\p}_{f_\oct,\min}({\bf H}_{i-1}\oplus{\bf F}_{i-1},\Acal_{i-1})+s_{i-1}\\
&=\hat{\p}_{f_\oct,\min}({\bf H}_{i}\oplus{\bf F}_{i},\Acal_i)+s_{i},
\end{align*}
where
\begin{align*}
({\bf H}_i,S_i,s_i,E_i)=
\begin{cases}
({\bf H}_i',S_{i-1}\cup\{v\},s_{i-1}+s_S-\bar{s}-1,E_{i-1}) & \text{if } s_i^S\leq s_i^1,s_i^2, \\
({\bf H}_i',S_{i-1},s_{i-1}+s_1-\bar{s},E_{i-1}) &  \text{if } s_i^1=s_i^2<s_i^S,\\
({\bf H}_i^2,S_{i-1},s_{i-1}+s_1-\bar{s},E_{i-1}\cup\{v_iu_2\}) & \text{if } s_i^1<s_i^S,s_i^2,\\
({\bf H}_i^1,S_{i-1},s_{i-1}+s_2-\bar{s},E_{i-1}\cup\{v_iu_1\}) & \text{otherwise.}
\end{cases}
\end{align*}

Let $S_B=S_{|B|}\setminus S\subseteq B$.
Let ${\bf H}_B={\bf H}_{|B|}\setminus S_B$.
We have $({\bf H}_{|B|}\oplus{\bf F}_{|B|})\setminus S_B=({\bf H}_{|B|}\oplus{\bf F})\setminus S_B=({\bf H}_{|B|}\setminus S_B)\triangleright{\bf F}={\bf H}_{B}\triangleright{\bf F}$.
Observe that
\begin{align*}
\hat{\p}_{f_\oct,\min}({\bf H}_{|B|}\oplus{\bf F}_{|B|},\Acal_{|B|})
&= \hat{\p}_{f_\oct,\min}(({\bf H}_{|B|}\oplus{\bf F}_{|B|})\setminus S_B,\Acal_{|B|}\setminus S_B)+|S_B|\\
&= \hat{\p}_{f_\oct,\min}({\bf H}_{B}\triangleright{\bf F},\Acal')+|S_B|.
\end{align*}
Hence,
$$\hat{\p}_{f_{\oct},\min}({\bf G}\oplus{\bf F},\Acal)=\hat{\p}_{f_\oct,\min}({\bf H}_{B}\triangleright{\bf F},\Acal')+|S_B|+s_{|B|},$$
so item (iii) holds.

Observe that ${\bf H}_B$ is isomorphic to the boundaried graph obtained from ${\bf X}$ by adding two new vertices and the at most $|B|$ edges in $E_{|B|}$, and removing the vertices in $S_B$.
Hence, $|V(H_B)|\leq |X|+2$ and $|E(H_B)|\leq |E(G[X])|+|B|$, so item (ii) holds.
Moreover, $|{\cup\Acal'}|=|{\cup\Acal}|+2$, so item (i) holds.
Suppose that $F\setminus A\in\Hcal$.
Observe that, since the edges in $E_{|B|}$ all have one endpoint in $\{u_1,u_2\}$, $({\bf H}_{B}\triangleright{\bf F})\setminus (\cup\Acal)\setminus\{u_1,u_2\}$ is isomorphic to $(F\setminus A)\setminus S_B$.
Since $\Hcal$ is hereditary, $({\bf H}_{B}\triangleright{\bf F})\setminus (\cup\Acal')\in\Hcal$, and thus, item (iv) holds.
Thus, $({\bf H}_B,\Acal',|S_B|+s_{|B|})$ is an $\Hcal$-nice reduction of $(G,\Acal)$ with respect to {\sc Annotated Odd Cycle Transversal}.

At each step $i$, $(H_i,S_i,s_i,E_i)$ is computable in time $\Ocal(|A|\cdot |I_i|)$, so
the computation takes time $\Ocal(|A|\cdot d)$.
Hence, {\sc Annotated Odd Cycle Transversal} is $\Hcal$-nice.
 \end{proof}

In the next lemma, we adapt the seminal proof of Reed, Smith and Vetta \cite{ReedSV04find} that uses iterative compression to solve {\sc \Annotated Odd Cycle Transversal} in \FPT-time parameterized by \oct.

Given a graph $G$ and two sets $A,B\subseteq V(G)$, an \emph{$(A,B)$-cut} is a set $X\subseteq V(G)$ such that there are no paths from a vertex in $A$ to a vertex in $B$ in $V(G)\setminus X$.

\begin{lemma}\label{lem-oct-annot}
There is an algorithm that, given a graph $G$, (a weight function $w:V(G)\to\bN$,) and three disjoint sets $S,A,B\subseteq V(G)$, such that $G\setminus (S\cup A\cup B)$ is bipartite, solves {\sc \Annotated Odd Cycle Transversal} (resp. {\sc Weighted \Annotated Odd Cycle Transversal}) on $(G,S,A,B)$ in time $\Ocal((m+k^2)^{1+o(1)})$, where $k=|A\cup B|$.
\end{lemma}

\begin{proof}
We can assume that $S=\emptyset$, given that $S^\star$ is an optimal solution for $(G,S,A,B)$ if and only if $S^\star\setminus S$ is an optimal solution for $(G\setminus S,\emptyset,A,B)$.

Let $G^+$ be the graph obtained from $G$ by joining each $a\in A$ with each $b\in B$.
Let $X=A\cup B$.
Let $(S_1,S_2)$ be a partition witnessing the bipartiteness of $G\setminus X=G^+\setminus X$.
We construct an auxiliary bipartite graph $G'$ from $G^+$ as follows.
The vertex set of the auxiliary graph is $V'=V(G)\setminus X\cup\{x_1,x_2\mid x\in X\}$.
We maintain a one-to-one correspondence between the edges of $G^+$
and the edges of $G'$ by the following scheme (and with the same weight):
\begin{itemize}
\item for each edge $e$ of $G^+\setminus X$, there is a corresponding edge in $G'$ with the same endpoints,
\item for each edge $e\in E(G^+)$ joining a vertex $y\in S_i$ to a vertex $x\in X$, the corresponding edge in $G'$ joins $y$ to $x_{3-i}$, and
\item for each edge $e\in E(G^+)$ joining two vertices $a\in A$ and $b\in B$, the corresponding edge of $G'$ joins $a_1$ to $b_2$.
\end{itemize}
For $i\in\{1,2\}$, let $X_i=\{x_i\mid x\in X\}$, $A_i=\{a_i\mid a\in A\}$, $B_i=\{b_i\mid b\in B\}$.
Note that $G'$ is a bipartite graph, witnessed by the partition $(S_1\cup X_1,S_2\cup X_2)$.
Let $Y_1=A_1\cup B_2$ and $Y_2=A_2\cup B_1$.
Note also that there is no edge joining $Y_1$ and $Y_2$, so there exists a $(Y_1,Y_2)$-cut in $G'$ that is actually contained in $V(G^+)\setminus X$, and hence in $V(G)\setminus X$.
Let us show that $S^\star\subseteq V(G)\setminus X$ is a $(Y_1,Y_2)$-cut in $G'$ if and only if $S^\star$ is an odd cycle transversal of $G$ with $A$ on one side of the bipartition and $B$ on the other one.

\begin{claim}
If $S^\star\subseteq V(G)\setminus X$ is a cutset separating $Y_1$ from $Y_2$ in $G'$, then $S^\star$ is an odd cycle transversal of $G$ with $A$ on one side of the bipartition and $B$ on the other one.
\end{claim}
\begin{cproof}
Let $C$ be an odd cycle of $G^+$.
Suppose toward a contradiction that $C\cap S^\star=\emptyset$.
$G^+\setminus X$ is bipartite, so $C$ intersects $X$.
Moreover, we assumed $G^+[X]$ to be bipartite since $A$ and $B$ are two disjoint independent sets, so $C$ intersects $V(G^+)\setminus X$.
Hence, if we divide $C$ into paths whose endpoints are in $X$ and whose internal vertices are in $V(G^+)\setminus X$, each such path in $G'$ has either both endpoints in $Y_1$ or both endpoints in $Y_2$, since it is not intersected by the cutset $S^\star$.
More specifically, each such path in $G'$ has either both endpoints in $A_i$, or both endpoints in $B_i$, or one in $A_i$ and one in $B_{3-i}$, for $i\in\{1,2\}$.
A path with both endpoints in $A_i$ or $B_i$ is even, since the internal vertices are alternatively vertices of $S_i$ and $S_{3-i}$, with the first and the last in $S_{3-i}$.
A path with one endpoint in $A_i$ and one in $B_{3-i}$ is odd by a similar reasoning, but the number of such paths must be even in order to have a cycle.
Therefore, the cycle $C$ is even.
Hence the contradiction.
So $S^\star$ is an odd cycle transversal of $G^+$.
Given that we added in $G^+$ all edges between $A$ and $B$ and that $A,B\subseteq V(G^+)\setminus S^\star$, $A$ and $B$ belong to different sides of the bipartition.
An odd cycle of $G$ is also an odd cycle of $G^+$, so $S^\star$ is also an odd cycle transversal of $G$, with $A$ and $B$ on different sides of the bipartition.
\end{cproof}

\begin{claim}
If $S^\star\subseteq V(G)\setminus X$ is an odd cycle transversal of $G$ with $A$ and $B$ on different sides of the bipartition, then $S^\star$ is a cutset separating $Y_1$ from $Y_2$ in $G'$.
\end{claim}
\begin{cproof}
Suppose toward a contradiction that there is a path $P$ between $Y_1$ and $Y_2$ that does not intersect $S^\star$.
Choose $P$ of minimum length. Hence, the internal vertices of $P$ belong to $G\setminus X$.
The endpoints $u$ and $v$ of $P$ are such that, either $u\in A_i$ and $v\in A_{3-i}$, or $u\in B_i$ and $v\in B_{3-i}$, or $u\in A_i$ and $v\in B_i$, or $u\in B_i$ and $v\in A_i$, for $i\in\{1,2\}$.
By symmetry, we can assume without loss of generality that $u\in A_1$ and $v\in A_2$ or $v\in B_1$.
If $v\in A_2$, then $P$ is an odd path since $G\setminus X$ is bipartite.
However, since $G\setminus S^\star$ is bipartite, with $A$ on one side of the bipartition, $P$ is an even path.
If $v\in B_1$, then $P$ is an even path since $G\setminus X$ is bipartite.
However, since $G\setminus S^\star$ is bipartite, with $A$ and $B$ on different sides of the bipartition, $P$ is an odd path.
Hence the contradiction.
\end{cproof}

The graph $G'$ has $n'=n+|X|$ vertices and at most $m'=m+|X|^2/4$ edges.
Finding a minimum (weighted) vertex-cut can be reduced to the problem of finding a minimum (weighted) edge-cut.
To do so, we transform $G'$ into an arc-weighted directed graph $G''$, by first replacing every edge by two parallel arcs in opposite directions, and then replacing every vertex $v$ of $G'$ by an arc $(v_{\sf in}, v_{\sf out})$, such that the arcs incoming (resp. outgoing) to $v$ are now incoming to $v_{\sf in}$ (resp. outgoing of $v_{\sf out}$).
We give weight $w(v)$ to $(v_{\sf in}, v_{\sf out})$, and weight $w(V(G))+1$ to the other arcs.
Then, computing a minimum (weighted) vertex-cut in $G'$ is equivalent to computing a minimum (weighted) edge-cut in $G''$, which can be done in time $\Ocal((m')^{1+o(1)})$~\cite{ChenKLPGS25maxi}.
Hence, the running time of the algorithm is $\Ocal((m+|X|^2)^{1+o(1)})$.
 \end{proof}

We apply \autoref{oct-nice} and \autoref{lem-oct-annot} to the dynamic programming algorithm of \autoref{DP} to obtain the following result.

\begin{corollary}\label{co-algo-OCT}
Given a graph $G$ and a bipartite tree decomposition of $G$ of width $k$, there is an algorithm that solves {\sc Odd Cycle Transversal} on $G$ in time $\Ocal(3^{k}\cdot (m+k^2\cdot n)^{1+o(1)})$.
\end{corollary}

\subsubsection{Maximum Weighted Cut}\label{sec-cut}

The {\sc Maximum Weighted Cut} problem is defined as follows.

\begin{center}
	\fbox{
		\begin{minipage}{11.5cm}
			\noindent{\sc Maximum Weighted Cut}\\
			\noindent\textbf{Input}:~~A graph $G$ and a weight function $w:E(G)\to\bN$.\\
			\textbf{Objective}:~~Find an edge cut of maximum weight.
		\end{minipage}
	}
\end{center}

Let $H$ be a graph.
We define $f_ {\sf cut}$  as the $2$-partition-evaluation function where, for every graph $G$ with edge weight $w$ and for every $\Pcal=(X_1,X_2)\in\Pcal_2(V(G))$,
\begin{equation*}
   f_ {\sf cut}(G,\Pcal)=w(\Pcal)=w(E(X_1,X_2)).
\end{equation*}

Hence, {\sc Maximum Weighted Cut} is the problem of computing $\p_{f_ {\sf cut},\max}(G)$.
We call its annotated extension {\sc \Annotated Maximum Weighted Cut}.
In other words, {\sc \Annotated Maximum Weighted Cut} is defined as follows.

\begin{center}
	\fbox{
		\begin{minipage}{11.5cm}
			\noindent{{\sc \Annotated Maximum Weighted Cut}}\\
			\noindent\textbf{Input}:~~A graph $G$, a weight function $w:E(G)\to\bN$, and two disjoint sets $X_1,X_2\subseteq V(G)$.\\
			\textbf{Objective}:~~Find an edge cut of maximum weight such that the vertices in $X_1$ belongs to one side of the cut, and the vertices in $X_2$ belong to the other side.
		\end{minipage}
	}
\end{center}
For simplicity, given a set $X\subseteq V(G)$ and $\Xcal=(A,B)\in\Pcal_2(X)$, we write $w(\Xcal)$, to denote the sum of the  weights of the edges with one endpoint in $A$ and the other in $B$.

\medskip
We first prove that {\sc \Annotated Maximum Weighted Cut} has the gluing property.
\begin{lemma}[Gluing property]\label{glu-cut}
{\sc \Annotated Maximum Weighted Cut} has the gluing property.\\
More precisely, given two boundaried graphs ${\bf F}=(F,B_F,\rho_F)$ and ${\bf G}=(G,B_G,\rho_G)$, a weight function $w:E({\bf F\oplus G})\to\bN$, a set $X\subseteq V({\bf F}\oplus{\bf G})$ such that $B_F\cap B_G\subseteq X$, and $\Xcal=(X_1,X_2)\in\Pcal_2(X)$, if we set $\bar{w}={w}(\Xcal\cap B_F\cap B_G)$, then we have
$$\hat{\p}_{f_{\cut},\max}({\bf F}\oplus{\bf G},\Xcal,w)=\hat{\p}_{f_{\cut},\max}(F,\Xcal\cap V(F),w)+\hat{\p}_{f_{\cut},\max}(G,\Xcal\cap V(G),w)-\bar{w}.$$
\end{lemma}

\begin{proof}
Let $\Pcal\in\Pcal_2(V({\bf F}\oplus{\bf G}))$ be such that $\Xcal\subseteq\Pcal$ and $\hat{\p}_{f_{\cut},\max}({\bf F}\oplus{\bf G},\Xcal,w)=f_{\cut}({\bf F}\oplus{\bf G},\Pcal,w).$
Then,
\begin{align*}
\hat{\p}_{f_{\cut},\max}({\bf F}\oplus{\bf G},&\Xcal,w)=w(\Pcal)\\
&=w(\Pcal\cap V(F))+w(\Pcal\cap V(G))-\bar{w}\\
&\leq \hat{\p}_{f_{\cut},\max}(F,\Xcal\cap V(F),w)+\hat{\p}_{f_{\cut},\max}(G,\Xcal\cap V(G),w)-\bar{w}.
\end{align*}

Reciprocally, for $H\in\{F,G\}$, let $\Pcal_H=(X_1^H,X_2^H)\in\Pcal_2(V(H))$ be such that $\Xcal\cap V(H)\subseteq\Pcal_H$ and $\hat{\p}_{f_{\cut},\max}(H,\Xcal\cap V(H),w)=f_{\oct}(H,\Pcal_H,w)$.
Then, since $\Pcal_H\cap B_F\cap B_G=\Xcal\cap B_F\cap B_G$ for $H\in\{F,G\}$, we have
\begin{align*}
\hat{\p}_{f_{\cut},\max}({\bf F}\oplus{\bf G},&\Xcal,w)\geq w(E(X_1^F\cup X_1^G,X_2^F\cup X_2^G))\\
&=w(E(X_1^F,X_2^F))+w(E(X_1^G,X_2^G))-\bar{w}\\
&=\hat{\p}_{f_{\cut},\max}(F,\Xcal\cap V(F),w)+\hat{\p}_{f_{\cut},\max}(G,\Xcal\cap V(G),w)-\bar{w}.
\end{align*}
 \end{proof}

We now show how to reduce a graph ${\bf F}\oplus{\bf G}$ to a graph $F'$ when the boundary of ${\bf F}$ and ${\bf G}$ has a single vertex $v$ that is not annotated. See \autoref{fig_gadget_max_cut} for an illustration.

\begin{figure}[h]
\center
\includegraphics[scale=0.7]{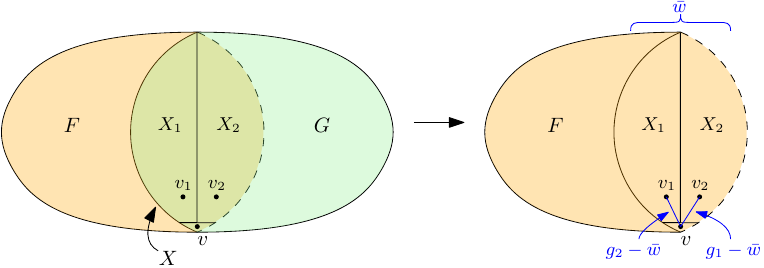}
\caption{Illustration of the gadgetization for {\sc Weighted Max Cut}.}
\label{fig_gadget_max_cut}
\end{figure}

\begin{lemma}[Gadgetization]\label{obs+cut}
Let ${\bf F}=(F,B_F,\rho_F)$ and ${\bf G}=(G,B_G,\rho_G)$ be two boundaried graphs.
Let $w:E({\bf F\oplus G})\to\bN$ be a weight function and
let $X\subseteq V({\bf F}\oplus{\bf G})$ be such that $B_F\cap B_G\subseteq X$.
Let also $v\in B_F\cap B_G$ and let $\Xcal=(X_1,X_2)\in\Pcal_2(X\setminus\{v\})$.
Suppose that there is $v_1\in X_1\cap B_F\cap B_G$ and $v_2\in X_2\cap B_F\cap B_G$ adjacent to $v$ with $w(vv_1)=w(vv_2)=0$.
We set $\Xcal^1=(X_1\cup\{v\},X_2)$ and $\Xcal^2=(X_1,X_2\cup\{v\})$.
We also set, for $a\in[2]$, $g_a=\hat{\p}_{f_{\cut},\max}(G,\Xcal^a\cap V(G),w)$,
we set $\bar{w}={w}(\Xcal\cap B_F\cap B_G)$, and we set
$w':E(F)\to\bN$ such that $w'(vv_1)=g_2-\bar{w}$,
$w'(vv_2)=g_1-\bar{w}$, and $w'(e)=w(e)$ otherwise.
Then
$$\hat{\p}_{f_{\cut},\max}({\bf F}\oplus{\bf G},\Xcal,w)=\hat{\p}_{f_{\cut},\max}(F,\Xcal,w').$$
\end{lemma}

\begin{proof}
For $a\in[2]$, let $f_a=\hat{\p}_{f_{\cut},\max}(F,\Xcal^a\cap V(F),w)$.
Note that in $F$ with partition $\Xcal$, if $v$ is on the same side as $X_1$, then we must count the weight of the edge $vv_2$, but not the weight of $vv_1$, and vice versa when exchanging 1 and 2.
Thus, using \autoref{glu-cut}, we have
\begin{align*}
\hat{\p}_{f_{\cut},\max}({\bf F}\oplus{\bf G},\Xcal,w)
&=\max\{\hat{\p}_{f_{\cut},\max}({\bf F}\oplus{\bf G},\Xcal^1,w),\hat{\p}_{f_{\cut},\max}({\bf F}\oplus{\bf G},\Xcal^2,w)\}\\
&=\max\{f_1+g_1-\bar{w},f_2+g_2-\bar{w}\}\\
&=\max\{f_1+w'(vv_2),f_2+w'(vv_1)\}\\
&=\max\{\hat{\p}_{f_{\cut},\max}(F,\Xcal^1,w'),\hat{\p}_{f_{\cut},\max}(F,\Xcal^2,w')\}\\
&=\hat{\p}_{f_{\cut},\max}(F,\Xcal,w').
\end{align*}
 \end{proof}

Using \autoref{glu-cut} and \autoref{obs+cut}, we prove that {\sc \Annotated Maximum Weighted Cut} is $\Hcal$-nice.
Essentially, given an instance \hspace{-0.05cm}$({\bf G}={\bf X}\boxplus(\boxplus_{i\in[d]}{\bf G}_i),(A,B), \Acal,w)$,
we reduce ${\bf G}$ to ${\bf X}$ where we add two new vertices in $\Acal$ and add all edges between these new vertices and the vertices in $B$.
We then show that if the appropriate weight is given to each new edge, then the resulting boundaried graph is an equivalent instance to $\bf G$ modulo some constant $s$.

\begin{lemma}[Nice problem]\label{cut-nice}
Let $\Hcal$ be a graph class.
{\sc \Annotated Maximum Weighted Cut} is $\Hcal$-nice.
\end{lemma}

\begin{proof}
Let ${\bf G}=(G,X,ρ)$ be a boundaried  graph,
let $w:E(G)\to\bN$ be a weight function,
let ${\bf X}=(G[X],X,\rho_X)$ be a trivial boundaried graph and
let $\{{\bf G}_i=(G_{i},X_{i},ρ_{i})\mid i\in[d]\}$ be a collection of boundaried graphs,
such that ${\bf G}={\bf X}\boxplus(\boxplus_{i\in[d]}{\bf G}_i)$,
let $(A,B)$ be a partition of $X$ such that
for all $i \in [d]$,
$|X_i\setminus A|\leq 1$,
and
let $\Acal=(X^1,X^2)\in\Pcal_2(A)$.
Suppose that we know,
for every $i\in[d]$ and each $\mathcal{X}_i \in \Pcal_3(X_i)$, the value $\hat{\mathsf{p}}_{f_{\cut},\max}(G_i,\mathcal{X}_i,w)$.

Let $\bar{\bf G}$, $\bar{\bf X}$, and $\bar{\bf G}_i$ be the boundaried graphs obtained from $\bf G$, $\bf X$, and ${\bf G}_i$, respectively, by adding two new vertices $u_1$ and $u_2$ in the boundary (with unused labels) and making them adjacent to every vertex in $B$.
We extend $w$ to $E(G')$ such that any edge adjacent to $u_1$ or $u_2$ has weight zero.
Let $\Acal'=(X^1\cup\{u_1\},X^2\cup\{u_2\})$. This operation is done to ensure that $X_i'=X^i\cup\{u_i\}$ is non-empty for $i\in[2]$.
For any boundaried graph $\bf F$ compatible with $G$, let $\bar{\bf F}$ be the boundaried graph obtained similarly to $\bar{\bf G}$.
Obviously, $\hat{\p}_{f_{\cut},\max}({\bf G}\oplus{\bf F},\Acal,w)=\hat{\p}_{f_{\cut},\max}(\bar{\bf G}\oplus\bar{\bf F},\Acal',w)=\hat{\p}_{f_{\cut},\max}(\bar{\bf G}\triangleright{\bf F},\Acal',w)$.

Let $v_1,\ldots,v_{|B|}$ be the vertices of $B$.
For $i\in[|B|]$, let $I_i=\{j\in[d]\mid X_j\setminus A=\{v_i\}\}$.
Let $I_0=\{j\in[d]\mid X_j\subseteq A\}$.
Obviously, $(I_i)_{i\in[0,|B|]}$ is a partition of $[d]$.

Let us show that item (iii) holds.
Let $({\bf H}_{-1},s_{-1},w_{-1})=(\bar{\bf G},0,w)$.
For $i$ going from 0 up to $|B|$, we will construct $({\bf H}_i,s_i,w_i)$ from $({\bf H}_{i-1},s_{i-1},w_{i-1})$
such that $w_i$ and $w_{i-1}$ may only differ on $v_iu_1$ and $v_iu_2$ and
such that for any boundaried graph $\bf F$ compatible with $\bf G$ and with underlying graph $F$,
$$\hat{\p}_{f_{\cut},\max}({\bf G}\oplus{\bf F},\Acal)=\hat{\p}_{f_{\cut},\max}({\bf H}_i\oplus\bar{\bf F},\Acal',w_i)+s_i.$$
This is obviously true for $i=-1$.

Let $i\in[0,|B|]$.
Let ${\bf G}_i'=\boxplus_{j\in I_i}\bar{\bf G}_j$.
Let ${\bf H}_i$ be the boundaried graph such that ${\bf H}_{i-1}={\bf H}_i\boxplus{\bf G}_i$.
By induction, we have
$$\hat{\p}_{f_{\cut},\max}({\bf G}\oplus{\bf F},\Acal,w)=
\hat{\p}_{f_{\cut},\max}({\bf H}_{i-1}\oplus\bar{\bf F},\Acal',w_{i-1})+s_{i-1}.$$

\medskip
\noindent{\bf Suppose first that $i=0$.}
Let $\Pcal_0=(X_1^0,X_2^0)\in\Pcal_2(X\cup\{u_1,u_2\})$ be such that $\Acal'\subseteq \Pcal_0$ and
$\hat{\p}_{f_{\cut},\max}(\bar{\bf G}\oplus\bar{\bf F},\Acal',w)=\hat{\p}_{f_{\cut},\max}(\bar{\bf G}\oplus\bar{\bf F},\Pcal_0,w).$
According to \autoref{glu-cut},
\begin{align*}
\hat{\p}_{f_{\cut},\max}(\bar{\bf G}\oplus\bar{\bf F},\Pcal_0,w)
&= \hat{\p}_{f_{\cut},\max}(\bar{F},\Pcal_0,w)+\hat{\p}_{f_{\cut},\max}(\bar{G},\Pcal_0,w)-w(E(\Pcal_0\cap X))\\
&=\hat{\p}_{f_{\cut},\max}(\bar{F},\Pcal_0,w)-w(E(\Pcal_0\cap X))+\hat{\p}_{f_{\cut},\max}(H_0,\Pcal_0,w)\\
&\ \ \ +\sum_{j\in I_0}(\hat{\p}_{f_{\cut},\max}(G_j,\Pcal_0\cap X_j,w)-w(\Pcal_0\cap X_j))\\
&= \hat{\p}_{f_{\cut},\max}({\bf H}_0\oplus\bar{\bf F},\Pcal_0,w)+\sum_{j\in I_i}(\hat{\p}_{f_{\cut},\max}(G_j,\Acal\cap X_j,w)-w(\Acal\cap X_j)).
\end{align*}
Since this is the case for all such $\Pcal_0$, it implies that
\begin{align*}
\hat{\p}_{f_{\cut},\max}(\bar{\bf G}\oplus\bar{\bf F},\Acal',w)&=
\hat{\p}_{f_{\cut},\max}({\bf H}_0\oplus\bar{\bf F},\Acal',w)+\sum_{j\in I_i}(\hat{\p}_{f_{\cut},\max}(G_j,\Acal\cap X_j,w)-w(\Acal\cap X_j)).
\end{align*}
Therefore, if $w_0=w$ and $s_0=s_{-1}+\sum_{j\in I_i}(\hat{\p}_{f_{\cut},\max}(G_j,\Acal\cap X_j,w)-w(\Acal\cap X_j))$, then
$$\hat{\p}_{f_{\cut},\max}({\bf G}\oplus{\bf F},\Acal,w)=\hat{\p}_{f_{\cut},\max}({\bf H}_0\oplus\bar{\bf F},\Acal,w_0)+s_0.$$

\medskip
\noindent{\bf Otherwise, $i\in[|B|]$ and $X_j\setminus A=\{v_i\}$ for each $j\in I_i$.}
Let $X_i'=\bigcup_{j\in I_i}X_j$.
Let $\Xcal_i^1=(X^1,X^2\cup\{v_i\})$ and $\Xcal_i^2=(X^1\cup\{v_i\},X^2)$.
Let $g_i^1=\hat{\p}_{f_{\cut},\max}(G_i',\Xcal_i^1\cap X_i',w)$ and $g_i^2=\hat{\p}_{f_{\cut},\max}(G_i',\Xcal_i^2\cap X_i',w)$.
For $a\in[2]$, $g_i^a$ can be computed in time $\Ocal(|A|\cdot |I_i|)$ since, by \autoref{glu-cut},
\begin{align*}
g_i^a&=\hat{\p}_{f_{\cut},\max}(\oplus_{j\in I_i} \bar G_j,\Xcal_i^a\cap X_i',w)\\
&=\sum_{j\in I_i}(\hat{\p}_{f_{\cut},\max}(\bar G_j,\Xcal_i^a\cap X_i'\cap V(\bar G_j))-w(\Xcal_i^a\cap X_i'\cap V(\bar G_j)))+w(\Xcal_i^a\cap X_i')\\
&=\sum_{j\in I_i}(\hat{\p}_{f_{\cut},\max}(G_j,\Xcal_i^a\cap X_j,w)-w(\Xcal_i^a\cap X_j))+w(\Xcal_i^a\cap X_i').
\end{align*}
Let $w_i:E({\bf H}_i\oplus\bar{\bf F})\to\bN$ be such that $w_i(v_iu_1)=g_i^2-w(\Xcal_i^a\cap X_i')$, $w_i(v_iu_2)=g_i^1-w(\Xcal_i^a\cap X_i')$, and $w_i(e)=w_{i-1}(e)$ otherwise.
Let $\Pcal_i=(R^i,S^i)\subseteq\Pcal_2(X\cup\{u_1,u_2\}\setminus\{v_i\})$ be such that $\Acal'\subseteq \Pcal_i$ and
$\hat{\p}_{f_{\cut},\max}({\bf H}_{i-1}\oplus\bar{\bf F},\Acal,w_{i-1})=\hat{\p}_{f_{\cut},\max}({\bf H}_{i-1}\oplus\bar{\bf F},\Pcal_i,w_{i-1}).$
Then, using \autoref{obs+cut},
\begin{align*}
\hat{\p}_{f_{\cut},\max}({\bf H}_{i-1}\oplus\bar{\bf F},\Pcal_i,w_{i-1})
&=\hat{\p}_{f_{\cut},\max}(({\bf H}_{i}\boxplus\bar{\bf F})\oplus{\bf G}_i',\Pcal_i,w_{i-1})\\
&=\hat{\p}_{f_{\cut},\max}({\bf H}_{i}\oplus\bar{\bf F},\Pcal_i,w_{i}).
\end{align*}
Since this is the case for all such $\Pcal_i$, it implies that
$$\hat{\p}_{f_{\cut},\max}({\bf H}_{i-1}\oplus\bar{\bf F},\Acal',w_{i-1})=
\hat{\p}_{f_{\cut},\max}({\bf H}_i\oplus\bar{\bf F},\Acal',w_{i}).$$
Therefore, given $s_i=s_{i-1}$,
$$\hat{\p}_{f_{\cut},\max}({\bf G}\oplus{\bf F},\Acal,w)=\hat{\p}_{f_{\cut},\max}({\bf H}_i\oplus\bar{\bf F},\Acal',w_i)+s_i.$$

Note that ${\bf H}_{|B|}$ is isomorphic to $\bar{\bf X}$.
We thus have
$$\hat{\p}_{f_{\cut},\max}({\bf G}\oplus{\bf F},\Acal,w)=\hat{\p}_{f_{\cut},\max}(\bar{\bf X}\triangleright{\bf F},\Acal',w_{|B|})+s_{|B|},$$
so item (iii) holds.

We have $|V(\bar{X})|\leq |X|+2$ and $|E(\bar{X})|\leq |E(G[X])|+2|B|$.
Moreover, $|\cup\Acal'|=|\cup\Acal|+2$.
Therefore, item (i) and (ii) hold.
Suppose that $F\setminus A\in\Hcal$.
Observe that, since the edges added in $\bar{\bf X}$ compared to $\bf X$ all have one endpoint in $\{u_1,u_2\}$, $({\bf X}'\triangleright{\bf F})\setminus (\cup\Acal)\setminus\{u_1,u_2\}$ is isomorphic to $F\setminus A$ and thus belong to $\Hcal$. So item (iv) holds.
Thus, $({\bf X}',\Acal',s_{|B|},w_{|B|})$ is an $\Hcal$-nice reduction of $(G,\Acal,w)$ with respect to {\sc Annotated Maximum Weighted Cut}.

Computing $({\bf H}_i,s_i,w_i)$ takes time $\Ocal(|A|\cdot |I_i|)$ at each step $i$, so
the computation takes time $\Ocal(|A|\cdot d)$.
Hence, {\sc Annotated Maximum Weighted Cut} is $\Hcal$-nice.
 \end{proof}

{\sc Maximum Weighted Cut} is an \NP-hard problem \cite{Karp10redu}.
However, there exists a polynomial-time algorithm when restricted to some graph classes.
In particular, Gr{\"{o}}tschel and Pulleyblank \cite{GrotschelP81weak} proved that {\sc Maximum Weighted Cut} is solvable in polynomial-time on weakly bipartite graphs, and Guenin \cite{Guenin01acha} proved that weakly bipartite graphs are exactly $K_5$-odd-minor-free graphs, which gives the following result.

\begin{proposition}[\hspace{1sp}\cite{GrotschelP81weak,Guenin01acha}]\label{k5}
There is a constant $c\in\bN$ and an algorithm that solves
{\sc Maximum Weighted Cut} on $K_5$-odd-minor-free graphs in time $\Ocal(n^c)$.
\end{proposition}

Moreover, we observe the following.

\begin{lemma}\label{oct2}
A graph $G$ such that $\oct(G)\leq2$ does not contain $K_5$ as an odd-minor.
\end{lemma}

\begin{proof}
Suppose that $G$ contains $K_5$ as an odd-minor and
let $\eta$ be an odd $K_5$-expansion of $G$.
Let $u,v\in V(G)$ be such that $G'=G\setminus\{u,v\}$ is bipartite.
Given that $\eta$ has at least three nodes that do not intersect $\{u,v\}$, it implies that $K_3$ is an odd-minor of $G'$, contradicting its bipartiteness.
 \end{proof}

Combining \autoref{k5} and \autoref{oct2}, we have that {\sc \Annotated Maximum Weighted Cut} is \FPT\ parameterized by \oct.

\begin{lemma}\label{cut-oct}
There is an algorithm that, given a graph $G$, a weight function $w:E(G)\to\bN$, and two disjoint sets $X_1,X_2\subseteq V(G)$,
such that $G'=G\setminus(X_1\cup X_2)$ is bipartite, solves {\sc \Annotated Maximum Weighted Cut} on $(G,X_1,X_2,w)$ in time $\Ocal(k\cdot n'+n'^c)$, where $k=|X_1\cup X_2|$, $n'=|V(G')|$, and $c$ is the constant of \autoref{k5}.
\end{lemma}

\begin{proof}
Let $G''$ be the graph obtained from $G$ by identifying all vertices in $X_1$ (resp. $X_2$) to a new vertex $x_1$ (resp. $x_2$).
Let $w':V(G'')\to\bN$ be such that $w'(x_1x_2)=\sum_{e\in E(G)}w(e)+1$, $w'(x_iu)=\sum_{x\in X_i}w(xu)$ for $i\in[2]$ and $u\in N_G(X_i)$, and $w'(e)=w(e)$ otherwise.
Let $(X_1^\star,X_2^\star)\in\Pcal_2(V(G))$ be such that $(X_1,X_2)\subseteq (X_1^\star,X_2^\star)$.
For $i\in[2]$, let $X_i'=X_i^\star\setminus X_i$.
Then
\begin{align*}
w(X_1^\star,X_2^\star)
&=w(X_1,X_2)+w(X_1',X_2')+\sum_{xy\in E(X_1,X_2')}w(xy)+\sum_{xy\in E(X_1',X_2)}w(xy)\\
&=w(X_1,X_2)+w'(X_1',X_2')+\hspace{-0.5cm}\sum_{u\in X_2'\cap N_G(X_1)}w'(x_1u)+\hspace{-0.5cm}\sum_{u\in X_1'\cap N_G(X_2)}\hspace{-0.1cm}w'(x_2u)\\
&=w'(X_1'\cup\{x_1\},X_2'\cup\{x_2\})+w(X_1,X_2)-w'(x_1x_2)
\end{align*}
Let $\bar{w}$ be the constant $w(X_1,X_2)-w'(x_1x_2)$.
Hence, $$f_\cut(G,(X_1^\star,X_2^\star))=f_\cut(G'',(X_1'\cup\{x_1\},X_2'\cup\{x_2\}))+\bar{w},$$ and so $\hat{p}_{f_\cut,\max}(G,(X_1,X_2))=\hat{p}_{f_\cut,\max}(G'',(\{x_1\},\{x_2\}))+\bar{w}$.
Moreover, given that the weight of the edge $x_1x_2$ is larger than the sum of all other weights, $x_1$ and $x_2$ are never on the same side of a maximum cut in $G''$.
Hence, $\hat{p}_{f_\cut,\max}(G'',(\{x_1\},\{x_2\}))=\p_{f_\cut,\max}(G'')$, and therefore,
$$\hat{p}_{f_\cut,\max}(G,(X_1,X_2))={p}_{f_\cut,\max}(G'')+\bar{w}.$$

Constructing $G''$ takes time $\Ocal(k\cdot n)$ and computing $\bar{w}$ takes time $\Ocal(k^2)$.
Since $\oct(G'')=2$, according to \autoref{k5} and \autoref{oct2}, an optimal solution to {\sc Maximum Weighted Cut} on $G''$ can be found in time $\Ocal(n'^c)$, and thus, an optimal solution to {\sc \Annotated Maximum Weighted Cut} on $(G,X_1,X_2)$ can be found in time $\Ocal(k\cdot (k+n')+n'^c)$.
 \end{proof}
We apply \autoref{cut-nice} and \autoref{cut-oct} to the dynamic programming algorithm of \autoref{DP} to obtain the following result.

\begin{corollary}
Given a graph $G$ and a bipartite tree decomposition of $G$ of width $k$, there is an algorithm that solves {\sc Maximum Weighted Cut} on $G$ in time $\Ocal(2^k\cdot (k\cdot (k+n)+n^c))$, where $c$ is the constant of \autoref{k5}.
\end{corollary}

\section{\XP-algorithms for packing problems}\label{xp}
\newcommand\packing{\mathcal{S}}

Let $\Gcal$ be a graph class.
We define the {\PbPack} problem as follows.

\begin{center}
	\fbox{
		\begin{minipage}{11.5cm}
			\noindent{\PbPack}\\
			\noindent\textbf{Input}:~~A graph $G$.\\
			\textbf{Objective}:~~Find the maximum number $k$ of pairwise-disjoint subgraphs\\
			 \phantom{\textbf{Question}:~~} $H_1,\ldots,H_k$ such that, for each $i\in[k]$, $H_i\in\Gcal$.
		\end{minipage}
	}
\end{center}
Let $H$ be a graph.
If $\Gcal=\{H\}$ (resp. $\Gcal$ is the class of all graphs containing  $H$ as a minor/odd-minor/induced subgraph), then we refer to the corresponding problem as {\sc $H$-Subgraph-Packing} (resp. {\sc $H$-Minor-Packing}/{\sc $H$-Odd-Minor-Packing}/{\sc $H$-Induced-Subgraph-Packing}).
We remark, in particular, that {\sc $K_3$-Odd-Minor-Packing} is exactly {\sc Odd Cycle Packing}.

If in the definition of \PbPack\ we add the condition that there is no edge in the input graph between vertices of different $H_{i}$'s, then we refer to the corresponding problem as {\sc $H$-Scattered-Packing}, where we implicitly assume that we refer to the subgraph relation, and where we do {\sl not} specify a degree of ``scatteredness'', as it is usual in the literature when dealing, for instance, with the scattered version of \textsc{Independent Set}.
For instance, {\sc $K_2$-Scattered-Packing} is exactly {\sc Induced Matching}.
\medskip

As we prove in \autoref{pack-bip-min}, {\sc $H$-Minor-Packing} is {\sf para-NP}-complete parameterized by \btw\ when $H$ is 2-connected.\
This is however not the case for {\sc $H$-Subgraph-Packing} and {\sc $H$-Odd-Minor-Packing} when $H$ is 2-connected non-bipartite: we provide below an \XP-algorithm for both problems.

These algorithms have a similar structure to the dynamic programming algorithm of \autoref{DP}.
Namely, for each node $t$ of the bipartite tree decomposition, we reduce $\bf G_t$ to a smaller equivalent instance $\bf G_t'$, and solve the problem on $\bf G_t'$.
The main observation here is that the maximum size of a packing is small when restricted to the bag of $t$. Thus, we guess the packing in each bag, and how it intersects the adhesion of $t$ and its neighbors. This guessing is the source of \XP-time in our algorithms.\medskip

Recall that $\p_{f,\opt}$ is defined in \autoref{sec:nice}.
Given a graph class $\Gcal$, {\sc $\Gcal$-Packing}, seen as an optimization problem, is the problem of computing $\p_{f_{\sf pack},\max}(G)$, where, for every graph $G$, for every $(S,R)\in\Pcal_2(V(G))$,
\begin{align*}
f_{\sf pack}(G,(S,R))=
\begin{cases}
|S| & \text{if there are $|S|$ disjoint subgraphs $H_1,\ldots,H_{|S|}$ in $G$}\\
 & \text{ such that $|H_i\cap S|=1$ and $H_i\in\Gcal$ for $i\in[|S|]$},\\
0 & \text{otherwise.}
\end{cases}
\end{align*}

In the above equation, the pair $(S,R)$ means that $S$ is the part that interacts with the solution, and $R$ is the remainder.

\subsection{(Induced) Subgraph packing}

Let $H$ be a graph.
A \emph{partial copy} of $H$ is a boundaried graph ${\bf F}$ such that there is a boundaried graph ${\bf F'}$ compatible with ${\bf F}$ with ${\bf F} \oplus{\bf F'}=H$.
Given a graph $G$ and a set $X\subseteq V(G)$, we denote by $\Ccal_{G,X}^H$ the set of all partial copies ${\bf F}$ of $H$ such that  $G[\bd({\bf F})]$ is a subgraph of $G[X]$.

\begin{lemma}\label{Hpack}
	Let $\Hcal$ be a graph class and let $H$ be a 2-connected graph with $h$ vertices that does not belong to $\Hcal$.
	There is an algorithm that, given a graph $G$ and a 1-$\Hcal$-tree decomposition of $G$ of width at most $k$,
	solves the following problem in time $n^{\Ocal(h\cdot k)}$:
	\begin{enumerate}
		\item\label{Hpack:ind} \textsc{$H$-Induced-Subgraph-Packing}, if $\Hcal$ is hereditary.
		\item\label{Hpack:monotone} \textsc{$H$-Subgraph-Packing}, if $\Hcal$ is monotone.
		\item\label{Hpack:monotone2} \textsc{$H$-Scattered-Packing}, if $\Hcal$ is monotone.
	\end{enumerate}
\end{lemma}

\begin{proof}
	The three cases follow from the same argument.
	When speaking of a ``copy'' of $H$ in $G$,
	we mean an induced subgraph of $G$ isomorphic to $H$ for Case~\ref{Hpack:ind},
	and a subgraph of $G$ isomorphic to $H$ for Case~\ref{Hpack:monotone} and Case~\ref{Hpack:monotone2}.
	If we speak of an $H$-packing,
	we mean an $H$-induced-subgraph packing in Case~\ref{Hpack:ind},
	an $H$-subgraph packing in Case~\ref{Hpack:monotone},
	and an $H$-scattered packing in Case~\ref{Hpack:monotone2}.

Let $(T,\alpha,\beta,r)$ be  a rooted 1-$\Hcal$-tree decomposition of $G$ of width at most $k$.
Let $k\cdot H$ be the union of $k$ disjoint copies of $H$.
For each $t\in V(T)$, in a bottom-up manner, for each ${\bf F}\in\Ccal_{G,\delta_t}^{k\cdot H}$, we compute the maximum integer $s_t^{\bf F}$, if it exists, such that there exists an
$H$-packing
$H_1,\ldots,H_{s_t^{\bf F}}$  in $G_t$ such that ${\bf F}$ is a boundaried (induced in Case~\ref{Hpack:ind}) subgraph of ${\bf G_t}\setminus\bigcup_{i\in[s_t^{\bf F}]}V(H_i)$.
Since $\delta_r=\emptyset$, $s_r^\emptyset$ is the optimum for  {\sc $H$-(Induced-)Subgraph-Packing\xspace} on $G$.

Let $t\in V(T)$.
We observe that no copy of $H$ is fully contained in $G[\beta(t)]$.
In Case~\ref{Hpack:ind}, this follows immediately from the fact that $H \notin \Hcal$,
whereas $G[\beta(t)] \in \Hcal$, and that $\Hcal$ is hereditary.
In Case~\ref{Hpack:monotone} and Case~\ref{Hpack:monotone2},
we additionally use the assumption that $\Hcal$ is monotone.
We call any copy of $H$ in $G_t$ that intersects $\alpha(t)$ a \emph{$t$-inner copy}.
Given that $H$ is 2-connected and that $H\notin\Hcal$, for any copy of $H$ in $G_t$ that is not a $t$-inner copy, the definition of a 1-$\Hcal$-tree decomposition implies that there is $t'\in\ch_r(t)$ such that this copy is contained in $G_{t'}\setminus\alpha(t)$.
We call these copies \emph{$t$-outer copies}.
See \autoref{fig_pack} for an illustration.
\begin{figure}[ht]
\center
\includegraphics[scale=0.8]{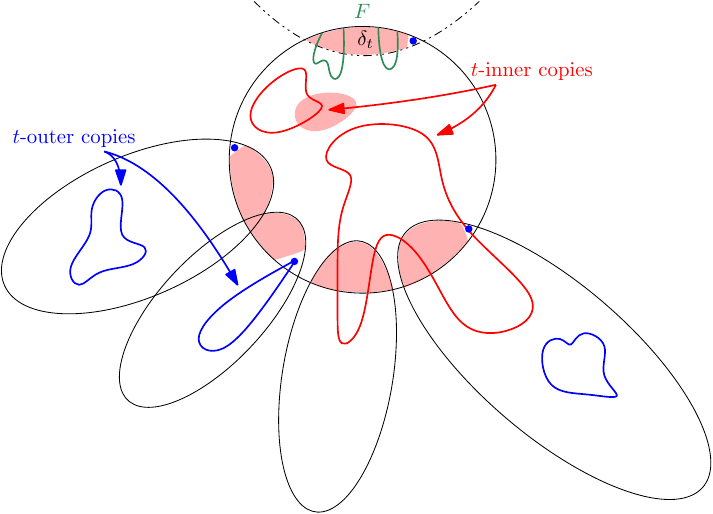}
\caption{Illustration of $t$-inner copies (the blue cycles) and $t$-outer copies (the red cycles). The red regions correspond to $\alpha(t)$ and each blue vertex is the vertex in $\delta_{t'}\setminus\alpha(t)$ for some node $t'$ adjacent to $t$.}
\label{fig_pack}
\end{figure}
Given that $|\alpha(t)|\leq k$, it follows that any $H$-subgraph-packing contains at most $k$ $t$-inner copies.
For each $t$-inner copy $H'$ of $H$, let $Y_{H'}$ be the set of the children $t'$ of $t$ such that $H'$ intersects $G_{t'}\setminus\alpha(t)$.
Since $H'$ has at most $h$ vertices, $|Y_{H'}|\leq h$.
Thus, for any maximum
$H$-packing
$\packing=\{H_1,\ldots,H_r\}$ in $G$, the set $Y_t=\{t_1,\ldots,t_{|Y_t|}\}$ of children of $t$ that intersect a $t$-inner copy in $\packing$ has size at most $h\cdot k$.
We guess $Y_t$, and note that there are at most $n^{h\cdot k}$ choices.
Let $Z_t=\ch_r(t)\setminus Y_t$.

No $t$-inner copy of $\packing$ intersects $G_{t'}\setminus\delta_{t'}$ for $t'\in Z_t$.
Thus, $\packing$ is maximum when restricted to $G_{t'}\setminus\delta_{t'}$.
However, if there is a vertex $v_{t'}$ such that $\delta_{t'}\setminus\alpha(t)=\{v_t'\}$,
there might still be one $t$-inner copy of $\packing$ that contains $v_{t'}$.
We want to find a maximum
$H$-packing
on $G_{t'}\setminus\alpha(t)$ that is contained in a maximum
$H$-packing
of $G$ such that no $t$-inner copy intersecting $Z_t$ is in the packing.
We compute inductively $s_{t'}^+=s_{t'}^{\bf F+}$ and $s_{t'}^-=s_{t'}^{\bf F-}$, where ${\bf F+}$ and ${\bf F-}$ are trivial boundaried graphs with boundary $\delta_{t'}\cap\alpha(t)$ and $\delta_{t'}$, respectively.
In other words, $s_{t'}^+$ (resp. $s_{t'}^-$) is the maximum size of an
$H$-packing
on $G_{t'}\setminus\alpha(t)$ (resp. $G_{t'}\setminus\delta_{t'}$).
Note that $s_{t'}^-\leq s_{t'}^+\leq s_{t'}^-+1$.
\begin{itemize}
\item If $s_{t'}^-= s_{t'}^+$, then it is optimal to choose a maximum $H$-subgraph-packing on $G_{t'}\setminus\delta_{t'}$. Therefore, we remove $V(G_{t'})\setminus\delta_{t'}$ from $G$ and set $s_{t'}=s_{t'}^-$.
\item Otherwise, $s_{t'}^+= s_{t'}^- +1$.
	In this case, $v_{t'}$ has to be part of some copy of $H$ in $\packing$.
	So, we may assume that $\packing$ consists of a maximum $H$-packing in $G_{t'} \setminus \alpha(t')$
	and a $t$-outer copy $H'$ of $H$ containing $v_{t'}$.
Indeed, if a $t$-inner copy of $\packing$ intersects $v_{t'}$, then $\packing$ restricted to $G_{t'}\setminus\delta_{t'}$ has size $s_{t'}^-$.
If we replace $H'$ and $\packing$ restricted to $G_{t'}\setminus\delta_{t'}$ by a maximum
$H$-packing
on $G_{t'}\setminus\alpha(t')$,
then we obtain a packing of the same size.
Therefore, we remove $(V(G_{t'})\setminus\delta_{t'})\cup\{v_{t'}\}$ from $G$ and set $s_{t'}=s_{t'}^+$.
If there are $z_1,\ldots,z_d\in Z_t$ such that $v_{z_1}=\ldots=v_{z_d}$ and $s_{z_i}^+= s_{z_i}^- +1$ for $i\in[d]$,
then we choose one of them arbitrarily, say $z_1$, for which we set $s_{z_1}=s_{z_1}^+$, and we set $s_{z_i}=s_{z_i}^-$ otherwise, since $v_{z_1}$ is contained in only one copy of an
$H$-packing.
\end{itemize}

Let ${\bf F}\in\Ccal_{G,\delta_t}^{k\cdot H}$.
Let $\Delta=\delta_t\cup\bigcup_{t'\in Y_t}\delta_{t'}$.
Now that we have dealt with the children containing no $t$-inner copy, we only have a few children left. To reduce the size of $G_t$, we will guess the partial $t$-inner copies in each child in $Y_t$.
$\Fcal_{\bf F}$ defined below is the set of all such possible guesses.
We set
$\Fcal=\{({\bf F_{t'}})_{t'\in Y_t}\mid\forall t'\in Y_t, {\bf F_{t'}}\in\Ccal_{G,\delta_{t'}}^{k\cdot H}\}$.
For each $\Lcal=({\bf F_{t'}})_{t'\in Y_t}\in\Fcal$,
we say that $\Lcal$ is \emph{compatible} with ${\bf F}$ if there is ${\bf F_\Lcal}\in\Ccal_{G,\Delta}^{k\cdot H}$
and a partition $(U,U_{t_1},\ldots,U_{t_{|Y|}})$ of the connected components of ${\bf F_\Lcal}\setminus\Delta$ such that
\begin{itemize}
\item ${\bf F_\Lcal}[V(U)\cup(\delta_t\cap V(F))]={\bf F}$ and
\item for $t'\in Y$, ${\bf F_\Lcal}[V(U_{t'})\cup(\delta_{t'}\cap V(F_{t'}))]={\bf F_{t'}}$.
\end{itemize}
Let $\Fcal_{\bf F}$ be the set of all $\Lcal\in\Fcal$ compatible with ${\bf F}$.
Let $\Lcal\in\Fcal_{\bf F}$.
Since we guessed how the $t$-inner copies of $\packing$ interact with each $t'\in Y_t$, we can now compute the $t$-outer copies of $\packing$ in $G_{t'}$. This is done iteratively by computing $s_{t'}^{\bf F_{t'}}$.
Let $G_\Lcal$ be the graph obtained from $\bf G_t$ by replacing ${\bf G_{t'}}$ by ${\bf F_{t'}}$ for each $t'\in Y_t$ and removing the vertices of $\Delta\setminus V(F_\Lcal)$.

Then we guess a set $W$ of $|V(F)|\leq h\cdot k$ vertices of $G_\Lcal$ that realize ${\bf F}$ (there are at most $n^{h\cdot k}$ possible choices) and we remove $W$  from $G_\Lcal$ to obtain $G_\Lcal^W$.
Then we compute the size $s_W$ of a maximum
$H$-packing
on $G_\Lcal^W$.
For each $i \in [k]$, we check whether we can find $i$ disjoint copies of $H$ in $G_\Lcal^W$
by brute-force.
Let $s_\Lcal=\max_W s_W +\sum_{t'\in Y_t}s_{t'}^{\bf F_{t'}}+\sum_{t'\in Z_t}s_{t'}$.
Then it follows that $s_t^{\bf F}=\max_{\Lcal\in\Fcal_{\bf F}}s_\Lcal$.

It remains to upper-bound the running time.
There are at most $n^{h\cdot k}$ choices for $Y_t$.
Let $h'=h\cdot k$ be the number of vertices of $H'$ and $d=|\Delta|\leq k+|Y_t|+1=\Ocal(h\cdot k)$.
There are at most $d!\cdot\binom{d}{h'}\cdot 2^{h'}\cdot|Y_t|^{h'}=(h\cdot k)^{\Ocal(h\cdot k)}$ choices for $\bf F_\Lcal$, and thus for $\Lcal$ and $\bf F$.
There are at most $n^{h\cdot k}$ choices for $W$.
Computing $s_W$ takes time at most $(n+h\cdot k)^{\Ocal(h\cdot k)}$,
since $G_\Lcal^W$ has size at most $n+h\cdot k$.
Hence, the running time of the algorithm is upper-bounded by $(h\cdot k\cdot n)^{\Ocal(h\cdot k)}$,
which is $n^{\Ocal(h\cdot k)}$ since $k \le n$ by definition and we may assume that $h \le n$.
 \end{proof}

\paragraph{Remark}
	Note that the above algorithm has two guessing phases that require \XP-time.
	For the case of \textsc{$H$-Subgraph-Packing}, however, we can get rid of
	the second guessing stage using color-coding~\cite{AlonYZ95colo},
	which allows us to find a subgraph with $\ell$ vertices and treewidth $t$
	inside an $n$-vertex graph in time $2^{\Ocal(\ell)}n^{\Ocal(t)}$.
	Since the disjoint union of $k$ copies of $H$ has treewidth less than $h = |V(H)|$
	and $h\cdot k$ vertices,
	the second \XP guessing stage can be replaced by an algorithm running in $2^{\Ocal(h\cdot k)}n^{\Ocal(h)}$ time,
	which is \FPT in our parameterization as $h$ is a fixed constant.

\subsection{Odd-minor-packing}
\label{sec-algo-odd-packing}

The algorithm for {\sc $H$-Odd-Minor-Packing} is very similar to the one for {\sc $H$-Subgraph-Packing}.
The main difference is the use of an \FPT-algorithm of Kawarabayashi, Reed, and Wollan~\cite{KawarabayashiRW11} solving the {\sc Parity $k$-Disjoint Paths} defined below.
We say that the parity of a path is zero if its length is even, and one otherwise.
\begin{center}
	\fbox{
		\begin{minipage}{11.5cm}
			\noindent{{\sc Parity $k$-Disjoint Paths}}\\
			\noindent\textbf{Input}:~~A graph $G$, two disjoint sets $S=\{s_1,\ldots,s_k\},T=\{t_1,\ldots,t_k\}\subseteq V(G)$, and values $j_1,\ldots,j_k\in\{0,1\}$.\\
			\textbf{Objective}:~~Find, if it exists, $k$ internally-vertex-disjoint paths $P_1,\ldots,P_k$ from $S$ to $T$ in $G$ such that $\Xcal_i$ has endpoints $s_i$ and $t_i$, and parity $j_i$ for $i\in[k]$.
		\end{minipage}
	}
\end{center}

\begin{proposition}[\hspace{1sp}\cite{KawarabayashiRW11}]\label{parity-pack}
There is an algorithm running in time $\Ocal_{k}(m\cdot\alpha(m,n)\cdot n)$ for the {\sc Parity $k$-Disjoints Paths} problem, where $\alpha$ is the inverse of the Ackermann function.
\end{proposition}

Given that an odd-minor preserves the cycle parity (\autoref{odd-eq}), we make the following observation.
\begin{observation}\label{odd-bip}
Any odd-minor of a bipartite graph is bipartite.
\end{observation}

Given an odd $H$-expansion $\eta$ for some graph $H$ in a graph $G$, the \emph{branch vertices} of $\eta$ are the vertices of $\eta$ of degree at least three and the ones incident to an edge not contained in a node of $\eta$.

\begin{lemma}\label{branch}
Let $H$ and $G$ be two graphs such that $H$ is an odd-minor of $G$ with $h$ vertices.
Any odd $H$-expansion in $G$ has at most $h\cdot(2h-3)$ branch vertices.
\end{lemma}

\begin{proof}
Let $\eta$ be an odd $H$-expansion in $G$.
At most $h-1$ vertices of each node of $\eta$ are adjacent to vertices of another node.
These $h-1$ vertices are the leaves of the tree induced by the vertices of the node in $G$.
By an easy induction, the number of vertices of degree at least three in a tree is at most the number of leaves minus one.
Therefore, there are at most $h-1+h-2=2h-3$ branch vertices in each node of $\eta$.
Since $\eta$ has $h$ nodes, the result follows.
 \end{proof}

Let $H$ be a graph.
A \emph{partial odd-model} of $H$ is an odd-minor-minimal boundaried graph ${\bf F}$ such that there is a boundaried graph ${\bf F'}$ compatible with ${\bf F}$ and such that $H$ is an odd-minor of ${\bf F} \oplus{\bf F'}$.
Given a graph $G$ and a set $X\subseteq V(G)$, we denote by $\Mcal_{G,X}^H$ the set of all boundaried graphs ${\bf F}$ such that $G[\bd({\bf F})]$ is a subgraph of $G[X]$.

\begin{lemma}\label{lem-algo-odd-minor-packing}
Let $H$ be a 2-connected non-bipartite graph with $h$ vertices.
There is an algorithm that, given a graph $G$ and a bipartite tree decomposition of $G$ of width at most $k$, solves {\sc $H$-Odd-Minor-Packing} in time $\Ocal_k(n^{\Ocal(h^2\cdot k)})$.
\end{lemma}

\begin{proof}
Let $(T,\alpha,\beta,r)$ be  a rooted bipartite tree decomposition of $G$ of width at most $k$.
Let $k\cdot H$ be the union of $k$ disjoint copies of $H$.
For each $t\in V(T)$, in a bottom-up manner, for each ${\bf F}\in\Mcal_{G,\delta_t}^{k\cdot H}$, we compute the maximum integer $s_t^{\bf F}$, if it exists, such that there exists an $H$-odd-minor-packing $H_1,\ldots,H_{s_t^{\bf F}}$  in $G_t$ where ${\bf F}$ is a boundaried odd-minor
of ${\bf G_t}\setminus\bigcup_{i\in[s_t^{\bf F}]}V(H_i)$.
Since $\delta_r=\emptyset$, it implies that $s_r^\emptyset$ is the optimum for  {\sc $H$-Odd-Minor-Packing} on $G$.

Let $t\in V(T)$.
Given that $H$ is non-bipartite, according to \autoref{odd-bip}, no odd $H$-expansion is contained in $G[\beta(t)]$.
We call any odd $H$-expansion in $G_t$ that intersects $\alpha(t)$ a \emph{$t$-inner odd-model}.
Given that $H$ is 2-connected and non-bipartite, for any odd $H$-expansion in $G_t$ that is not a $t$-inner odd-model, there is a $t'\in\ch_r(t)$ such that this odd $H$-expansion is contained in $G_{t'}\setminus\alpha(t)$.
We call these odd $H$-expansions \emph{$t$-outer odd-models}.
Given that $|\alpha(t)|\leq k$, it follows that at most $k$ many $t$-inner odd-models can be packed at once.
Since a $t$-inner odd-model of $H$ has at most $h'=h\cdot(2h-3)$ branch vertices by \autoref{branch},
for any maximum $H$-odd-minor-packing $\packing=\{H_1,\ldots,H_r\}$ in $G$, the set $Y_t=\{t_1,\ldots,t_{|Y_t|}\}$ of children of $t$ that intersect a $t$-inner odd-model in $\packing$ has size at most $h'\cdot k+k$.
The additive term $k$ refers to the fact that each path of the packing between branch vertices may go in and out of a child, but has to intersect at least one vertex in $\alpha(t)$ when crossing, so at most $|\alpha(t)|$ children can be intersected by a $t$-inner odd-model this way.
We guess $Y_t$, and note that there are at most $n^{\Ocal(h^2\cdot k)}$ choices.
Let $Z_t=\ch_r(t)\setminus Y_t$.

We proceed similarly to the proof of \autoref{Hpack}.
We compute inductively $s_{t'}^+=s_{t'}^{\bf F+}$ and $s_{t'}^-=s_{t'}^{\bf F-}$, where ${\bf F+}$ and ${\bf F-}$ are trivial boundaried graphs with boundary $\delta_{t'}\cap\alpha(t)$ and $\delta_{t'}$, respectively.
In other words, $s_{t'}^+$ (resp. $s_{t'}^-$) is the maximum size of an $H$-odd-minor-packing on $G_{t'}\setminus\alpha(t)$ (resp. $G_{t'}\setminus\delta_{t'}$).
Thus, we observe that $s_{t'}^-\leq s_{t'}^+\leq s_{t'}^-+1$.
\begin{itemize}
\item If $s_{t'}^-= s_{t'}^+$, then we remove $V(G_{t'})\setminus\delta_{t'}$ from $G$ and set $s_{t'}=s_{t'}^-$.
\item Otherwise, $s_{t'}^+= s_{t'}^- +1$.
Then, we remove $(V(G_{t'})\setminus\delta_{t'})\cup\{v_{t'}\}$ from $G$ and set $s_{t'}=s_{t'}^+$.
If there are $z_1,\ldots,z_d\in Z_t$ such that $v_{z_1}=\ldots=v_{z_d}$ and $s_{z_i}^+= s_{z_i}^- +1$ for $i\in[d]$,
then we choose one of them arbitrarily, say $z_1$, for which we set $s_{z_1}=s_{z_1}^+$, and we set $s_{z_i}=s_{z_i}^-$ otherwise, since $v_{z_1}$ is contained in only one odd $H$-expansion of an $H$-odd-minor-packing.
\end{itemize}

Let ${\bf F}\in\Mcal_{G,\delta_t}^{k\cdot H}$.
Let $\Delta=\delta_t\cup\bigcup_{t'\in Y_t}\delta_{t'}$.
Now that we have dealt with the children containing no $t$-inner odd-model, we only have a few children left. To reduce the size of $G_t$, we will guess the partial $t$-inner copies in each child in $Y_t$.
$\Fcal_{\bf F}$, defined as in the proof of \autoref{Hpack} but replacing partial copies with partial odd-models, is the set of all possible guesses.

Let $\Lcal\in\Fcal_{\bf F}$.
Since we guessed how the $t$-inner odd-models of $\Hcal$ interact with each $t'\in Y_t$, we can now compute the $t$-outer odd-models of $\Hcal$ in $G_{t'}$. This is done iteratively by computing $s_{t'}^{\bf F_{t'}}$.
Let $G_\Lcal$ be the graph obtained from $\bf G_t$ by replacing ${\bf G_{t'}}$ by ${\bf F_{t'}}$ for each $t'\in Y_t$ and removing the vertices of $\Delta\setminus V(F_\Lcal)$.

Then we guess a set $W$ of $|V(F)|\leq h'\cdot k$ vertices of $G_\Lcal$ that would be the branch vertices of the $t$-inner odd-models,
and of ${\bf F}$.
There are at most $n^{\Ocal(h^2\cdot k)}$ possible such choices.
Let $s_W$ be the size of the $H$-odd-minor-packing on $G_\Lcal$ corresponding to this choice of $W$, if it exists.
We can check its existence using \autoref{parity-pack}.
To do so, we guess which vertices in $W$ are joined by a path, and what is the parity of the path, so that we obtain $s_W$ disjoint odd $H$-expansions in $G_\Lcal$.
Since $G_\Lcal^W$ has size at most $n+h'\cdot k$, this takes time $\Ocal_k((n+h'\cdot k)^{\Ocal(1)})$.
Let $s_\Lcal=\max_W s_W +\sum_{t'\in Y_t}s_{t'}^{\bf F_{t'}}+\sum_{t'\in Z_t}s_{t'}$.
Then it follows that $s_t^{\bf F}=\max_{\Lcal\in\Fcal_{\bf F}}s_\Lcal$.

It remains to upper-bound the running time.
There are $n^{\Ocal(h^2\cdot k)}$ choices for $Y_t$.
There are $(h^2\cdot k)^{\Ocal(h^2\cdot k)}$ choices for $\bf F_\Lcal$, and thus for $\Lcal$ and $\bf F$.
There are at most $n^{\Ocal(h^2\cdot k)}$ choices for $W$.
Computing $s_W$ takes time $\Ocal_k((n+h^2\cdot k)^{\Ocal(1)})$.
Hence, the running time of the algorithm is $\Ocal_k((h^2\cdot k\cdot n)^{\Ocal(h^2\cdot k)})$,
which is $\Ocal_k(n^{\Ocal(h^2 \cdot k)})$ as $k \le n$ and we may assume that $h \le n$.
 \end{proof}

\section{\NP-completeness on graphs of bounded \btw}\label{bip}

In this section we present our hardness results.
For any graph $G$, it holds that $\btw(G)\leq\oct(G)$.
Thus, for a problem $\Pi$ to be efficiently solvable on graphs of bounded \btw, $\Pi$ needs to be efficiently solvable on graphs of bounded \oct, and first and foremost on bipartite graphs.
Unfortunately, many problems are \NP-complete on bipartite graphs (or on graphs of small \oct), and hence \textsf{para-}\NP-complete parameterized by \btw.
In this section we provide a non-exhaustive list of such problems. In fact, there also exist problems that
 are trivial or polynomially solvable on bipartite graphs, but are \textsf{para-}\NP-complete parameterized by \btw, such as the \textsc{$3$-Coloring} problem discussed in \autoref{color}.

\subsection{Coloring}\label{color}

The {\sc $3$-Coloring} problem is defined as follows.

\begin{center}
	\fbox{
		\begin{minipage}{5cm}
			\noindent{\sc $3$-Coloring}\\
			\noindent\textbf{Input}:~~A graph $G$.\\
			\textbf{Question}:~~Is $G$ $3$-colorable?
		\end{minipage}
	}
\end{center}

Bipartite graphs are 2-colorable, so we could hope for positive results about {\sc $3$-Coloring} on graphs of bounded \oct, or even bounded \btw.
In fact, we have the following result.

\begin{lemma}\label{approx-col}
If a graph has bipartite treewidth at most $k$, then it is $(k+2)$-colorable.
\end{lemma}

\begin{proof}
Let $G$ be a graph.
Let $\Tcal=(T,\alpha,\beta)$ be a bipartite tree decomposition of $G$ of width at most $k$.
We proceed by induction on $|V(T)|$.
For the base case, suppose that $T$ has a unique node $t$.
$G[\beta(t)]$ is bipartite, so it is 2-colorable.
Thus, if we color each vertex in $\alpha(t)$ with a unique new color, given that $|\alpha(t)|\leq k$, we can extend the 2-coloring of $G[\beta(t)]$ to a $(k+2)$-coloring of $G$.
Now suppose that $T$ has $\ell\geq 2$ nodes.

Let $t$ be a leaf of $T$.
Let $H=G[\bigcup_{t'\in V(T)\setminus\{t\}}(\alpha\cup\beta)(t')]$.
$(T\setminus\{t\},\alpha',\beta')$, where $\alpha'$ and $\beta'$ are the restrictions of $\alpha$ and $\beta$ to $T\setminus\{t\}$, respectively, is a bipartite tree decomposition of $H$ of width at most $k$.
By induction, $H$ admits an $(k+2)$-coloring $c$.
Let $\delta$ be the adhesion of $t$ and its neighbor.
Let $a=|\alpha(t)\cap \delta|\leq k$.
Given that $|\delta\cap\beta(t)|\leq 1$, it follows that $|\delta|\leq a+1$.
We extend $c$ to a coloring $c'$ of $G[V(H)\cup\alpha(t)]$ by coloring each one of the at most $k-a$ vertices in $\alpha(t)\setminus \delta$ with a unique color that $c$ does not use in $\delta$.
Given that no vertex in $\alpha(t)\setminus \delta$ is adjacent to a vertex in $H\setminus \delta$, this coloring is proper.

If $a+1$ colors are used in $\delta$, then it means that there is a (unique) vertex $v\in \delta\cap\beta(t)$.
In this case, at node $t$, there is at least $(k+2)-(a+1)-(k-a)=1$ unused color.
Otherwise, at most $a$ colors are used in $\delta$, and thus, there are at least $(k+2)-a-(k-a)=2$ unused colors at node $t$.
Let $(A,B)$ be a partition witnessing the bipartiteness of $G[\beta(t)]$ such that $v\in A$, if it exists.
We color the vertices of $A$ with the color of $v$, if it exists, or one of the two unused color and the vertices of $B$ with the last unused color.
Given that the vertices in $A$ (resp. $B$) are not pairwise adjacent and that the vertices in $\beta(t)\setminus\{v\}$ are not adjacent to vertices in $H\setminus \delta$, the coloring remains proper and uses at most $k+2$ colors.
 \end{proof}

The result of \autoref{approx-col} is tight since any
even cycle has bipartite treewidth zero and is 2-colorable, and any
odd cycle has bipartite treewidth one and is 3-colorable.
Unfortunately, despite of \autoref{approx-col}, the problem is \textsf{para-}\NP-complete parameterized by \oct.

\begin{lemma}\label{lem-hard-coloring}
{\sc 3-Coloring} is \NP-complete even for graphs of \oct\ at most three.
\end{lemma}

\begin{proof}
We present a reduction from the {\sc List-$3$-Coloring} problem that is defined as follows.
\begin{center}
	\fbox{
		\begin{minipage}{11.5cm}
			\noindent{\sc List-$3$-Coloring}\\
			\noindent\textbf{Input}:~~A graph $G$ and a set $L(v)$ of colors in $[3]$ for each $v\in V(G)$.\\
			\textbf{Question}:~~Is there a proper 3-coloring $c$ of $G$ such that $c(v)\in L(v)$ for each $v\in V(G)$?
		\end{minipage}
	}
\end{center}
The {\sc List-$3$-Coloring} problem is \NP-complete even when restricted to  planar 3-regular bipartite graphs~\cite{ChlebikC06}. Let $(G,L)$ be such an instance of {\sc List-$3$-Coloring}.
Let $G^+$ be the graph obtained from $G$ by adding three vertices $v_1,v_2$, and $v_3$ that are pairwise adjacent and such that for each $v\in V(G)$, $v$ is adjacent to $v_i$ for $i\in[3]\setminus L(v)$.
See \autoref{fig_3-color} for an illustration.
It is easy to see that $G^+$ admits a proper 3-coloring $c$ if and only if $(G,L)$ admits a proper list coloring $c'$, and, in this case, necessarily, $c_{|V(G)}=c'$.
Given that {\sc List-3-Coloring} is \NP-complete on bipartite graphs and that $\oct(G^+)\leq 3$, {\sc 3-Coloring} is \NP-complete even for graphs of \oct\ at most three.
\begin{figure}[h]
\center
\includegraphics[scale=1]{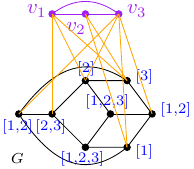}
\caption{Example of the construction of $G^+$.}
\label{fig_3-color}
\end{figure}
 \end{proof}

\subsection{Hardness of covering problems}
\label{sec-covering}

\PbCov\ is known to be \NP-complete on general graphs for every non-trivial graph class $\Gcal$ \cite{LewisY80then}.
However, for some graph classes $\Gcal$, it might change when we restricte the input graph to be bipartite.
Yannakakis \cite{Yannakakis81node} characterizes hereditary graph classes $\Gcal$ for which \PbCov\ on bipartite graphs is polynomial-time solvable and those for which \PbCov\ remains \NP-complete.

A problem $\Pi$ is said to be \emph{trivial} on a graph class $\Gcal$ if the solution to $\Pi$ is the same for every graph $G\in\Gcal$. Otherwise, $\Pi$ is called \emph{nontrivial} on $\Gcal$.
Given a graph $G$, let $\nu(G)=|\{N_G(v)\mid v\in V(G)\}|$.
Given a graph class $\Gcal$, let $\nu(\Gcal)=\sup\{\nu(G)\mid G\in\Gcal\}$.

\begin{proposition}[\hspace{1sp}\cite{Yannakakis81node}]\label{yan}
Let $\Gcal$ be a hereditary graph class such that {\sc Vertex Deletion to $\Gcal$} is nontrivial on bipartite graphs.
\begin{itemize}
\item If $\nu(\Gcal)=+\infty$, then \PbCov\ is \NP-complete on bipartite graphs.
\item If $\nu(\Gcal)<+\infty$, then \PbCov\ is polynomial time-solvable on bipartite graphs.
\end{itemize}
\end{proposition}

Hence, here is a non-exhaustive list of problems that are \NP-complete on bipartite graphs:
{\sc Vertex Deletion to $\Gcal$} where $\Gcal$ is a minor-closed graph class that contains edges (and hence
{\sc Feedback Vertex Set}, {\sc Vertex Planarization}, {\sc $H$-Minor-Cover} for $H$ containing $P_3$ as a subgraph),
{\sc $H$-Subgraph-Cover}, {\sc $H$-Induced-Subgraph-Cover}, and {\sc $H$-Odd-Minor-Cover} for $H$ bipartite graph containing $P_3$ as a (necessarily induced) subgraph, {\sc Vertex Deletion to graphs of degree at most $p$} for $p\geq 1$, {\sc Vertex Deletion to graphs of girth at least $p$} for $p\geq 6$ (note that the smallest non-trivial lower bound on the length of a cycle in a bipartite graph is six, or equivalently five).
\medskip

As a consequence of the above results, all the above problems, when parameterized by bipartite treewidth, are \textsf{para-NP}-complete.

\subsection{Hardness of packing problems}
\label{sec-packing-hardness}

Kirkpatrick and Hell~\cite[Theorem 4.2]{KirkpatrickH83onth} proved that if $H$ is a graph that contains a $P_3$ (the path on three vertices) as a subgraph, then the problem of partitioning the vertex set of an input graph $G$ into subgraphs isomorphic to $H$ is \NP-complete. Let us call this problem {\sc $H$-Partition}. This immediately implies that the {\sc $H$-Subgraph-Packing} problem is \NP-complete if $H$ contains a $P_3$.
In the next lemma we observe that the reduction of Kirkpatrick and Hell~\cite{KirkpatrickH83onth}  can be carefully analyzed so that the same result also applies to bipartite input graphs, and to the induced packing version as well.

\begin{lemma}\label{pack-bip}
Let $H$ be a bipartite graph containing $P_3$ as a subgraph. Then {\sc $H$-partition} is \NP-complete on bipartite graphs.
As a consequence, {\sc $H$-Subgraph-Packing} and {\sc $H$-Induced-Subgraph-Packing} are \NP-complete on bipartite graphs.
\end{lemma}
\begin{proof}
 We proceed to discuss the proof for the {\sc $H$-Subgraph-Packing} problem, and finally we observe that the same proof applies to the induced subgraph version. The reduction of Kirkpatrick and Hell~\cite[Lemma 4.1]{KirkpatrickH83onth} is from the \textsc{$h$-Dimensional Matching} problem, where $h=|V(H)|$. In this problem, we are given a set of $h$-tuples ${\mathcal T} \subseteq [p]^h$, where $p$ is any positive integer, and the goal is to decide whether there exists a subset ${\mathcal S} \subseteq {\mathcal T}$ with $|{\mathcal S}|=p$ such that no two elements in $S$ agree in any coordinate. This problem is well-known to be \NP-complete for any $h\geq 3$~\cite{Karp10redu}, which is guaranteed by the hypothesis of the lemma. We need to introduce some notation. Given a vertex $v \in H$, we denote by $H_v$ the graph obtained from $H$ by adding a new vertex $v'$ adjacent to $N_H(v)$ (that is, $v'$ becomes a false twin of $v$). Vertices $v,v'$ are called the \emph{connector vertices of $H_v$} and the newly introduced vertex $v'$ is called the \emph{twin vertex} of $H_v$. Given a vertex $v \in H$, we denote by $H\langle v\rangle$ the graph obtained from $H$ by adding, for every vertex $u \in V(H)$, a distinct copy of $H_v$ and identifying vertex $v$ of $H_v$ with $u$. The $h$ twin vertices of the copies of $H_v$,  which were not identified with any vertex, are called the \emph{connector vertices of $H\langle v\rangle$}. In this construction, we call the initial copy of $H$ the \emph{base copy} of $H$ in $H\langle v\rangle$.

   The reduction of Kirkpatrick and Hell~\cite[Lemma 4.1]{KirkpatrickH83onth} proceeds as follows. Let $v^\star$ be any vertex of $H$ that is not a cut-vertex and belongs to a biconnected component of $H$ containing at most one cut-vertex, which is always guaranteed to exist. Given an instance ${\mathcal T} \subseteq [p]^h$ of \textsc{$h$-Dimensional Matching}, we construct a graph $G$ as follows. We start with an independent set of size $hp$ whose elements are labeled with pairs $(i,j)$ with $i \in [h]$ and $j \in [p]$. For each $h$-tuple $(t_1, \ldots, t_h) \in {\mathcal T}$ we introduce a distinct copy of $H\langle v^\star\rangle$ and we identify its $h$ connector vertices arbitrarily with the $h$ vertices labeled $(1, t_1), \ldots, (h,t_h)$. It is proved in~\cite{KirkpatrickH83onth} that $V(G)$ can be partitioned into copies of $H$ if and only if ${\mathcal T}$ is a \yes-instance of \textsc{$h$-Dimensional Matching}.

  All we need to show is that, if $H$ is bipartite, then the constructed graph $G$ is bipartite as well. For this, we proceed to define a bipartition function $b:V(G) \to \{0,1\}$ such that for every edge $uv \in E(G)$ it holds that $b(u) \neq b(v)$. Since $H$ is bipartite, such a function exists for $H$; we denote it by $b_H$. We start by defining $b$ restricted to each of the copies $H\langle v^\star\rangle$ introduced in the construction (all copies are labeled equally). Recall that each such copy contains a copy of $H_{v^\star}$ for every vertex $v \in V(H)$. We proceed to define $b$ so that the twin vertex of every copy of $H_{v^\star}$ lies on the same side as the vertex of $H$ to which this copy has been attached. Formally, for every vertex $u$ in the base copy of $H$ in $H\langle v^\star\rangle$, we define $b(u)=b_H(u)$. Consider any vertex $u$ in the base copy of $H$, and let $v$ be any vertex in the copy of $H_{v^\star}$ that has been attached to $u$, different from the twin vertex of that copy. If $b_H(u)=b_H(v^\star)$, we let $b(v)=b_H(v)$, otherwise we let $b(v)=1-b_H(v)$. That is, in the latter case, when $u$ and $v^\star$ do {\sl not} agree in $H$, we swap the bipartition of that copy, except for its twin vertex. Finally, let $w$ be the twin vertex of the copy of $H_{v^\star}$  attached to $u$. We define $b(w) = b(u)$, and note that this is always possible since the two connector vertices of each copy of $H_{v^\star}$ are false twins. It can be easily verified that, for each copy of $H\langle v^\star\rangle$, the defined function $b$ induces a proper bipartition.  To conclude the proof, we need to guarantee that, when identifying the connector vertices of distinct copies of $H\langle v^\star\rangle$ according to the $h$-tuples in ${\mathcal T}$, we do not identify two vertices $u$ and $v$ with $b(u) \neq b(v)$. For this, we use the following trick. Let us consider any fixed ordering of $V(H)$. In the construction of $G$ described above,  for each $h$-tuple $(t_1, \ldots, t_h) \in {\mathcal T}$ and its associated copy of $H\langle v^\star\rangle$, instead of identifying its $h$ connector vertices {\sl arbitrarily} with the $h$ vertices in the independent set labeled $(1, t_1), \ldots, (h,t_h)$, we do it as follows. Each connector vertex $w$ in $H\langle v^\star\rangle$ is naturally associated with a vertex $u$ of the base copy of $H$ to which its copy of $H_{v^\star}$ has been attached. Suppose that vertex $u$ is the $i$-th vertex in the considered ordering of $V(H)$. We then identify the connector vertex $w$ with the vertex labeled $(i,t_i)$ in the independent set. This way, there is no conflict between the function $b$ defined for different copies of $H\langle v^\star\rangle$, and it indeed holds that $b(u) \neq b(v)$ whenever $uv \in E(G)$, concluding the proof.

  Finally, to prove the same result for {\sc $H$-Induced-Subgraph-Packing}, we use exactly the same construction as above, and it suffices to observe that, in a \yes-instance, each of the copies of $H$ in $G$ is induced. Indeed, in the proof of~\cite[Lemma 4.1]{KirkpatrickH83onth}, the idea of the construction is the following. If a tuple $(t_1, \ldots, t_h) \in {\mathcal T}$ is taken into the solution, then in its associated copy of $H\langle v^\star\rangle$, the copies of $H$ chosen for the packing are, on the one hand, the base copy of $H$ and, on the other hand, the copy of $H$ in each $H_{v^\star}$ containing the connector vertex. For a tuple $(t_1, \ldots, t_h) \in {\mathcal T}$ that is {\sl not} taken into the solution,  in its associated copy of $H\langle v^\star\rangle$, the copies of $H$ chosen for the packing are just the copy of $H$ in each $H_{v^\star}$ {\sl not} containing the connector vertex. In both cases, each of the chosen copies of $H$ is an induced subgraph of $G$, and the lemma follows.
 \end{proof}

\autoref{pack-bip} is tight for connected graphs $H$.
Indeed, if $H$ is connected and $P_3$ is not a subgraph of $H$, then $H$ is either a vertex or an edge. In the former case, the problem is trivial, and the latter is {\sc Maximum Matching}, which is polynomially-time-solvable on general graphs \cite{MicaliV80ano(}.
\medskip

In the next lemma, we derive from \autoref{pack-bip} that {\sc $H$-Minor-Packing} is also \textsf{para-NP}-complete parameterized by $\btw$ when $H$ contains a $P_3$.

\begin{lemma}\label{pack-bip-min}
Let $H$ be a bipartite graph containing $P_3$ as a subgraph.
Then {\sc $H$-Minor-Packing} is \NP-complete on bipartite graphs.
\end{lemma}

\begin{proof}
We reduce from {\sc $H$-Partition} restricted to bipartite graphs, which is \NP-complete by \autoref{pack-bip}.
Let $G$ be a bipartite graph as an instance of {\sc $H$-partition}.
Then $G$ has a minor-packing of size at least $|V(G)|/|V(H)|$ if and only if $V(G)$ can be partitioned into copies of $H$, and the lemma follows.
\end{proof}
Note that a similar result to \autoref{pack-bip-min} would hold for {\sc $H$-Induced-Minor-Packing}.

\medskip
As a corollary of \autoref{pack-bip-min}, {\sc $H$-Odd-Minor-Packing} is \textsf{para-NP}-complete parameterized by $\btw$ when $H$ is bipartite and contains $P_3$ as a subgraph.

\begin{lemma}\label{pack-bip-odd}
Let $H$ be a bipartite graph containing $P_3$ as a subgraph.
Then {\sc $H$-Odd-Minor-Packing} is \NP-complete on bipartite graphs.
\end{lemma}

\begin{proof}
Given that odd-minors preserve cycle parity (\autoref{odd-eq}), when $H$ is bipartite,
{\sc $H$-Odd-Minor-Packing} and {\sc $H$-Minor-Packing} are the same problem on bipartite graphs.
 \end{proof}

\bigskip
As stated in~\autoref{mis-btw}, {\sc (Weighted) Independent Set} is solvable in \FPT-time when parameterized by \btw.
Can we go beyond?
Unfortunately, {\sc $d$-Scattered  Set}, that is, the problem of asking for a  set of vertices of size at least $k$ that are pairwise within distance at least $d$,
is \NP-complete for $d\geq 3$ even for bipartite planar graphs of maximum degree three \cite{EtoGM14dist} (when $d=2$, this corresponds to {\sc Independent Set}), so we cannot hope to go further than
 {\sc $d$-Scattered  Set}, parameterized by $\btw$, is \textsf{para-NP}-complete, when $d≥3$.

Similarly, {\sc Induced Matching}, that is, the problem of finding an induced matching with $k$ edges, is \NP-complete on bipartite graphs~\cite{Cameron89indu}.
In the next lemma, we reduce from   {\sc Induced Matching} to prove that {\sc $H$-Scattered-Packing} is \textsf{para-NP}-complete parameterized by $\btw$ when $H$ is 2-connected and bipartite.

\begin{lemma}\label{pack-bip-scat}
Let $H$ be a 2-connected bipartite graph with at least one edge.
Then {\sc $H$-Scattered-Packing} is \NP-complete on bipartite graphs.
\end{lemma}

\begin{proof}
We reduce from {\sc Induced Matching} on bipartite graphs, which is \NP-complete~\cite{Cameron89indu}, as follows. Let $uv\in E(H)$ and let $H'=H\setminus\{u,v\}$.
Let $G$ be a bipartite graph as an instance of {\sc Induced Matching}.
We build $G'$ as an instance of {\sc $H$-Scattered-Packing} as follows.
Let $e_1,\ldots,e_m$ be the edges of $G$.
$G'$ is obtained from the disjoint union of $G$ and $m$ graphs $H_1,\ldots,H_m$ isomorphic to $H'$ by adding the appropriate edges between $e_i$ and $H_i$ for $i\in[m]$ to create a graph isomorphic to $H$.
Note that $G'$ is bipartite.

Given an induced matching $\{e_i\mid i\in I\subseteq [m]\}$ in $G$, $\{G'[e_i \cup V(H_i)\mid i\in I\}$ is an $H$-scattered packing of the same size in $G'$.
Conversely, given an $H$-scattered packing in $G'$, given that $|V(H')|=|V(H)|-2$ and that $H$ is 2-connected, any occurrence of $H$ in the packing intersects at least one edge of $G$.
Hence, this gives rise to an induced matching of the same size in $G$.
Thus, for any $k\in\bN$, $(G,k)$ is a \yes-instance for {\sc Induced Matching} if and only if $(G',k)$ is a \yes-instance for {\sc $H$-Scattered-Packing}.
 \end{proof}

For $q \geq 2$,
it turns out that even asking that the graph $H$ to be packed or covered is non-bipartite is {\sl not} enough to make the problem tractable. As an illustration of this phenomenon, in the next lemma we prove that {\sc $H$-Scattered-Packing}
is {\sf para-NP}-complete parameterized by $q$(-torso)-$\Bcal$-treewidth for $q\geq 2$ even if $H$ is {\sl not} bipartite (\autoref{lem-paraNP-non-bip}).

\begin{lemma}\label{lem-paraNP-non-bip}
Let $H$ be a 2-connected graph containing an edge, and let $q\in\bN_{\geq 2}$.
Then {\sc $H$-Scattered-Packing} is {\sf para-NP}-complete parameterized by $(q,\Bcal)^{(*)}$-$\tw$.
\end{lemma}

\begin{proof}
{\sc Induced Matching} is \NP-complete on bipartite graphs~\cite{Cameron89indu}.
Let $H'=H\setminus\{u,v\}$, where $uv\in E(H)$.
We reduce from {\sc Induced Matching} as follows.
Let $G$ be a bipartite graph as an instance of {\sc Induced Matching}.
We build $G'$ as in the proof of \autoref{pack-bip-scat}, so that, for any $k\in\bN$, $(G,k)$ is a \yes-instance for {\sc Induced Matching} if and only if $(G',k)$ is a \yes-instance for {\sc $H$-Subgraph-Cover}.

Let us show that $(q,\Bcal)^{(*)}$-$\tw(G')\leq(q,\Bcal)^{(*)}$-$\tw(H)$.
Let $\Tcal=(T,\alpha,\beta)$ be a $q$(-torso)-$\Bcal$-tree decomposition of $H$.
Since $uv\in E(H)$, there is $t_0\in V(T)$ such that $u,v\in(\alpha\cup\beta)(t_0)$.
Then we build a $q$(-torso)-$\Bcal$-tree decomposition of $G'$ as follows.
Let $T'$ be the tree obtained by taking $|E(G)|$ copies of $T$ and making the node $t_0$ of each such copy adjacent to a new node $r$.
We set $\beta'(r)=V(G')$, $\alpha'(r)=\emptyset$, and $\alpha'$ and $\beta'$ take the same values as $\alpha$ and $\beta$, respectively, for the other nodes of $T'$.
There are at most two vertices in the adhesion of $r$ and any other node, and they are in $\beta(r)$.
Moreover, $\torso_G(V(G'))=G'\in\Bcal$.
Hence, $(T',\alpha',\beta')$ is a $q$(-torso)-$\Bcal$-tree decomposition of $G'$ of width at most $(q,\Bcal)^{(*)}$-$\tw(H)$.
 \end{proof}

\section{Further research}
\label{sec-conclusions}

 In this paper we study the complexity of several problems parameterized by bipartite treewidth, denoted by \btw. In particular, our results
 extend the graph classes for which {\sc Vertex Cover/\hspace{1sp}Independent Set}, {\sc Maximum Weighted Cut}, and \textsc{Odd Cycle Transversal} are polynomial-time solvable. A number of interesting questions remain open.

Except for \textsc{3-Coloring}, all the problems we consider are covering and packing problems. We are still far from a full classification of the variants that are {\sf para-NP}-complete, and those that are not (\FPT or \XP). For instance, concerning {\sc $H$-Subgraph-Cover}, we provided \FPT-algorithms when $H$ is a clique (\autoref{cor-Kt-cover}). This case is particularly well-behaved because we know that in a tree decomposition every clique appears in a bag. On the other hand, as an immediate consequence of the result of Yannakakis~\cite{Yannakakis81node} (\autoref{yan}), we know that {\sc $H$-Subgraph-Cover} is {\sf para-NP}-complete for every bipartite graph $H$ containing $P_3$ (cf.~\autoref{sec-covering}). We do not know what happens when $H$ is non-bipartite and is not a clique. An apparently simple but challenging case is  {\sc $C_5$-Subgraph-Cover} (or any other larger odd cycle).
The main difficulty seems to be that {\sc $C_5$-Subgraph-Cover} does not have the gluing property, which is the main ingredient in this paper to show that a problem is nice, and therefore to obtain an \FPT-algorithm.
We do not exclude the possibility that the problem is {\sf para-NP}-complete, as we were not able to obtain even an \XP\ algorithm.

Concerning the packing problems, {\sc $H$-(Induced-)Subgraph/Scattered/Odd-Minor-Packing}, we provide \XP-algorithms for them in \autoref{xp} when $H$ is non-bipartite. Unfortunately, we do not know whether any of them admit an \FPT-algorithm, although we suspect that it is indeed the case. We would like to mention that  it is possible to apply the framework of equivalence relations and representatives (see for instance~\cite{GarneroPST15,GarneroPST19,BasteST20}) to obtain an \FPT-algorithm for {\sc $K_t$-Subgraph-Packing} parameterized by \btw.
However, since a number of definitions and technical details are required to present this algorithm, we decided not to include it in this paper. However, when $H$ is not a clique, we do not know whether {\sc $H$-Subgraph-Packing} admits an \FPT-algorithm. A concrete case that we do not know how to solve is when
 $H$ is the \emph{paw}, i.e., the 4-vertex graph consisting of one triangle and one pendent edge.

Beyond bipartite tree decompositions, we introduce a more general  type of decompositions that we call $q$(-torso)-$\Hcal$-tree decompositions. For $\Bcal$ being the class of bipartite graphs, we prove in
\autoref{lem-paraNP-non-bip} that for every $q \geq 2$ and every
 2-connected graph $H$ with an edge,  {\sc $H$-Scattered-Packing} is {\sf para-NP}-complete parameterized by $q$(-torso)-$\Bcal$-treewidth. It should be possible to prove similar results for other covering and packing problems considered in this article.

Most of our {\sf para-NP}-completeness results consist in proving \NP-completeness on bipartite graph (i.e., those with bipartite treewidth zero). There are two exceptions. On the one hand, the {\sf NP}-completeness of \textsc{$3$-Coloring} on graphs with odd cycle transversal at most three (\autoref{lem-hard-coloring}) and \autoref{lem-paraNP-non-bip} mentioned above for {\sc $H$-Scattered-Packing} parameterized by $q$-$\Bcal$-treewidth for every integer $q \geq 2$. It is worth noting that none of our hardness results really exploit the structure of  bipartite tree decompositions (i.e., for $q=1$), beyond being bipartite or having bounded odd cycle transversal.
Hence, an interesting question would be whether there is a problem that is \NP-complete on graph of bounded bipartite treewidth, but not on graphs of bounded odd cycle transversal.

Finally, as mentioned in the introduction, the goal of this article is to make a first step toward efficient algorithms to solve problems related to odd-minors. We already show in this paper that bipartite treewidth can be useful in this direction, by providing an \XP-algorithm for {\sc $H$-Odd-Minor-Packing}. Bipartite treewidth, or strongly related notions, also plays a strong role in the recent series of papers about odd-minors by
Campbell, Gollin, Hendrey, and Wiederrecht~\cite{GollinW23odd-,Campbell23odd-}. This looks like an emerging topic that is worth investigating.

\subparagraph{Acknowledgments.} We thank Sebastian Wiederrecht for his many helpful remarks, and the referees for their thorough reading of the manuscript.

\bibliographystyle{elsarticle-num-names}
\bibliography{bib_odd}

\end{document}